\documentclass[11pt,letterpaper]{article}
\usepackage{latexsym,color,amssymb,amsthm,fullpage,enumitem,
amsmath,amsfonts,bm}

\usepackage{graphicx}
\usepackage{euscript}

\DeclareFontFamily{OT1}{pzc}{}
\DeclareFontShape{OT1}{pzc}{m}{it}{<-> s * [1.2500] pzcmi7t}{}
\DeclareMathAlphabet{\mathscr}{OT1}{pzc}{m}{it}

\usepackage{graphicx}

\usepackage[left=1in, right=1in, top=1in, bottom=1in]{geometry}

\newcommand{\ignore}[1]{}

\newtheorem{theorem}{Theorem}[section]
\newtheorem{lemma}[theorem]{Lemma}

\newtheorem{proposition}[theorem]{Proposition}
\newtheorem{observation}[theorem]{Observation}

\newtheorem{corollary}[theorem]{Corollary}

\def\part{{\it sp}}
\def\strong{{strong}}
\def\eps{{\epsilon}}
\def\aff{{{\mathtt{aff}}}}
\def\conv{{{\mathtt{conv}}}}

\def\1{{(1)}}
\def\2{{(2)}}

\def\A{{\cal A}}
\def\Z{{\cal Z}}

\def\comp{{\sf{\varpi}}}

\def\I{{\cal I}}
\def\r{{s}}

\def\s{{t}}
\def\u{{u}}

\def \nar {{nar}}
\def\F{{\cal F}}
\def\L{{\cal {L}}}
\def\HH{{\cal{H}}}
\def\C{{\cal C}}

\def\D{{\cal D}}
\def\B{{\cal B}}

\def\E{{{\cal{E}}}}

\def\K{{{\mathcal {K}}}}

\def\N{{\cal N}}
\def\W{{\cal W}}
\def\R{{\cal{R}}}

\def\SS{{\mathbb S}}
\def\reals{{\mathbb R}}

\def\P{{\cal P}}

\def\V{{\cal V}}

\def\p{{\bm{p}}}

\def\T{{{\cal{T}}}}

\begin{document}

\begin{titlepage}

\title{Stronger Bounds for Weak Epsilon-Nets in Higher Dimensions\thanks{A preliminary version of the present study has appeared in the Proceedings of the 53rd Annual ACM Symposium on Theory of Computing, 2021.}}

\author{
Natan Rubin\thanks{Email: {\tt rubinnat.ac@gmail.com}. Ben Gurion University of the Negev, Beer-Sheba, Israel.  
Supported
by grant 2891/21 from Israel Science Foundation. The project leading to this application has received funding from European Research Council (ERC)
under the European Unions Horizon 2020 research and innovation programme under grant agreement No. 678765.
} }

 \maketitle

\begin{abstract}
Given a finite point set $P$ in $\reals^d$, and $\eps>0$ we say that $N\subseteq \reals^d$ is a weak $\eps$-net if it pierces every convex set $K$ with $|K\cap P|\geq \eps |P|$.
We show that for any finite point set in dimension $d\geq 3$, and any $\eps>0$, one can construct a weak $\eps$-net whose cardinality is $\displaystyle O^*\left(\frac{1}{\eps^{2.558}}\right)$ in dimension $d=3$, and $\displaystyle o\left(\frac{1}{\eps^{d-1/2}}\right)$ in all dimensions $d\geq 4$.\footnote{In the sequel, the $O^*(x)$ notation hides asymptotic factors that are bounded by $x^{\gamma}$, for arbitrary small $\gamma>0$, whereas the $\tilde{O}(x)$-notation hides multiplicative factors that are polylogarithmic in $x$.}

To be precise, our weak $\eps$-net has cardinality $\displaystyle O\left(\frac{1}{\eps^{\alpha_d+\gamma}}\right)$ for any $\gamma>0$, with 

$$
    \alpha_d=
    \begin{cases}
      $2.558$ & \text{if} \ d=3\\
      $3.48$ & \text{if} \ d=4\\
      \left(d+\sqrt{d^2-2d}\right)/2 & \text{if} \ d\geq 5.
    \end{cases}
$$

This is the first significant improvement of the bound of $\displaystyle \tilde{O}\left(\frac{1}{\eps^d}\right)$ that was obtained in 1993 by Chazelle, Edelsbrunner, Grigni, Guibas, Sharir, and Welzl for general point sets in dimension $d\geq 3$.
\end{abstract}

\maketitle


\end{titlepage}

\section{Introduction} \label{sec:intro}
\subsection{Background}
\paragraph{Transversals and $\eps$-nets.} Given a family $\K$ of geometric ranges in $\reals^d$ (e.g., lines, triangles, or convex sets),
we say that $N\subset \reals^d$ is a transversal to $\K$ (or $N$ pierces $\K$) if each range $K\in \K$ is pierced by at least one point of $N$.
Given an underlying set $P$ of $n$ points, we say that a range $K\in \K$ is {\it $\eps$-heavy} if $|P\cap K|\geq \eps n$.
We say that $N$ is an {\it $\eps$-net} for $\K$ if it pierces every $\eps$-heavy range in $\K$.
We say that such a set $N$ is a {\it strong $\eps$-net} for $\K$ if $N\subset P$, that is, the points of the net are drawn from the underlying point set $P$. Otherwise (i.e., if $N$ includes additional points outside $P$), we say that $N$ is a {\it weak \it $\eps$-net}.

The study of $\eps$-nets was initiated by Vapnik and Chervonenkis \cite{VC1971}, in the context of Statistical Learning Theory.
Following a seminal paper of Haussler and Welzl \cite{HW87}, $\eps$-nets play a central role in  Discrete and Computational Geometry \cite{Trends}.
For example, bounds on $\eps$-nets determine the performance of the best-known algorithms for Minimum Hitting Set/Set Cover Problem in geometric hypergraphs \cite{NetsRectangles,BroGood,ClaVar,Even}, and the transversal numbers of families of convex sets \cite{AlonKalai,AKMM,AlonKleitman,Shakhar}.

Informally, the cardinality of the smallest possible $\eps$-net for the range set $\K$ determines the integrality gap of the corresponding transversal problem -- the ratio between (1) the size of the smallest possible transversal $N$ to $\K$ and (2) the weight of the ``lightest" possible fractional transversal to $\K$ \cite{AlonKleitman,AlonKalai,Even}.

Haussler and Welzl \cite{HW87} proved in 1986 the existence of strong $\eps$-nets of cardinality $O\left(\frac{1}{\eps}\log\frac{1}{\eps}\right)$ for families of simply-shaped, or semi-algebraic geometric ranges in $d$-space, for a fixed $d>0$ (e.g., lines, boxes, spheres, halfspaces, or simplices), by observing that their induced hypergraphs have a bounded Vapnik-Chervonenkis dimension (so called {\it VC-dimension}).\footnote{The constant hidden within the $O(\cdot)$-notation is specific to the family of geometric ranges under consideration, and is proportional to the VC-dimension of the induced hypergraph.}
While the bound is generally tight for set systems with a bounded VC-dimension \cite{KPW90}, tight $o\left(\frac{1}{\eps}\log \frac{1}{\eps}\right)$-size constructions were discovered for discs in $\reals^2$, halfplanes in $\reals^2$ and halfspaces in $\reals^3$ \cite{ClaVar,KPW90,MSW90}, and rectangles in $\reals^2$ and boxes in $\reals^3$ \cite{NetsRectangles,PachTardos}. We refer the reader to a recent state-of-the-art survey \cite{HandbookNets}.

It had long been conjectured that all the ``natural" geometric instances, that involve simply-shaped geometric ranges in a fixed-dimensional Euclidean space $\reals^d$, admit a strong $\eps$-net 
of cardinality $O(1/\eps)$. 
The conjecture was refuted, for the particular case of line ranges, in 2010 by Alon \cite{Alon}.
Pach and Tardos \cite{PachTardos} subsequently demonstrated that the multiplicative term $\Theta\left(\log 1/\eps\right)$ is necessary for halfspaces in dimension higher than $3$.

\paragraph{\bf Weak $\eps$-nets for convex sets.} In sharp contrast to the case of simply-shaped ranges, no constructions of small-size strong $\eps$-nets exist for general families of convex sets in $\reals^d$, for $d\geq 2$.  For example, given an underlying set of $n$ points in convex position in $\reals^2$, any strong $\eps$-net with respect to convex ranges must include at least $n-\eps n$ of the points. Informally, this  phenomenon can be attributed to the fact that 
the VC-dimension of a geometric set system is closely related to the {\it description complexity} of the underlying ranges, and it is unbounded for general convex sets.
Nevertheless, B\'{a}r\'{a}ny, F\"{u}redi and Lov\'{a}sz \cite{BFL} observed in 1990 that families of convex sets in $\reals^2$ still admit weak $\eps$-nets of cardinality $O(\eps^{-1026})$.
Alon, B\'{a}r\'{a}ny, F\"{u}redi, and Kleitman \cite{AlonSelections} were the first to show in 1992 that families of convex sets in any dimension $d\geq 1$ admit weak $\eps$-nets whose cardinality is bounded in terms of $1/\eps$ and $d$. The subsequent study and application of weak $\eps$-nets bear strong relations to convex geometry, including Helly-type, Tverberg-type, and Selection Theorems; see \cite[Sections 8 -- 10]{JirkaBook} for a comprehensive introduction.

\medskip
\noindent{\bf Weak $\eps$-nets and the Hadwiger-Debrunner Problem.} Alon and Kleitman \cite{AlonKleitman} used the boundedness of weak $\eps$-nets to confirm a long-standing {\it $(p,q)$-conjecture} by Hadwiger and Debrunner \cite{HD}. To this end, we say that a family $\K$ of convex sets satisfies the {\it $(p,q)$-property}  if any its $p$-size subfamily $\K'\subset \K$ contains a $q$-size subset $\K''\subset \K$ with a non-empty common intersection $\bigcap \K''\neq \emptyset$. 
Hadwiger and Debrunner conjectured that for every positive integers $p,q$ and $d$ that satisfy $p\geq q\geq d+1$, there exists an integer $C_d(p,q)<\infty$ such that the following statement holds: {\it Any family $\K$ of convex sets in $\reals^d$ with the $(p,q)$-property admits a transversal by at most $C_d(p,q)$ points.} Note that the celebrated Helly Theorem yields a transversal by a {\it single} point in the case $p=q=1$ whenever $|\K|\geq d+1$.

Showing good quantitative estimates for the Hadwiger-Debrunner numbers $C_d(p,q)$ is a formidable open problem which requires tight asymptotic bounds for weak $\eps$-nets; see the latest study by Keller, Smorodinsky and Tardos \cite{Shakhar}, and the concluding discussion in Section \ref{Sec:Final}.
Very recently, lower bounds for several of the above questions -- including strong and weak $\eps$-nets with respect to line ranges \cite{BS18}, and the 2-dimensional Hadwiger-Debrunner numbers $C_2(p,q)$ \cite{KS18} -- were improved using the novel combinatorial machinery of hypergraph containers \cite{BMS,SaxTh}.

\medskip
\noindent{\bf Weak $\eps$-nets, Radon numbers, and the fractional Helly theorem.} Alon, Kalai, Matou\v{s}ek, and Meshulam \cite{AKMM} studied the above problems in a more general setting of abstract hypergraphs that are closed under intersections. They showed that  the existence of weak $\eps$-net (whose size is bounded in $1/\eps$) can be combinatorially deduced from the so called fractional Helly property. More recently,  Holmsen and Lee \cite{Holmsen, HolmsenLee} showed that the latter property is a purely combinatorial consequence of the bounded Radon number, which is equal to $d+2$ for convex sets in $\reals^d$; also see a related study \cite{MoYe}.  

\paragraph{Point-selection theorems.} The aforementioned first general construction of weak $\eps$-nets in the Euclidean spaces $\reals^d$ of dimension $d\geq 2$, due to Alon {\it et al.} \cite{AlonSelections}, used the following property of simplicial hypergraphs: for any $n$ point set $P$ in general position in $\reals^d$ (with $n\geq d+1$), and any dense collection $E$ of $d$-dimensional simplices over $P$, which are identified with a subset of ${P\choose d+1}$, there is a point $x\in \reals^d$ piercing a significant fraction of the simplices of $E$. (This fraction is polynomial in the edge density $|E|/{n\choose d+1}$. The asymptotic guarantee was dramatically improved by the author \cite{SelectionsSODA}, partly using the methods that are developed in this paper.)
This deep result, which has since become known as {\it The Second Selection Theorem}, underlies a number of other fundamental results in discrete and computational geometry, such as the best-known estimates for $k$-sets and $k$-levels in dimension $d\geq 5$ \cite{BFL}. Its special case dealing with $E={P\choose d+1}$, is known as {\it The First Selection Theorem}.

\paragraph{Bounds on weak $\eps$-nets.} For any $\eps>0$ and $d\geq 0$, let $f_d(\eps)$ be the smallest number $f>0$ so that, for any underlying finite point set $P$, one can pierce all the $\eps$-heavy convex sets
using only $f$ points in $\reals^d$. 
It is
an outstanding open problem in Discrete and Computational geometry to determine the true asymptotic behaviour of $f_d(\eps)$ in dimensions $d\geq 2$. As Alon, Kalai, Matou\v{s}ek, and Meshulam noted in 2001: ``{\it Finding the correct estimates for weak $\eps$-nets is, in our opinion, one of the truly important open problems in combinatorial geometry"} \cite{AKMM}.

\medskip
Alon, B\'ar\'any, F\"{u}redi, and Kleitman \cite{AlonSelections} (see also \cite{AlonKleitman}) showed that $f_d(\eps)=O\left(1/\eps^{(d+1)(1-1/\beta_d)}\right)$, where $0<\beta_d<1$ denotes the so called {\it point selection exponent}, to be defined in the sequel. Using a different argument, they showed that $f_2(\eps)=O\left(1/\eps^2\right)$.\footnote{An outline of the planar $f_2(\eps)=O\left(1/\eps^2\right)$ bound can be found in a popular textbook by Chazelle \cite{ChazelleBook}.}
The bound in higher dimensions $d\geq 3$ has been subsequently improved in 1993 by Chazelle {\it et al.}  \cite{Chazelle} to roughly $\displaystyle \tilde{O}\left(1/\eps^d\right)$. Though the latter construction was somewhat simplified in 2004 by Matou\v{s}ek and Wagner \cite{MatWag04} using simplicial partitions with low hyperplane-crossing number \cite{PartitionTrees}, no study in the subsequent 25 years came close to tackling the notorious ``$\displaystyle \frac{1}{\eps^d}$-barrier" for general families of convex sets and arbitrary finite point sets in the Euclidean spaces $\reals^d$ of (fixed) dimension $d\geq 2$. 

In view of the best known lower bound of $\Omega\left(\frac{1}{\eps} \log^{d-1}\left(\frac{1}{\eps}\right)\right)$ for $f_d(\eps)$ due to Bukh, Matou\v{s}ek and Nivasch \cite{Staircase},  it still remains to settle whether the asymptotic behaviour of this quantity substantially deviates from the long-known ``almost-$(1/\eps)$" bounds on strong $\eps$-nets (e.g., for lines and triangles in $\reals^2$ or simplices in $\reals^d$)? 
The only interesting instances in which the gap has been nearly closed, involve essentially lower-dimensional distributions of points \cite{Chazelle,Sphere,AlonChains}.
For example, Alon, Kaplan, Nivasch, Sharir, and Smorodinsky \cite{AlonChains} showed in 2008 that any finite point set in {\it a convex position in $\reals^2$} allows for a weak $\eps$-net of cardinality $\displaystyle O\left(g(\eps)/\eps\right)$ with respect to convex sets, where $g(\cdot)$ denotes the very slowly inverse Ackerman function.

A previous study by the author \cite{FOCS18} made the first step to breaching the infamous $\displaystyle \frac{1}{\eps^d}$-barrier by showing that $\displaystyle f_2(\eps)=O\left(\frac{1}{\eps^{3/2+\gamma}}\right)$, for any constant $\gamma>0$. Unfortunately, its (partly) ad-hoc machinery did not directly apply to dimensions $d\geq 3$. 

\subsection{Our result} 

The present study offers the first $o\left(1/\eps^d\right)$-size construction in the Euclidean spaces $\reals^d$ of dimension $d\geq 3$.\footnote{The preliminary version of this paper has claimed a slightly weaker bound in dimension $d\geq 4$, and misstated the bound for $f_3(\eps)$.}

\begin{theorem}\label{Thm:Main} Let $d\geq 3$, and
\begin{equation}\label{Eq:Main}
    \alpha_d:=
    \begin{cases}
      $2.558$ & \text{if} \ d=3\\
      $3.48$ & \text{if} \ d=4\\
      \left(d+\sqrt{d^2-2d}\right)/2 & \text{if} \ d\geq 5.
    \end{cases}
\end{equation}

\noindent Then we have that
$\displaystyle f_d(\eps)=O\left(\displaystyle\frac{1}{\eps^{\alpha_d+\gamma}}\right)$ for any $\gamma>0$. That is, for any underlying set $P$ of $n$ points in $\reals^d$, and any $\eps>0$, one can construct a weak $\eps$-net with respect to convex sets, whose cardinality is $O\left(\displaystyle\frac{1}{\eps^{\alpha_d+\gamma}}\right)
$, for any $\gamma>0$.\footnote{As a rule, the constants proportionality in our $O(\cdot)$ and $\Omega(\cdot)$ notation heavily depend on the dimension $d$, and the choice of $\gamma>0$.} 

In particular, we have that $f_3(\eps)=O^*\left(1/\eps^{2.558}\right)$, and $f_d(\eps)=o\left(1/\eps^{d-1/2}\right)$ for any $d\geq 4$.
\end{theorem}

\noindent{\bf Proof overview.} Our proof of Theorem \ref{Thm:Main} is fully constructive. 
Similar to a plethora of previous works \cite{Chazelle, MatWag04,FOCS18}, it uses standard divide-and-conquer methods (e.g., Matou\v{s}ek's simplicial partitions \cite{PartitionTrees}) to derive an efficient recurrence in $\eps$, and the ground point set $P$.   

However, at the heart of our improved estimate lies a rather basic reduction of the weak $\eps$-net problem in $\reals^d$ to a sequence of non-recursive {\it $1$-dimensional} strong $\nu$-nets, with $\eps^d\ll  \nu\leq 1$, which are restricted to a handful of vertical lines within $\reals^d$. 
To this end, we establish a novel {\it line-selection} result, which yields a small set of vertical lines that would pierce many $(d-1)$-dimensional simplices induced by an {\it unknown} set of $\lceil \eps n\rceil$ points.

As is shown in the sequel, any 1-dimensional reduction of this sort must leave out some sub-family of ``narrow" sets $K$ whose simplices are hard to intercept by few vertical lines. Fortunately, such sets $K$ can be passed on to subsidiary recursive instances, which involve smaller ground sets in $\reals^d$, and larger fractions $\eps$. 

More importantly, though, a naive implementation of this scheme would yield $\nu\approx \eps^d$, thus leading to 1-dimensional $\nu$-nets of overall cardinality at least $\Omega(1/\nu)=\Omega\left(1/\eps^d\right)$.
To attain a better ratio $\nu$, we first construct a small auxiliary net $\tilde{N}\subset \reals^d$, of roughly $O\left(1/\eps^{\alpha_d}\right)$ points, with the following property: every $\eps$-heavy convex set $K$ that is missed by $\tilde{N}$, must be {\it flat} in a certain combinatorial sense. 
(The measure of this relative flatness of $K$ will be determined using such spatial decomposition tools as simplicial partitions, and cuttings in hyperplane arrangements, which were previously used in the study of point-hyperplane incidences; see, e.g., \cite{ManyCells} and \cite[Section 4]{JirkaBook}.) Specializing to this category of convex sets $K$, our $1$-dimensional reduction will be achieved with $\nu\approx \eps^{\alpha_d}$.

Our weak $\eps$-net construction combines such powerful notions as simplicial partitions \cite{PartitionTrees}, cuttings in hyperplane arrangements \cite{Cuttings,Cuttings1}, upper
bounds on the complexity of many cells in hyperplane arrangements \cite{ManyCellsAMS,ManyCellsAS}, strong $\eps$-nets \cite{HW87}, and point-selection theorems \cite{AlonSelections,SelectionsSODA}. Nevertheless, the eventual net boils down to the following two basic ingredients:  (1) $1$-dimensional $\Omega(\eps^{\alpha_d+\gamma})$-nets, which are constructed within few vertical lines and with respect to carefully chosen point sets, and (2) strong $\Omega\left(\eps^{\alpha_d+\gamma}\right)$-nets with respect to convex $(2d)$-hedra in $\reals^2$.

\subsection{Paper organization} 


\smallskip
In Section \ref{Sec:Prelim} we introduce the essential geometric machinery, 
lay down the recursive framework for bounding the quantities $f_d(\eps)$, and, lastly, provide a more comprehensive roadmap to the entire proof of Theorem \ref{Thm:Main}. 

In Section \ref{Sec:MultipleSelection} we obtain the cornerstone reduction of the weak $\eps$-net problem to a sequence of $1$-dimensional $\nu$-nets. 

In Section \ref{Sec:MainRecurrence} we establish Theorem \ref{Thm:Main} by fixing the underlying set $P$ of $n$ points, and then describing a small-size net that pierces all the $\eps$-heavy convex sets $K$. 
To this end, we analyze the implicit structure of the subsets of (at least) $\lceil\eps n\rceil$ points within $P$, that are cut out by the $\eps$-heavy convex sets $K$, with particular emphasis on certain families of simplices that are determined by such subsets of $\lceil \eps n\rceil$ points. As a result,
 the convex sets under our consideration are subdivided into finer sub-categories, to be dispatched by separate nets whose details are relegated to Sections \ref{Sec:Surrounded} and \ref{Sec:VerticallyConvex}. 
Namely, in Section \ref{Sec:Surrounded}, we show how to pierce the convex sets $K$ that determine (so called) {\it $\delta$-punctured} subsets of $\lceil\eps n\rceil$ points within $P$, whereas
 in Section \ref{Sec:VerticallyConvex} we pierce the remaining $\eps$-heavy convex sets, which yield
{\it $\delta$-hollow} subsets within $P$.

In Section \ref{Sec:Final}, we sum-up the properties of our improved weak $\eps$-net construction, and outline a few promising directions for further investigation.


\section{Geometric preliminaries}\label{Sec:Prelim}

\subsection{Basic notation}\label{Subsec:Notation}

\noindent{\bf Asymptotic estimates.} For any $x,y\in \reals$, we denote $x\ll y$ whenever $x=O(y)$.

\medskip
\noindent{\bf Convex hulls and affine closures.} In the sequel, we use $\conv(A)$ to denote the convex hull of any point set $A\subseteq \reals^d$, and we use $\aff(A)$ to denote the {\it affine closure} of $A$ -- the smallest affine space (i.e.,  a translate of a linear subspace of $\reals^d$) that contains $A$.
We say that a finite point set $A\subset \reals^d$ is in {\it a convex position} if every point of $A$ lies on the boundary of $\conv\left(A\right)$.

\medskip
\noindent{\bf Projections.} 
For any point $x=(x_1,\ldots,x_d)\in \reals^d$, and any set $A\subseteq \reals^d$, we denote their vertical projections by $x^\perp:=(x_1,\ldots,x_{d-1})$ and $A^\perp:=\{x^\perp\mid x\in A\}$.

Conversely, for any point $y\in \reals^{d-1}$, and any set $B\subseteq \reals^{d-1}$, we use $y^*$ to denote the vertical line $\{x\mid x^\perp=x\}$ {\it in $\reals^d$} over $x$, and we
use $B^*$ to denote vertical prism $\bigcup\{y^*\mid y\in B\}$ over $B$. 
(With some abuse of notation, these definitions will also extend to points and sets in $\reals^{d}$. For example, for any point $x\in \reals^d$, we will use $x^*$ to denote the vertical line through $x$ in $\reals^d$.)

We say that a set $A\subseteq \reals^d$ lies {\it above} another set $B\subseteq \reals^d$ (whereas $B$ lies {\it below} $A$) if (i) we have that either $A^\perp\subseteq B^\perp$ or $A^\perp\supseteq B^\perp$, and (ii) for any vertical line $\ell$, that crosses both sets, the interval $\ell\cap A$ lies entirely above $\ell\cap B$. (Notice that the described partial relation between sets in $\reals^d$ is not transitive.)

\medskip
\noindent {\bf General position.} To simplify the exposition, we assume in what follows, that the underlying point set $P$ in Theorem \ref{Thm:Main} is in {\it a general position} in $\reals^d$ (as described, e.g.,  \cite[Section 5]{JirkaBook} and \cite{GeneralPosition}).  In particular, no $d+1$ of them lie on the same hyperplane, and no $d$ of them determine a hyperplane that is parallel to any of the coordinate axes; hence, no $d$ of the vertically projected points lie in the same $(d-2)$-plane within $\reals^{d-1}$. 
Furthermore, the affine hull $\aff(A)$ of any subset $A\subseteq P$ of at most $d+1$ points must be an affine space of dimension $|A|-1$. 
More generally, for any $k$ pairwise disjoint subsets $A_1,\ldots,A_k\subseteq P$, cardinality $|A_i|\leq d+1$ each,  the intersection $\bigcap_{i=1}^d \aff(A_i)$ is an affine space of dimension $\max\{-1,d-\sum_{i=1}^l (d-|A_i|+1)\}$ (where $-1$ is the dimension of the empty set), and the analogous property should hold, in $\reals^{d-1}$, for the vertically projected set $P^\perp$.

A standard perturbation argument \cite{GeneralPosition} shows that these assumptions incur no loss of generality: a weak $\eps$-net with respect to an infinetisimately perturbed point set, now in a general position, would immediately yield such a net with respect to the original set.


\medskip
\noindent{\bf Simplices.} Let $A$ be a finite set, and $k\geq 0$. In the sequel, we use ${A\choose k}$ to denote the collection of all the $k$-size subsets of $A$. (This collection is empty if $|A|<k$.)  We then use ${A\choose \leq k}$ to denote the collection of all such subsets of $A$ whose cardinality is at most $k$.

If $A$ is a finite {\it point set} in general position in $\reals^d$, then the elements of ${A\choose \leq d+1}$ are called {\it simplices}.
Specifically, for $0\leq k\leq d$, every $(k+1)$-size set $S\in {A\choose k+1}$ determines a {\it $k$-dimensional simplex $\conv(S)$} (or, shortly, a {\it $k$-simplex}) which supports a $k$-dimensional affine space $\aff(S)$  \cite[Section 5]{JirkaBook}.
For the sake of brevity, we identify the induced set of ${|A|\choose k+1}$ of $k$-dimensional simplices with ${A\choose k+1}$. The $1$-dimensional simplices of ${A\choose 2}$ will be called {\it $2$-edges}.

\medskip
\noindent {\bf Supporting hyperplanes.} 
In what follows, we refer to any $k$-dimensional affine space as a {\it $k$-dimensional plane}, or shortly a {\it $k$-plane}; the $(d-1)$-dimensional planes in $\reals^d$ are called {\it hyperplanes}. 

Any $d$-point set (i.e., a $(d-1)$-dimensional simplex) $S$ in a general position in $\reals^d$ determines the unique hyperplane $H(S)=\aff(S)$. We then say that $H(S)$ is {\it spanned} by $S$, while the point set (or the simplex) $S$ is {\it supported} by $H(S)$.

Given a subset $\Pi\subset {A\choose d}$ of $(d-1)$-dimensional simplices that are determined by a point set $A$ in general position, we use
$\HH(\Pi)$ to denote the set $\{H(S)\mid S\in \Pi\}$
of all the hyperplanes that support the elements of $\Pi$.

\medskip
\subsection{Random sampling and strong $\eps$-nets}\label{Subsec:PointConfig}

Let $X$ be a (finite) set of elements and $\F\subset 2^X$ be a set of hyperedges spanned by $X$. 
A {\it strong $\eps$-net} for the hypergraph $(X,\F)$
is a subset $Y\subseteq X$ so that $F\cap Y\neq \emptyset$ holds for all the hyperedges $F\in \F$ that satisfy $|F|\geq \eps n$.

\medskip
\noindent{\bf Definition.} Let $X$ be a finite set of $n$ elements, and $r$ be an integer that satisifes $1\leq r\leq n$. An {\it $r$-sample} of $X$ is a subset $Y\subseteq X$ of $r$ elements chosen at random from $X$, so that each such subset $Y\in {X\choose r}$ is selected with uniform probability $1/{n\choose r}$.

\medskip
The Epsilon-Net Theorem of Haussler and Welzl \cite{HW87} states that any hypergraph $(X,\F)$,  as above, that is drawn from a so called range space of a bounded VC-dimension $D>0$, admits a strong $\eps$-net $Y$ of cardinality $r=O\left(\frac{D}{\eps}\log \frac{D}{\eps}\right)$. 
Moreover, such a net $Y$ can obtained, with probability at least $1/2$, by choosing an $r$-sample of $X$. In particular, this implies the following property which will be used in the proof of Theorem \ref{Thm:Main}.

\begin{lemma} \label{Thm:StrongNet}
Let $\eps>0$. Then any finite point set $P$ in the Euclidean space $\reals^d$ of fixed dimension $d\geq 1$, contains a subset $Q$ of $O\left(\frac{1}{\eps}\log \frac{1}{\eps}\right)$ points piercing all the $\eps$-heavy polytopes in $\reals^d$ whose boundaries contain at most $2d$ boundary vertices; the implicit constant of proportionality depends on the dimension $d$.
\end{lemma}

The following more elementary property of $r$-samples will prove useful in Section \ref{Sec:Surrounded}.

 \begin{lemma} \label{Lemma:Sample}
Let $A$ be a set of $m$ elements, and let $A=A_1\uplus A_2\uplus \ldots \uplus A_r$ be a partition of $A$ into $r\geq 1$ parts. Let $S$ be a random $r'$-sample of $A$, so that $r'=\min\{10^3 r,m\}$. Then, with probability at least $1/2$, the parts $A_i$ that intersect $S$ encompass at least $m/3$ elements of $A$.
\end{lemma}
\begin{proof}
If $m\leq 10^3r$, then we have that $S=A$, and there is nothing to prove.
Otherwise, note that at most $m/10$ elements of $A$ can belong to the partition sets $A_i$ with $|A_i|\leq m/(10r)$. Let $A':=\bigcup\left\{A_i\mid 1\leq i\leq r, |A_i|>m/(10r)\right\}$. Then we have $|A'|> m-r\cdot m/(10r)=9r/10$. 
Assume with no loss generality that $r'<m(1-\frac{1}{10r})$ (or, else, the claim again holds with probability $1$).
Fix any $x\in A'$ with the ambient set $A_{i(x)}$, so that $1\leq i(x)\leq r$. Since $|A_{i(x)}|> m/(10r)$, the probability $\p_x$ that $A_{i(x)}$ is missed by the sample $S$ satisfies
$$
\p_x\leq {\lfloor m-m/(10r)\rfloor \choose r'}/{m\choose r'}\leq {\left(1-\frac{1}{100r}\right)}^{10^3r}\leq 1/100.
$$

Hence, the expected number of elements in the union $\bigcup_{x\in R}A_{i(x)}$ of the parts that intersect $S$, is  
$$
\sum_{x\in A'}\p_x\geq 0.99\cdot 99m/100.
$$ 

\noindent Now the lemma follows by Markov's inequality.
\end{proof}

\subsection{Configurations of point sets and hyperplanes in $\reals^d$}\label{Subsec:PrelimConfiguration}

In preparation for the proof of Theorem \ref{Thm:Main}, let us establish several elementary properties of point sets, and their transversals by hyperplanes.

\medskip
\noindent {\bf $\delta$-punctured and $\delta$-hollow point sets.} We say that a set $A$ of $d+1$ points in general position is {\it punctured} if a vertical line $a^*$ through some point $a\in A$ crosses the  interior of $\conv(A\setminus \{a\})$; see Figure \ref{Fig:HollowPunctured} (left).  If such a set $A$ is not punctured, we say it is {\it hollow}, in which case the vertical projection $A^\perp$ of such a set $A$ must be in convex position within $\reals^{d-1}$; see Figure \ref{Fig:HollowPunctured} (right).

More generally, for $\delta\geq 0$, we say that a finite set $A$ in general position in $\reals^d$ is {\it $\delta$-punctured} if it encompasses more than $\delta{|A|\choose d+1}$ punctured  subsets $B\in {A\choose d+1}$, and that it is {\it $\delta$-hollow} otherwise (in which case $A$ encompasses at least $(1-\delta){|A|\choose d+1}$ hollow subsets $B\in {A\choose d+1}$). 
In particular, this point set $A$ is $0$-hollow if and only if its projection $A^\perp$ is in convex position in $\reals^d$.

\begin{figure}
    \begin{center}
       
       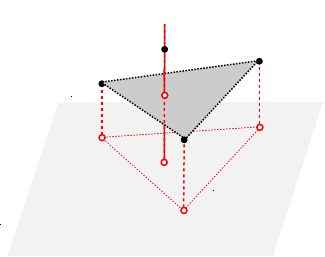\hspace{2cm}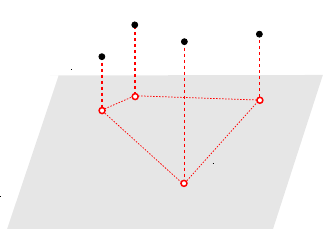
        \caption{\small Left: The set $A=\{p_0,p_1,p_2,p_3\}$ in $\reals^3$ is punctured, as the vertical line $p_0^*$ through $p_0$ crosses $\conv(p_1,p_2,p_3)$. Right: The set $A=\{p_0,p_1,p_2,p_3\}$ in $\reals^3$ is hollow.}
        \label{Fig:HollowPunctured}
    \end{center}
\end{figure}

\begin{lemma}\label{Lemma:ConvexProjSample}
For any dimension $d\geq 2$ there is a constant $c_0>0$ with the following property: 

Let $A$ be a $\delta$-hollow set of $m$ points in $\reals^d$, and $S\subseteq A$ be an $r$-sample of $A$, so that $1\leq r\leq m$ and $\delta\leq c_0/r^{d+1}$.  Then $S$ is $0$-hollow with probability at least $99/100$.
\end{lemma}
\begin{proof}
We can assume that $r\geq d+1$, or else the $0$-hollowness holds trivially, with ${S\choose d+1}=\emptyset$.
Let $B\in {A \choose d+1}$ be a punctured $(d+1)$-subset. Notice that $B$ arises in our $r$-sample $S\subseteq A$ with probability $\displaystyle \p:=\frac{{m-d-1\choose r-d-1}}{{m\choose r}}=\Theta\left(\frac{r^{d+1}}{m^{d+1}}\right)$. Hence, the expected number of such $(d+1)$-subsets $B\in {A\choose d+1}$ is at most

\begin{equation}\label{Eq:ConvProj}
\delta {m\choose d+1} \p=O\left(\delta {m\choose d+1}\left(\frac{r}{m}\right)^{d+1}\right).
\end{equation}

\noindent A suitably small choice of the constant $c_0>0$ guarantees that the left-hand side of (\ref{Eq:ConvProj}) is always smaller than $1/100$. In particular, this implies that, with probability at least $99/100$, {\it all} the $(d+1)$-size subsets in ${S\choose d+1}$ are hollow.
\end{proof}

\noindent {\bf Surrounded subsets of points.} We say that a compact (yet not necessarily finite, or convex) set $\Delta_0$ in $\reals^d$ is {\it surrounded} by $d+1$ compact sets $\Delta_1,\ldots,\Delta_{d+1}\subseteq \reals^{d}$ if for any choice $p_0\in \Delta_0,p_1\in \Delta_1,p_2\in \Delta_2,\ldots,p_{d+1}\in \Delta_{d+1}$,  the point $p_0$ lies in the interior of $\conv(p_1,\ldots,p_{d+1})$. (In particular, no $d+1$ points in such a $(d+2)$-size sequence $p_0,\ldots,p_{d+1}$ can lie in the same hyperplane.) See Figure \ref{Fig:SurroundedSet} (left).

\medskip
For the sake of brevity, we say that a {\it point} $q\in \reals^d$  is {\it surrounded} by $d+1$ sets $\Delta_1,\ldots,\Delta_{d+1}\subseteq \reals^{d}$ if the singleton set $\Delta_0:=\{q\}$ is surrounded by $\Delta_1,\ldots,\Delta_{d+1}$.

\begin{lemma}\label{Prop:NotSurrounded} Let $\Delta_1,\ldots,\Delta_{d+1}$ be $d+1$ distinct compact convex sets in $\reals^{d}$, and $\Delta_0$ be a compact convex set in $\reals^d$ that intersects $\conv\left(\bigcup_{i=1}^{d+1} \Delta_i\right)$, and yet is {\it not} surrounded by $\Delta_1,\ldots,\Delta_{d+1}$. Then there exists a hyperplane that crosses $\Delta_0$ together with at least $d$ among the remaining sets $\Delta_i$, with $1\leq i\leq d+1$. Refer to Figure \ref{Fig:SurroundedSet} (right).
\end{lemma}

\begin{proof} 
Notice that, since $\Delta_0$ intersects $\conv\left(\bigcup_{i=1}^{d+1} \Delta_i\right)$, there must exist points $p_0\in \Delta_0,p_1\in \Delta_1,\ldots,p_{d+1}\in \Delta_{d+1}$ so that $p_0\in \conv(p_1,\ldots,p_{d+1})$.\footnote{This statement does not necessarily hold if the sets $\Delta_i$ are not convex.}
The claim of the lemma follows immediately if there exists a hyperplane $H$ that contains $p_0$ together with some $d$ points among $p_1,\ldots,p_{d+1}$. Assume, then, that there is no such hyperplane $H$; hence, $p_0$ lies in the interior of $\conv(p_1,\ldots,p_{d+1})$. 
However, since $\Delta_0$ is {\it not} surrounded by $\Delta_1,\ldots,\Delta_{d+1}$, there must also exist some points $q_0\in \Delta_0,q_1\in \Delta_1,\ldots,q_{d+1}\in \Delta_{d+1}$ so that $q_0$ does not lie in the interior of $\conv(q_1,\ldots,q_{d+1})$.
For each $t\in [0,1]$, and each $0\leq i\leq d+1$, we define the point $q_i(t):=(1-t)p_i+tq_i$, which clearly lies within the convex set $\Delta_i$, and use $\tau(t)$ to denote the convex hull $\conv(q_1(t),\ldots,q_{d+1}(t))$. Notice that $\tau(0)=\conv\{p_i\mid 1\leq i\leq d+1\}$ contains the point $p_0=q_0(0)$ in its interior, whereas the interior of $\tau(1)=\conv\{q_i\mid 1\leq i\leq d+1\}$ no longer contains $q_0=q_0(1)$. 
Thus, there must be $t^*\in [0,1]$ so that $q_0(t^*)$ lies on the boundary of $\tau(t^*)$. 
It can, therefore, be assumed, with no loss of generality, that $q_0(t^*)\in \conv(q_1(t^*),\ldots,q_{d}(t^*))$, so that any hyperplane through $\conv(q_1(t^*),\ldots,q_{d}(t^*))$ contains $q_0(t^*)$ and intersects $\Delta_1,\ldots,\Delta_{d}$ at the respective points $q_1(t^*),\ldots,q_{d}(t^*)$.
\end{proof}

\noindent{\bf Separated families of convex sets.} We say that a finite family $\Sigma=\{\Delta_1,\ldots,\Delta_m\}$ of $m\geq d+1$ compact convex sets in $\reals^d$ is {\it separated} if no $d+1$ among its elements can be crossed by the same hyperplane. Equivalently, any choice $(p_1,\ldots,p_m)\in \Delta_1\times\ldots\times\Delta_m$ has to yield an $m$-point set $\{p_1,\ldots,p_m\}$ with the property that no $d+1$ among its elements lie in the same hyperplane. \cite{Cappell,PolWen,SelectionsSODA}.
Therefore, if a compact convex set $\Delta_0$ is surrounded by some $d+1$ compact convex sets $\Delta_1,\ldots,\Delta_{d+1}$, this necessarily yields a separated $(d+2)$-size family 
$\Sigma=\{\Delta_0,\ldots,\Delta_{d+1}\}$, which yields the following immediate corollary of Lemma \ref{Prop:NotSurrounded}.

\begin{corollary}\label{Corol:Surrounded}
	Let $\Delta_0,\Delta_1,\ldots,\Delta_{d+1}$ be $d+2$ distinct compact convex sets in $\reals^{d}$, so that $\Delta_0$ intersects $\conv\left(\bigcup_{i=1}^{d+1} \Delta_i\right)$. Then $\Delta_0$ is surrounded by $\Delta_1,\ldots,\Delta_{d+1}$ if and only if the family $\{\Delta_0,\Delta_1,\ldots,\Delta_{d+1}\}$ is separated.
\end{corollary}


\begin{figure}
    \begin{center}      
        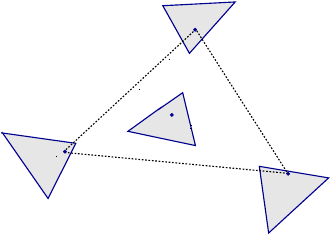\hspace{2cm}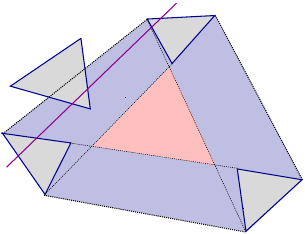
        \caption{\small Left: The set $\Delta_0$ is surrounded by the three other sets $\Delta_1,\Delta_2$ and $\Delta_3$ in $\reals^2$. For any choice of $p_0\in P_0,p_1\in \Delta_1,p_2\in \Delta_2$, and $p_3\in \Delta_3$, the point $p_0$ lies in the interior of $\conv(p_1,p_2,p_3)$. Right: Lemma \ref{Prop:NotSurrounded} in $\reals^2$. The compact convex set $\Delta_0$ intersects the convex hull of $\bigcup_{i=1}^3\Delta_i$, yet it is not surrounded by the compact convex sets $\Delta_1,\Delta_2,\Delta_{d+1}$. Hence, there is a hyperplane (i.e., a line) that intersects $\Delta_0$ and at least 2 sets among $\Delta_1,\Delta_2$ and $\Delta_3$. In particular, the family $\{\Delta_0,\Delta_1,\Delta_2,\Delta_3\}$ is not separated.}
        \label{Fig:SurroundedSet}
    \end{center}
\end{figure}

\medskip
\noindent{\bf Crossing simplices with hyperplanes.} Our proof of Theorem \ref{Thm:Main} will use following elementary property of hyperplane transversals to a family of simplices. 

\begin{lemma}\label{Lemma:ExtremalHyperplane}
Let $\Sigma=\{\Delta_1,\ldots,\Delta_k\}$ be a family of $k\geq 1$ closed simplices $\Delta_i$ in $\reals^d$ with pairwise disjoint vertex sets $V(\Delta_i)$, and so that the point set $\biguplus_{i=1}^d V(\Delta_i)$ is in general position, and has cardinality at least $d$. 
Furthermore, suppose that $\Sigma$ admits a transversal by a hyperplane. Then there is such a hyperplane that passes through some $d$ distinct vertices of $\biguplus_{i=1}^d V(\Delta_i)$, and intersects every simplex in $\Sigma$. See Figure \ref{Fig:ExtremalHyperplane}.
\end{lemma}

\begin{figure}
    \begin{center}      
        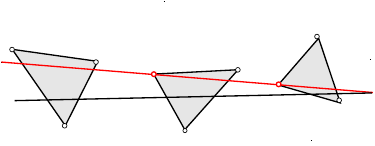
        \caption{\small Lemma \ref{Lemma:ExtremalHyperplane} in dimension $d=2$: moving a transversal hyperplane (i.e., a line) $H$ to $\Sigma=\{\Delta_1,\Delta_2,\Delta_3\}$ to a position in which it contains a pair of vertices while still intersecting every closed simplex $\Delta_i\in \Sigma$.}
        \label{Fig:ExtremalHyperplane}
    \end{center}
\end{figure}

\begin{proof}
Let $H$ be any hyperplane transversal to $\Sigma$. We continuously translate $H$ in a fixed direction until the first time $H$ meets a vertex $v_1$ of some simplex $\Delta_{i_1}\in \Sigma$. We then continuously rotate $H$ in a fixed direction around $v_1$, so the normal to $H$ is moving along a fixed great circle of $\SS^{d-1}$, until the first time $H$ hits an additional vertex $v_2$ of a simplex $\Delta_{i_2}\in \Sigma$ (where $i_1$ is not necessarily different from $i_2$). This procedure is repeated $d-1$ times, each time adding a vertex $v_j$ of some simplex $\Delta_{i_j}\in \Sigma$, for $2\leq j\leq d$, to the hyperplane $H$ which we rotate around the affine hull of the previously added vertices $v_1,\ldots,v_{j-1}$. (This is possible because the vertex sets $V(\Delta_i)$ of the simplices $\Delta_i\in \Sigma$ encompass a total of at least $d$ vertices, whilst $H$ never contains more than $d$ of these vertices.) Once $H$ encounters a new vertex $v_{i_j}$, this vertex never leaves $H$, which reduces by $1$ the available degrees of freedom that are available to the following rotation step. The process terminates when $H$ contains $d$ vertices, so its rotation is no longer possible. 

Notice that, at every step, $H$ meets every closed simplex $\Delta_{j}\in \Sigma$, for the contact between $H$ and $\Delta_j$ cannot be lost before $H$ ``gains" at least one of the vertices of $\Delta_j$.
\end{proof}

\subsection{Simplicial hypergraphs} 
Our improved weak $\eps$-net construction relies on the intrinsic structure of certain simplicial hypergraphs that are induced by the ground set $P$ (and by certain $\lceil \eps n\rceil$-size subsets within $P$).

\medskip
\noindent{\bf Definition.} Let $1\leq k\leq d$. A {\it $k$-uniform hypergraph in $\reals^d$} is a pair $(P,E)$, where $P$ is a finite point set in a general position in $\reals^d$, and $E\subseteq {P\choose k}$ is a family of $k$-subsets which can be identified with $(k-1)$-dimensional simplices whenever $1\leq k\leq d+1$.
In the vast majority of the instances in the sequel, we have $k=d$, so $E$ consists of $(d-1)$-dimensional simplices. 



\medskip
\noindent{\bf Joins.} Given a collection $\Pi\subseteq {P\choose d-1}$ of $(d-2)$-dimensional simplices, and a finite point set $X\subseteq \reals^d\setminus P$ (in general position with respect to the underlying set $P$), we define their {\it join} as $\Pi\ast X:=\{\conv(\tau\cup \{p\})\mid \tau\in \Pi,p\in X\}$. That is, $\Pi\ast X$ consists of all the $(d-1)$-dimensional simplices so that $d-1$ of their vertices form a hyperedge of $\Pi$, and the remaining vertex belongs to $X$.

\medskip
\noindent{\bf Selection theorems.} At the heart of our improved construction lies the following crucial property of simplicial hypergraphs, due to Alon, B\'ar\'any, F\"uredi, and Kleitman \cite{AlonSelections}. See \cite[Section 9]{JirkaBook} and \cite{SelectionsSODA} for a comprehensive exposition, and Figure \ref{Fig:SelectionSimplicial} (left) for an illustration in dimension $2$.

\begin{theorem}[``The Second Selection Theorem"]\label{Theorem:SecondSelection}
 For any dimension $d\geq 1$ there exists a constant $\beta_d>0$ with the following property:

Let $(P,E)$ be a $(d+1)$-uniform simplicial hypergraph in $\reals^d$ with $|E|=h{n\choose d+1}$ edges, then there is a point $x\in \reals^d$ that pierces at least $\Omega\left(h^{\beta_d}{n\choose d+1}\right)$ of the $d$-simplices of $E$. 
\end{theorem}

While the original proof of Alon {\it et al.} yields $\beta_d\leq (4d+1)^{(d+1)}$, the best known estimate for the second selection theorem (due to the author \cite{SelectionsSODA}) yields the exponent $\beta_d\leq d^5+O\left(d^4\right)$, and even better ad-hoc bounds for $\beta_d$ are known in dimension $2$ (see, e.g., \cite{EppsteinSelection,GabrielSelection}). Note that the implicit constant of proportionality in the lower bound $\Omega\left(h^{\beta_d}{n\choose d+1}\right)$ cannot be improved beyond $(d+1)^{-(d+1)}$ \cite{BMN} even in the special case $h=1$ (which yields the so called First Selection Theorem).

\begin{figure}
    \begin{center}
       
        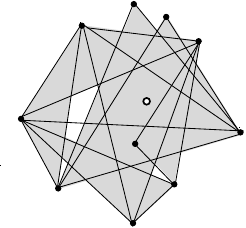\hspace{2cm}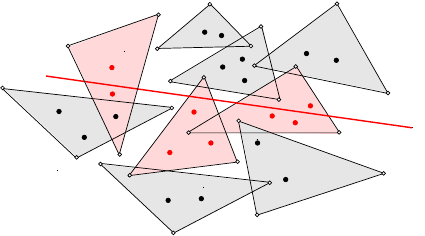
        \caption{\small Left: The Second Selection Theorem. Depicted is a 3-uniform simplicial hypergraph $(P,E)$, in dimension $d=2$; the edge set $E$ consists of at least $h{n\choose 3}$ triangles, and the theorem yields a point $x\in \reals^2$ that pierces $\Omega\left(h^{\beta_2}{n\choose 3}\right)$ of these triangles. Right: Matou\v{s}ek's Simplicial Partition in dimension $d=2$ (the zone of the hyperplane $H$ is colored red).}
        \label{Fig:SelectionSimplicial}
    \end{center}
\end{figure}


\subsection{Matou\v{s}ek's Simplicial Partition Theorem}
Let $P$ be a set of $n$ points in general position in $\reals^d$, and $r>0$ be an integer. A {\it simplicial $r$-partition $\P$} of $P$ is a collection $\{(P_i,\Delta_i)\mid 1\leq i\leq r\}$ of $r$ pairs, where for each $1\leq i\leq r$ we have that $P_i\subset P$ and $\Delta_i$ is a simplex of dimension at most $d$ in $\reals^d$, so that the following properties are satisfied:\footnote{For the sake of brevity, we require that each partition encompasses exactly $r$ sets $P_i$, for $1\leq i\leq r$, some of which can be empty.} 

\begin{enumerate}
\item $P=\biguplus_{i=1}^r P_i$.

\item For each $1\leq i\leq r$ so that $P_i\neq \emptyset$, the cardinality $n_i:=|P_i|$ of $P_i$ satisfies
$
 \lceil n/r\rceil \leq n_i<2\lceil n/r\rceil.
$
 
\item For each $1\leq i\leq r$, the set $P_i$ is contained in the relative interior of $\Delta_i$.
\end{enumerate}

\medskip
\noindent{\bf Definition.} In the sequel, we refer to the simplices $\Delta_i$, that enclose the parts $P_i$ in $\P$, as the {\it cells} of the partition $\P$.\footnote{This is done so as to differentiate between these simplices, which enclose the partition subsets $P_i$, on one hand, and the simplices of ${P\choose d}$ on the other. Notice that the relative interiors of these cells $\Delta_i$ may overlap. Though a point of $P$ may belong to several cells $\Delta_i$, for $1\leq i\leq r$, it is assigned to a {\it unique} cell $\Delta_i$ via the respective set $P_i$.}
For each set $P_i$ in an $r$-partition $\P=\{(P_i,\Delta_i)\mid 1\leq i\leq r\}$ of $P$, and each point $p\in P_i$, we refer to $\Delta_i$ as the {\it ambient cell} of $p$ in $\P$. 

For any hyperplane in $\reals^d$, we say that a point $p\in P$ {\it lies in the zone of $H$ within $\P$} if its ambient cell is crossed by $H$.

\begin{theorem}[The Simplicial Partition Theorem]\label{Theorem:Simplicial}
For any $d\geq 2$ there is a constant $c_\part=c_\part(d)$ with the following property.
For any $n$-point set, and any $1\leq r\leq n$, there is a simplicial $r$-partition $\{(P_i,\Delta_i)\mid 1\leq i\leq r\}$ so that any hyperplane crosses at most $c_\part\cdot  r^{1-1/d}$ of the cells $\Delta_i$, for $1\leq i\leq r$. See Figure \ref{Fig:SelectionSimplicial} (right).
Furthermore, such a partition $\P$ can be computed in time $O^*(n)$.\footnote{The running time was improved by Chan \cite{Chan} to $O(n\log n)$ (with high probability).}
\end{theorem}

For any $d\geq 2$, any $n$-point set $P\subset \reals^{d}$, and any $1\leq r\leq n$, we fix a unique partition $\P=\P_d(P,r)$ as in Theorem \ref{Theorem:Simplicial}, which we briefly denote by $\P(P,r)$ when the ambient dimension of $P$ is clear from the context.

\medskip
If the points of the underlying set $P$ are in a general position, it can be assumed that all the cells $\Delta_i$ in Theorem \ref{Theorem:Simplicial} are $d$-dimensional; furthermore, their vertices can be perturbed in a general position with respect to one another, and with respect to the point set $P$.

\subsection{Arrangements of hyperplanes in $\reals^d$}\label{Subsec:Arrangements}



\paragraph{Definition.} Any finite family $\HH$ of $m$ hyperplanes in $\reals^d$ induces the {\it arrangement} $\A(\HH)$ -- the partition of $\reals^d$ into open faces. The $d$-dimensional faces of $\A(\HH)$ are called {\it cells}, and they are maximal connected regions of $\reals^d\setminus \left(\bigcup \HH\right)$. 
Each cell is a convex polyhedron whose boundary is comprised of $k$-dimensional faces of $\A(\HH)$, for $0\leq k\leq d-1$; each of these faces is a $k$-dimensional polyhedron which lies within a $k$-dimensional plane -- an intersection of some $d-k$ of the hyperplanes of $\HH$.
Specifically, the $(d-1)$-dimensional faces of $\A(\HH)$ are $(d-1)$-dimensional polyhedral portions of the hyperplanes of $\HH$.
The $0$-dimensional faces are called {\it vertices}, and each of them is an intersection of some $d$ among the hyperplanes in $\HH$. 

It is known that an arrangement $\A(\HH)$ of $n$ hyperplanes encompasses a total of $O\left(n^d\right)$ faces of all dimensions (and this bound is tight for $\HH$ in a general position).  

The {\it complexity} $\comp(\Delta)$ of a polyhedral cell $\Delta$ of $\A(\HH)$ is the total number of faces, whose dimensions may vary from $0$ and $d-1$, that lie on its boundary. See, e.g., \cite[Section 6]{JirkaBook} for a more comprehensive background on hyperplane arrangements. 

One can restrict the unbounded faces in $\A(\HH)$ by augmenting $\HH$ with a pair of imaginary hyperplanes that are orthogonal to the $i$-th axis,
$H^+_{i}:=\{X_i=\infty\},H^-_i:=\{X_i=-\infty\}$, for each $1\leq i\leq d$. With this adjustment, every original cell in $\A(\HH)$ now corresponds to exactly one bounded {\it polytopal} cell in the augmented arrangement $\A(\HH)$. In what follows, we can restrict our analysis to these latter cells. 

\begin{theorem}\label{Theorem:SampleHyperplanes}
Let $\HH$ be a family of $m$ hyperplanes in $\reals^d$, and $0<r\leq m$ integer. Then, with probability at least $1/2$, an $r$-sample $\R\in {\HH\choose r}$ of $\HH$ crosses every open segment in $\reals^d$ that is intersected by at least $c(m/r)\log r$ hyperplanes of $\HH$.\footnote{In the sequel we use $\log x$ denotes the standard binary logarithm $\log_2 x$.} 
Here $c>0$ is a sufficiently large constant that does not depend on $m$ or $r$.
\end{theorem}

The proof of Theorem \ref{Theorem:SampleHyperplanes} can be found, e.g., in \cite{ChazelleBook}. It is established by applying the Epsilon Net Theorem to the range space of hypergraphs in which every vertex set is a finite family $\HH$ of hyperplanes in $\reals^d$, and every edge consists of all the hyperplanes in $\HH$ that are crossed by a certain segment with $\reals^d$.

\medskip
\noindent{\bf The bottom-vertex triangulation.}
Every polytopal cell $\Delta$ of $\A(\HH)$ can be decomposed into a collection $\D_\Delta$ of interior-disjoint $d$-dimensional simplices that share the lowermost boundary vertex $v$ of $\Delta$.\footnote{The vertex $v$ is unique given that no hyperplane in $\HH$ is horizontal, which can be achieved by a slight generic rotation of the axis frame.} 
Each of these simplices is of the form $\tau=\conv\left(\{v\}\cup \kappa\right)$ where $\kappa$ is a $(d-1)$-dimensional simplex that is not adjacent to $v$, and arises in a recursive triangulation of some $(d-1)$-dimensional face on the boundary of $\Delta$. Let $\D(\HH)$ denote the resulting collection of interior-disjoint simplices $\bigcup_{\Delta\in \A(\HH)}\D_\Delta$ that comprise the bottom-vertex triangulations of all the $d$-cells of $\A(\HH)$. (A more comprehensive exposition of this fundamental structure can be found, e.g., in \cite{Cuttings1,Cuttings}.)

\smallskip
An easy consequence of Theorem \ref{Theorem:SampleHyperplanes} is that, with probability at least $1/2$, every cell of the bottom-vertex triangulation $\D(\R)$, that is determined by such a sample $\R$, is crossed by at most $(d+1)C(m/r)\log r$ hyperplanes of $\HH$; in other words, it serves as an {\it $\left(\frac{(d+1)C\log r}{r}\right)$-cutting of $\HH$} \cite{JirkaBook}. (This is because every hyperplane of $\HH$ that crosses a simplex $\tau$ of $\D(\R)$ must also cross at least $d$ of the ${d+1\choose 2}$ $1$-dimensional faces, or edges, of $\tau$.)

\medskip
\noindent{\bf Zones.} Let $\C$ be a family of polytopal cells in $\reals^d$ (e.g., faces in the above arrangement $\A(\HH)$, the $d$-dimensional simplices of its refinement $\D(\HH)$, or the cells in Matou\v{s}ek's partitioning theorem). The {\it zone} of any set $X\subset \reals^d$ in $\C$  is the subset of all the cells in $\C$ that intersect $X$.

The following elementary property can be instantly deduced through induction on the ambient dimension $d$, and it is explicitly established by Aronov, Pellegrini, and Sharir \cite[Lemmas 2.2 and 2.3]{APS} as the basis for a more general bound; also see \cite{Tagansky}. 

\begin{lemma}\label{Lemma:Zone}
Let $\HH$ be a collection of $n$ hyperplanes in $\reals^d$, and $K\subset \reals^d$ a convex set. Then the boundary $\partial K$ of $K$ intersects $O\left(n^{d-1}\right)$ cells in the arrangement $\A(\HH)$.
\end{lemma}

\bigskip
\noindent{\bf The complexity of many cells.} Note that for any collection $\HH$ of $n$ hyperplanes, the overall complexity of all the cells in their arrangement $\A(\HH)$ is $O\left(n^d\right)$. (This is because, after perturbing $\HH$ into a general position, every $k$-dimensional face is adjacent to at most $2^{d-k}$ cells.)
However, if $\C$ is a subset of cells in $\A(\HH)$ whose cardinality $m$ is much smaller than $n^d$, then the overall complexity $\comp(\C)=\sum_{\Delta\in \C}\comp(\Delta)$ is $o\left(n^d\right)$. 

\begin{theorem}[Aronov, Matousek, Sharir 1994 \cite{ManyCellsAMS}; Aronov, Sharir 2004 \cite{ManyCellsAS}]\label{Theorem:ManyCells}
Let $\HH$ be a collection of $n$ hyperplanes in $\reals^d$.
Then for any collection $\C$ of $m$ cells in $\A(\HH)$, their overall complexity $\comp(\C)$ is
$
O\left(m^{1/2}n^{d/2}\log^{\zeta_d}n\right),
$
where $\zeta_d=\left(\lfloor d/2\rfloor-1\right)/2$.
\end{theorem}



\subsection{A recursive framework for bounding $f_d(\eps)$}\label{Subsec:RecursiveFramewk}

Let us now lay down a more formal framework for the recursive analysis of the weak $\eps$-numbers $f_d(\eps)$, in which the proof of Theorem \ref{Thm:Main} will be cast.

\medskip
\noindent{\bf Definition.} 
For a finite point set $P$ in $\reals^d$ and $\eps>0$, let $\K(P,\eps)$ denote the family of all the $\eps$-heavy convex sets with respect to $P$. 
To simplify the exposition, every set $K\in \K(P,\eps)$ will be assigned a unique {\it principal subset} $P_K\subseteq P\cap K$ of exactly $\lceil\eps n\rceil$ points. Furthermore, for the sake of our inductive analysis of the quantities $f_d(\eps)$, every finite set $P$ in $\reals^d$
will be assigned a unique weak $\eps$-net $N(P,\eps)$ whose cardinality is at most $f_d(\eps)$.

Theorem \ref{Thm:Main} will be established in a divide-and-conquer manner, by subdividing the convex sets of $\K(P,\eps)$ into several finer sub-classes $\K$, and then constructing a separate net for each class $\K$. For the sake of brevity, we say that $N\subset \reals^d$ is {\it a weak $\eps$-net} for a family $\K$ of convex sets in $\reals^d$ if it pierces every in $\K$ that is $\eps$-heavy with respect to $P$. (In particular, every weak $\eps$-net with respect to $P$ is also a weak $\eps$-net for {\it any} subfamily $\K$ of convex sets in $\reals^d$.) If the parameter $\eps$ is fixed, we can assume that every set in the family $\K$ is $\eps$-heavy, so $N$ is simply a point transversal to $\K$.





\medskip
\noindent {\bf Recurrence in $\eps$ and $P$.} 
To bound the quantity $f_d(\eps)$, the previous recursive arguments (e.g., by Chazelle {\it et al.} \cite{Chazelle} for points in convex position in $\reals^2$, and Matou\v{s}ek and Wagner \cite{MatWag04} for the general case) typically advance by fixing the ground set $P\subset \reals^d$, and constructing a net $N$ for the induced family $\K(P,\eps)$. An upper bound on $|N|$ then serves as an upper bound on $f_d(\eps)$.

The desired net $N$ for $\K(P,\eps)$ is obtained in a top-down fashion, by first decomposing the underlying point set $P$ into $r$ subsets $P_1,\ldots,P_r\subset P$ of size $O\left(n/r\right)$ each (e.g., by the means of Matou\v{s}ek's simplicial partition or another, comparably efficient spatial subdivision). As a rule, the parameter $r$ is either a constant, or an arbitrary small, albeit positive, degree of $1/\eps$. Notice that at least one of the following scenarios is encountered for every convex set $K\in \K(P,\eps)$: 

\begin{enumerate}
\item  If $\eps n/r^{1-1/a}$ the points of $P_K$ fall into in a particular part $P_i$, for some constant $a=a(d)>1$, then $K$ can be relegated to the respective recursive instance $\K\left(P_i,\eps'\right)$, with $\eps'=r^{1/a}\eps>\eps$.
Hence, such a set $K$ can be pierced by the ``subsidiary" net $N(P_i,\eps')$, whose cardinality is $f_d(\eps')=f(r^{1/a}\eps)$.
	\item  Otherwise, if the points of $P_K$ are more evenly distributed between the parts $P_i$, then $K$ is instantly pierced by an explicit non-recursive net of cardinality $r^c/\eps^b$, where $b=b(d)$ and $c=c(d)$ denote fixed exponents that depend on the dimension $d$. 
\end{enumerate}


 The resulting upper bound\footnote{For the sake of brevity, in the sequel we will suppress constant factors within the arguments of the recursive terms $f_d(\eps\cdot w)$, provided that $w$ is a fixed degree of $1/\eps$. To guarantee that the arguments of such terms keep increasing with every recursive step, it will be necessary to assume that $\eps$ is smaller than a certain fixed threshold $\eps_0>0$.}
 
 \begin{equation}\label{Eq:SimpleRecurrence}
 	 f_d(\eps)\leq f_d\left(r^{1/a}\cdot \eps \right)+\frac{r^{c}}{\eps^b}
 \end{equation}

 \noindent combines the recursive term $r \cdot f_d\left(r^{1/a}\cdot \eps\right)$ with an additional, non-recursive, term $r^c/\eps^b$ that accounts for piercing the ``well-spread" convex sets, which could not be relegated to any of the subordinate instances $\K(P_i,\eps')$. 
As a rule, such recurrences in $\eps$ bottom out as soon as $\eps$ by-passes a certain constant $\eps_0$ (in which case the long-known bound of Alon {\it et al.} \cite{AlonSelections} yields $f_d(\eps)\leq f_d(\eps_0)=\Theta(1)$), and ultimately solve to $f_d(\eps)=O^*\left((1/\eps)^{\max\{a,b\}}\right)$. 
In particular, Matou\v{s}ek and Wagner \cite{MatWag04} used the simplicial partition of Theorem \ref{Theorem:Simplicial} to obtain a variant of  (\ref{Eq:SimpleRecurrence}) with $a=d$ and $b=1$ and $c=d^2$, which yields $f_d(\eps)=\tilde{O}\left(1/\eps^{d}\right)$.

\medskip  
\noindent {\bf Towards a more efficient divide-and-conquer.} At the center of our improved construction, whose more comprehensive sketch is given in Section \ref{Subsec:Overview}, lies a reduction of the weak $\eps$-net problem to a handful of {$1$-dimensional} strong $\nu$-nets, with $\nu\approx \eps^{\alpha_d}$, each of them restricted to a certain vertical line $\ell$, and defined over the $\ell$-intercepts of certain $(d-1)$-simplices within ${P\choose d}$.

To amplify this gain, our recursive framework has to be further specialized to dealing with more restricted sub-classes $\K$ of such $\eps$-heavy convex sets $K$ whose induced families ${P_K\choose d}$ of $(d-1)$-simplices attain a considerable overlap with some particular subset $\Pi$ of $(d-1)$-simplices within ${P\choose d}$.

\medskip
\noindent {\bf The edge-constrained families $\K(P,\Pi,\eps,\sigma)$.} Each of these finer families $\K=\K(P,\Pi,\eps,\sigma)$ is determined by $\eps>0$, a finite point set $P\subset \reals^d$, a set of hyperedges $\Pi\subseteq {P\choose d}$, and a threshold $0<\sigma\leq 1$, and is comprised of all the $\eps$-heavy convex sets $K\in \K(P,\eps)$ whose induced subsets ${P_K\choose d}$ of $(d-1)$-simplices satisfy $\left|{P_K\choose d}\cap \Pi\right| \geq \sigma{|P_K|\choose d}$.


In what follows, we refer to the resulting hypergraph $(P,\Pi)$ (or, simply, to $\Pi$, once the ground set $P$ is clear from the context) as the {\it restriction hypergraph}, and to $\sigma$ as the {\it restriction threshold} of the family $\K(P,\Pi,\eps,\sigma)$. 
We then say that the sets $K\in \K(P,\Pi,\eps,\sigma)$ are {\it $(\eps,\sigma)$-restricted} to $(P,\Pi)$.
 All the recursive instances in the sequel will involve positive thresholds $\sigma$ that are either constant, or bounded from below by an arbitrary small degree of $\eps$.\footnote{In the latter case, the dependence of our bounds on $\sigma>0$ will be explicitly spelled out.}

Thus, our recurrence can now advance not only by increasing the parameter $\eps$ while reducing the ground set $P$, but also by restricting the convex sets to ``include" almost ${\lceil\eps n\rceil\choose d}$ edges of the progressively sparser subset $\Pi$ of $(d-1)$-simplices over $P$.

\medskip
\noindent{\bf The improved recurrence for $f_d(\eps)$.} 
For any choice of $0<\eps,\sigma\leq 1$ and $0<\rho\leq 1$, we use $f(\eps,\rho,\sigma)$ to denote the smallest number $f$ so that for any finite point set $P$ in $\reals^d$, and any subset $\Pi\subseteq {P\choose d}$ of density $|\Pi|/{n\choose d}\leq \rho$, there is point transversal of size at most $f$ to $\K(P,\Pi,\eps,\sigma)$. In addition, we set $f(\eps,\rho,\sigma)=1$ whenever $\eps\geq 1$.

Since the underlying dimension $d$ is fixed, for the sake of brevity we use $f(\eps)$ to denote the quantity $f_d(\eps)=f(\eps,1,1)$, and note that the trivial bound $f(\eps,\rho,\sigma)\leq f(\eps)$ always applies.

To establish Theorem \ref{Thm:Main}, in Section \ref{Sec:MainRecurrence} we first obtain a recurrence for the quantity $f(\eps,\rho,\sigma)$, which takes the following more general form:
 
\begin{equation}\label{Eq:GenRecurrence}
f(\eps,\rho,\sigma)\leq f\left(\eps,\rho/h,\sigma/2\right)+\sum_{i=1}^l r_i \cdot f\left(\eps\cdot r_i^{1/a_i}\right)+O\left(1/\eps^{\alpha_d+\eta}\right).
\end{equation}

\noindent Here $\eta>0$ is an arbitrary small constant which is far smaller than our target constant $\gamma>0$ in Theorem \ref{Thm:Main}, $l=O(\log (1/\eps))$, $\alpha_d$ is the exponent in Theorem \ref{Thm:Main}, $a_i\leq \alpha_d+\eta$ for all $1\leq i\leq l$, while $h$ and all the remaining parameters $r_i$ are various degrees of $1/\eps$ that are far smaller than $1/\eps^{\alpha_d}$. 


To bound $f(\eps)$, we begin with applying the more general recurrence (\ref{Eq:GenRecurrence}) for $f(\eps,1,1)$.
This double recurrence in $\eps$ and $\rho$ will bottom out when either (i) the maximum density $\rho$ of $\Pi$ falls below $\eps$, or (ii) $\eps$ bypasses a certain constant threshold $\eps_0$ that is determined in the sequel. In the former case, our line-selection result will yield an ``$O\left(1/\eps^{d-1}\right)$-type" bound of the form  
$$
f(\eps,\rho,\sigma)=O^*\left(r\cdot f\left(\eps\cdot r^{1/(d-1)}
\right)+\frac{\rho}{\eps^d}\right)=O^*\left(r\cdot f\left(\eps\cdot r^{1/(d-1)}
\right)+\frac{1}{\eps^{d-1}}\right).
$$

\noindent In the latter case, we can use the $\tilde{O}\left(\left(1/\eps_0\right)^{d}\right)=O(1)$ bound of Chazelle {\it et al.} \cite{Chazelle}.

By keeping the main parameter $\eps>0$ fixed, and following through $\log_h(1/\eps)=O(1)$ applications of (\ref{Eq:GenRecurrence}) to its first term,
 one can get rid of the term $f(\eps,\rho,\sigma)$, which depends on $\rho$.
By fixing a suitably small $\eta>0$, and using the standard, and fairly general substitution methodology (see, e.g., \cite{MatWag04,EnvelopesHigh} and \cite[Section 7.3.2]{SA}), the remaining recurrence (in $\eps$ alone) will then solve to $\displaystyle f(\eps)=O\left((1/\eps)^{\alpha_d+\gamma}\right)$, for any constant $\gamma>0$.

\subsection{A roadmap to the proof of Theorem \ref{Thm:Main}} \label{Subsec:Overview}
Here is a more comprehensive outline of our improved weak $\eps$-net construction, and its subsequent analysis.


\medskip
\noindent{\bf The reduction to $1$-dimensional nets.}  Let us first describe a generic reduction of the weak $\eps$-net problem in $\reals^d$, in dimension $d\geq 2$, to a sequence of $1$-dimensional $\Theta\left(\eps^d\right)$-nets {\it with respect intervals of $\reals^1$}. Recall that every convex set $K$ that is $\eps$-heavy with respect to the underlying point set $P$, has been assigned a unique {\it principal subset} $P_K\subseteq P\cap K$ of $\lceil \eps n\rceil$ points, which determine a collection of ${|P_K|\choose d}=\Theta\left(\eps^dn^d\right)$ simplices.

Let $\ell$ be any vertical line that pierces $\xi{|P_K|\choose d}$ of the simplices of ${P_K\choose d}$, for some constant $\xi$ that depends only on the dimension $d$. Then, by the convexity of $K$, the interval $K\cap \ell$ must contain $\xi{|P_K|\choose d}$ of the $\ell$-intercepts $\conv(\tau)\cap \ell$ of the simplices $\tau\in {P_K\choose d}$. See Figure \ref{Fig:CrossingTriangles} (left). Hence, all such sets $K$ can be pierced by including in our net $N=N_\ell$ every $\lceil C\eps^dn^d\rceil$-th point of the set $X_\ell=\{\conv(\tau)\cap \ell\mid \tau\in {P\choose d}\}$, whose cardinality is at most ${n\choose d}=O\left(n^d\right)$; here $C>0$ is a constant that depends on $\xi$. In other words, we construct {\it a strong $\nu$-net} with respect to the {\it 1-dimensional set $X_\ell\subseteq \ell$}, and with 
\begin{equation}\label{Eq:Reduction}
	\nu=\Omega\left(\frac{{\eps n\choose d}}{{n\choose d}}\right)=\Omega\left(\eps^d\right).
\end{equation}

By using a small number of canonical vertical lines $\ell$, and effectively replacing the set ${P\choose d}$ with its more elaborately defined subsets of $o(n^d)$ simplices (which would still suit the above $1$-dimensional reduction), our divide-and-conquer strategy will bring the density parameter $\nu$ well above $\eps^d$, and ultimately result in a system of $\nu$-nets $N_\ell\subset \ell$ of overall cardinality $O^*(1/\nu)=o\left(1/\eps^d\right)$. 


\begin{figure}
    \begin{center}

        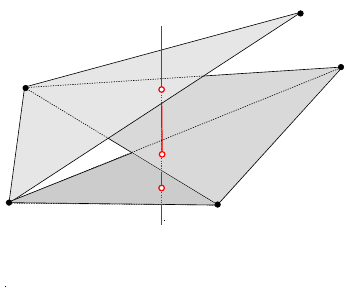\hspace{2cm}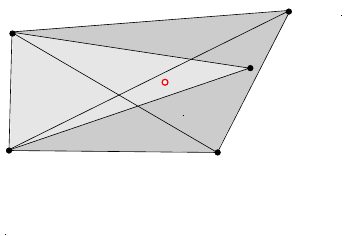
        \caption{\small The reduction to $1$-dimensional nets illustrated in $\reals^3$. The vertical line $\ell$ pierces the triangles $\triangle p_1p_2p_5, \triangle p_3p_4p_1,\triangle p_3p_4p_2$, which are determined by $P_K=\{p_1,\ldots,p_5\}$ (left), if and only if the projection of $\ell$ pierces the projections of the 3 triangles (right). Note that the $\ell$-intercepts of the triangles lie within $\ell\cap \conv(p_1,\ldots,p_5)\subseteq  K\cap \ell$.}
        \label{Fig:CrossingTriangles}
    \end{center}
\end{figure}

\medskip
\noindent{\bf Finding a small canonical set of vertical lines.} The existence of {\it at least one} vertical line $\ell$ for each $\eps$-heavy set $K$, which would pierce many of the simplices ${P_K\choose d}$, is guaranteed by the point-selection theorem of Alon {\it et al.} \cite{AlonSelections} (i.e., Theorem \ref{Theorem:SecondSelection}). To see this, it is enough to consider the vertical projection $P_K^\perp$ of $P_K$ (see Figure \ref{Fig:CrossingTriangles} (right)), apply the lemma to the complete $d$-uniform hypergraph ${P_K^\perp\choose d}$ (with $h=1$), and lift the resulting point $x$ back to a vertical line $x^*$ which now pierces a large fraction of the simplices of ${P_K\choose d}$. \footnote{So far, we have only used the special case of Theorem \ref{Theorem:SecondSelection}, with $h=1$ (the so called {\it First Selection Theorem}). The reasons why we need the full generality of Theorem \ref{Theorem:SecondSelection} (and even more than that) will become clear already in Section \ref{Sec:MultipleSelection}.}


\medskip
Unfortunately, the proposed choice of the vertical line $\ell$ for the sake of our our 1-dimensional reduction, heavily depends not only on $P$ but also on the convex set $K$ at hand. In fact, finding a small set $\L$ of canonical vertical lines, so that at least one of them would suit {\it every} $\eps$-heavy convex set $K$, is at least as hard as finding a weak $\eps$-net with respect to the $(d-1)$-dimensional projection of $P$. (In other words, such a canonical family $\L$ must encompass at least $f_{d-1}(\eps)$ vertical lines, for some ground sets $P$.) 

\medskip
\noindent {\it Remark.} To get around this problem, Chazelle {et al.} \cite{Chazelle} cross the simplices of ${P\choose d}$ using a fairly large canonical family $\L$ of vertical lines. Each line $\ell\in \L$ is assigned a carefully selected subset $\Pi_\ell\subseteq {P\choose d}$ of simplices (all of them crossed by $\ell$); its $\nu$-net $N_\ell$ is then constructed with respect to the $1$-dimensional set $X_\ell:=\{\conv(\tau)\cap \ell\mid \tau\in \Pi_\ell\}$, and with $\nu=\tilde{\Omega}\left(\eps^dn^d/|\Pi_\ell|\right)$ that depends on $|\Pi_\ell|$. 
Despite the considerable size of the canonical set $\L$, the ultimate cardinality of the net $\bigcup_{\ell\in \L}N_\ell$ is shown to be only $\tilde{O}\left(1/\eps^d\right)$.
While this method readily yields the bound $f_d(\eps)=\tilde{O}\left(1/\eps^d\right)$, it is of little use to our more elaborate scheme, which calls for a {\it small} canonical line set $\L$.

\medskip
\medskip
\noindent{\bf Pruning the narrow sets.} As is shown in the sequel, the desired {\it small-size} canonical set $\L$ of vertical lines can still be attained at the expense of introducing a very mild restriction on the principal subsets $P_K$, that are cut out by such $\eps$-heavy convex sets; the ``left-out" sets $K$ can be pierced using a more conventional recurrence in $\eps$, that was outlined in the beginning of Section \ref{Subsec:RecursiveFramewk} (with $a\leq 1/(d-1)<1/\alpha_d$).

To this end, we fix an integer $s$, that is an arbitrary small (albeit fixed) degree of $1/\eps$, project the points of $P$ onto $\reals^{d-1}$, and construct the simplicial $s$-partition $\P_{d-1}\left(P^\perp,s\right)$, within $\reals^{d-1}$, of the projected set $P^\perp$, as described in Theorem \ref{Theorem:Simplicial}.
Lifting the sets of $\P_{d-1}(P^\perp,s)$ back to $\reals^d$ yields a $d$-dimensional partition 
$\V(P,s)=\{\left(P_i,\Delta_i\right)\mid 1\leq i\leq s\}$ of $P$ so that each subset $P_i$ is contained in the vertical prism $\Delta_i^*$ above the respective simplicial cell $\Delta_i\subset \reals^{d-1}$, and satisfies $|P_i|\leq 2\lceil n/s \rceil$. In the sequel, we refer to $\V(P,s)$ as the {\it vertical simplicial partition}.
We say that an $\eps$-heavy convex set $K$ is {\it narrow} if a significant fraction of the points of $P_K$ fall into only $O\left(s^{1-1/(d-1)}\right)$ parts of $\V(P,s)$, whose ambient cells can be crossed by a single $(d-2)$-plane $g$ in $\reals^{d-1}$; see Figure \ref{Fig:NarrowSpread} (left). 
Though the simplex families ${P_K\choose d}$, that are determined by narrow sets, may prove hard for interception by few vertical lines,
such sets $K$ can be easily pierced, via elementary recurrence in $\eps>0$, by an auxiliary 
net whose cardinality is bounded by an essentially ``$O\left(1/\eps^{d-1}\right)$-type" term $O\left(s\cdot f_d\left(\eps\cdot s^{1/(d-1)}\right)\right)$. \footnote{In the sequel, $s$ is chosen to be a very small, albeit fixed, degree of $1/\eps$. Matou\v{s}ek and Wagner \cite{MatWag04} used a similar divide and conquer argument, albeit using a $d$-dimensional simplicial partition; the resulting recurrence $f_d(\eps)=s\cdot f_d\left(\eps \cdot s^{1/d}\right)+s^{d^2}$ solved to $f_d(\eps)=O^*\left(1/\eps^d\right)$.}

\begin{figure}
    \begin{center}
        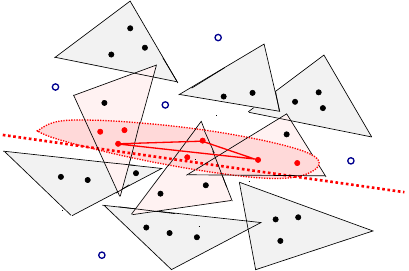\hspace{2cm}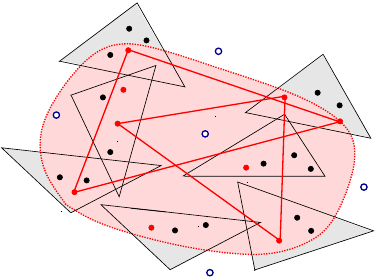
        \caption{\small Using the vertical simplicial partition $\V(P,s)$ in $\reals^3$ -- view from above. Left: The set $K$ (red) is {\it narrow}, so the points of $P_K$ project to the zone of the line $g\subseteq \reals^2$ (purple), in which case the triangles of ${P_K\choose 3}$ are likely to be missed by our canonical set of vertical lines $\ell\in \L$. Right: The set $K$ is {\it spread}, so many triangles in ${P_K\choose 3}$ can be crossed by a line $\ell\in \L$.}
        \label{Fig:NarrowSpread}
    \end{center}
\end{figure}

\medskip
Specializing to the remaining convex sets $K$, whose principal subsets $P_K$ are more ``broadly" distributed among the parts of our vertical partition $\V(P,s)$, our principal line-selection result (Theorem \ref{Theorem:MultipleSelection}) will yield the desired canonical collection $\L=\L(P,s)$ of at most $s^{(d-1)^2}$ vertical lines, and with the following property: {\it For each ``well-spread" $\eps$-heavy convex set $K$ as above, there exists a line $\ell\in \L$ that pierces a fixed fraction of the simplices of ${P_K\choose d}$.} See Figure \ref{Fig:NarrowSpread} (right).
Hence, with the aforementioned exception of the narrow sets, all the $\eps$-heavy sets $K$ can be pierced by a combination of at most $s^{(d-1)^2}$ $1$-dimensional strong $\nu$-nets $N_\ell\subset \ell$, for $\ell\in \L$ and $\nu=\Omega\left(\eps^d\right)$. 

\medskip
\noindent {\bf Attaining $\nu\gg \eps^{\alpha_d}$.} Most of our analysis is, therefore, devoted to bringing the effective density ratios $\nu$ in the previously described $1$-dimensional instances above $\eps^{\alpha_d}$. 
Informally, to reduce the enormous denominator ${n\choose d}$ in (\ref{Eq:Reduction}), we seek to replace the complete collection ${P\choose d}$ of $(d-1)$-simplices with some {\it smaller} subset $\Pi\subset {P\choose d}$. 
However, in order for the resulting net to still pierce a given $\eps$-heavy convex set $K$, it is essential for such a smaller-size substitute $\Pi$ of ${P\choose d}$ to include a large enough fraction of the $(d-1)$-simplices of ${P_K\choose d}$.
Since no single and {\it small-size} replacement $\Pi$ of ${P\choose d}$ can simultaneously suit {\it all} the $\eps$-heavy convex sets $K$, these sets must be subdivided into several categories, which can be more easily dispatched via separate partial nets.

\medskip
\noindent{\it Flat convex sets are easy to pierce.} It is well known that hypergraphs induced by hyperplane ranges have Vapnik-Chervonenkis dimension $d$ and, therefore, admit strong $\eps$-nets of cardinality $O\left(\frac{1}{\eps} \log \frac{1}{\eps}\right)$ \cite{HW87}. However, strong $\eps$-nets are no longer possible if the points of $P$ are slightly perturbed in a way that transforms many of the $\eps$-heavy hyperplanes $H$ into arbitrary thin, albeit strictly convex, sets. 
Instead, most bounds on point-hyperplane incidences, including the tight Szemeredi-Trotter estimate \cite{SzT} in the plane, stem from the subtler fact that, for any $\eps$ that is significantly larger than $1/n$,
any $n$-point set in $\reals^d$ determines only $o(n^d)$ $(d-1)$-simplices whose supporting hyperplanes are ``$\eps$-heavy" (i.e., contain at least $\eps n$ points).


At the heart of our improved construction lies a more robust notion of combinatorial ``co-planarity", which applies to a broader class of $\eps$-heavy ``hyperplane-like" convex sets (whose principal sets $P_K$ may be in a perfectly general position).
To this end, we will use a $d$-dimensional simplicial $t$-partition $\P(P,t)$ of $P$,\footnote{To be precise, such a partition $\P(P_i,t)$ will be constructed in Section \ref{Sec:Surrounded} for each part $P_i$ of the vertical simplicial partition $\V(P,s)$ that underpins our small-size canonical family $\L=\L(P,s)$ of vertical lines. 
Since $s$ is an arbitrary small degree of $1/\eps$, this will yield a simplicial $O^*(t)$-partition $\P(P,s,t)=\bigcup_{1\leq i\leq s}\P(P_i,t)$ of $P$ whose properties resemble those described in Matou\v{s}ek's Theorem \ref{Theorem:Simplicial}.} with $t=(1/\eps)^{3d/4+o_d(1)}=o\left((1/\eps)^{d}\right)$, and take into account the possible ``co-planarities" which occur between the {\it ambient cells} of the points of $P_K$.

Specifically, we consider such hyperplanes $H$ that contain, so called, {\it short} $2$-edges $\lambda\in {P_K\choose 2}$, so that both of their endpoints fall in the same part $Q\in \P(P,t)$.
Informally, an $\eps$-heavy convex set $K$ will be called {\it flat} if the following property holds for a fixed fraction of these 2-edges $\lambda\in {P_K\choose 2}$: there exists such a hyperplane $H$ through $\lambda$ whose zone, within $\P(P,t)$, encompasses a large fraction of the points of $P_K$. See Figure \ref{Fig:TightSurrounded}.

In what follows, we obtain a relatively sparse subset $\tilde{\Pi}\subseteq {P\choose d}$ which, nevertheless, encompasses nearly $\Omega\left(\eps^dn^d\right)$ of the simplices within every family ${P_K\choose d}$ that is determined by any ``flat" $\eps$-heavy convex set $K$. Hence, plugging this sparser set $\tilde{\Pi}$, or rather the $1$-dimensional cross-sections $\{\conv(\tau)\cap \ell\mid \tau\in \tilde{\Pi}\}$, into the 1-dimensional instances of the canonical lines $\ell\in \L$, will ultimately result in a net of overall cardinality $O^*\left(1/\eps^{\alpha_d}\right)$ that pierces all the ``flat" $\eps$-heavy sets $K$.

\begin{figure}
    \begin{center}
    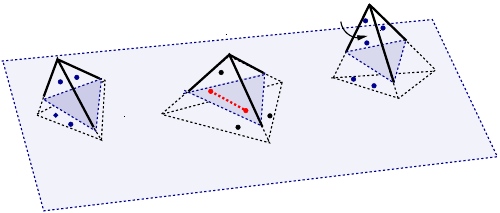
        \caption{\small Characterizing a ``flat" set $K$ in $\reals^3$. For a fixed fraction of the short 2-edges $\lambda\in {P_K\choose 2}$, there exists such a hyperplane $H$ through $\lambda$ whose zone in $\P(P,t)$ encompasses almost $\Omega(\eps n)$ points of $P_K$.}
        \label{Fig:TightSurrounded}
    \end{center}
\end{figure}

\medskip
The remaining convex sets, whose principal subsets $P_K$ are not sufficiently ``co-aligned" with their short edges $\lambda\in {P_K\choose 2}$, will be dispatched via a combination of more elementary (and, for most, non-recursive) nets. For example, the strong net of Lemma \ref{Thm:StrongNet} can be used to get rid of such convex sets that contain a sufficiently heavy convex $(2d)$-hedron. Furthermore, if a convex set $K$ is relatively ``fat" with respect to the partition $\P(P,t)$, the points of $P_K$ can be effectively replaced, for the sake of our 1-dimensional reduction, by the vertices of their ambient cells, which again results in a larger density parameter $\nu$.


\medskip
\noindent{\bf  $\delta$-punctured vs. $\delta$-hollow sets.} To facilitate the outlined divide-and-conquer treatment of the $\eps$-heavy convex sets $K$, each of them has to be assigned a $(d+1)$-tuple $A_0,A_1,\ldots,A_{d}\subseteq P_K$ of pairwise disjoint subsets whose cardinalities are close to $\eps n$, and so that $A_0^\perp$ is surrounded, within $\reals^{d-1}$, by $A_1^\perp,\ldots,A_d^\perp$; namely, any choice of $d+1$ points $p_0\in A_0,\ldots,p_{d}\in A_{d}$ yields a punctured set $\{p_0,\ldots,p_{d}\}$, with $p_{0}^\perp\in \conv\left(p^\perp_1,\ldots,p^\perp_{d}\right)$. See Figure \ref{Fig:DeltaPunctured} (left).

\begin{figure}
    \begin{center}
    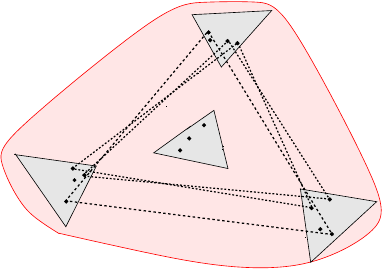\hspace{2cm}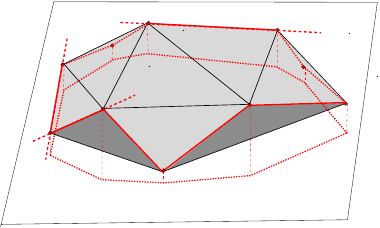
        \caption{\small The two classes of convex sets in $\reals^3$. Left: The principal subset $P_K$ is $\delta$-punctured (view from above). Hence, it contains 4 subsets $A_0,A_1,A_2,A_3$, of cardinality almost $\eps n$, which satisfy $p_0^\perp\in \conv(p_1^\perp,p_2^\perp,p_3^\perp)$ for all $p_0\in A_0,p_1\in A_1,p_2\in A_2$ and $p_3\in A_3$. Right: The principal subset $P_K$ is $0$-hollow; the edges of the silhouette of $\conv(P_K)$ are highlighted, along with their vertical projections. Notice that the directions of the silhouette edges may diverge, even if $\conv(P_K)$ is relatively slim.}
        \label{Fig:DeltaPunctured}
    \end{center}
\end{figure}

However, such a $(d+1)$-tuple $A_0,\ldots,A_{d}$ can possibly exist only if the principal set $P_K$ is $\delta$-punctured, where $\delta>0$ is a very small degree of $\eps$ that will be determined in the sequel. Indeed, if the set $P_K$ is $\delta$-hollow, then the vast majority of the projected $(d+1)$-tuples $\left\{p_1^\perp,\ldots,p_{d+1}^\perp\right\}\subseteq P^\perp_K$ are in a convex position.  
 Informally, the singular challenge that is posed to our machinery by the $\delta$-hollow sets $K$, boils down to the highly divergent nature of the short 2-edges of $\conv\left(P_K\right)$ that lie in the vicinity of the vertical silhouette of $\conv(P_K)$, as illustrated in Figure \ref{Fig:DeltaPunctured} (right). 
 As a result, the delicate ``co-alignment" between the simplices of ${P_K\choose d}$ and the short 2-edges of ${P_K\choose 2}$, may break down while the convex hull $\conv(P_K)$ remains overly ``slim" (and, therefore, difficult to pierce).

To exploit the almost {\it $(d-2)$-dimensional} distribution of such $\delta$-hollow principal subsets $P_K$ (whose points are concentrated in the vicinity of the vertical silhouette of $\conv(P_K)$), 
we will resort to the more basic recurrence in the ground set $P$ and $\eps>0$, which was described in the beginning of Section \ref{Subsec:RecursiveFramewk}. To this end, we will introduce yet another subdivision of $P$ (and, implicitly, of each principal set $P_K$), which is based on a cutting in the arrangement of the hyperplanes $H(\tau)$ that support certain simplices $\tau\in {P\choose d}$. 

\medskip
\noindent{\bf Fast forward.} The proof of Theorem \ref{Thm:Main} is organized as follows. In Section \ref{Sec:MultipleSelection} we formalize the crucial reduction of the weak $\eps$-net problem to a sequence of $1$-dimensional nets which are restricted to few vertical lines. This machinery is then used in Section \ref{Sec:MainRecurrence} to establish Theorem \ref{Thm:Main} via an efficient recurrence for the quantity $f_d(\eps)$.
To facilitate the weak $\eps$-net construction in Section \ref{Sec:MainRecurrence}, in Sections \ref{Sec:Surrounded} and \ref{Sec:VerticallyConvex} we describe separate nets for the special categories of the $\eps$-heavy convex sets $K$ whose principal subsets $P_K$ of $\lceil \eps n\rceil$ points are, respectively, $\delta$-punctured and $\delta$-hollow.


\section{Piercing many simplices with few canonical lines}\label{Sec:MultipleSelection}



Let us now cast the missing details into the key reduction of the weak $\eps$-net problem in dimension $d\geq 3$ to a system of {\it $1$-dimensional nets} which are constrained to a {\it small-size} canonical family $\L(P,s)$ of vertical lines. As was emphasized in Section \ref{Subsec:Overview}, any such reduction must leave out certain ``narrow" sets $K\in \K(P,\eps)$, to be dispatched by a separate recursive construction, at the expense of introducing an additional recursive term of the form $O\left(s\cdot f_d\left(\eps \cdot s^{1/(d-1)}\right)\right)$.

\subsection{The vertical simplicial partitions $\V(P,s)$} \label{Subsec:VerticalPartition}\label{Subsec:VerticalPartition}
To quantify the relative distributions of the projected principal subsets $P_K^\perp$, we vertically project the entire point set $P$ onto the copy of $\reals^{d-1}$, which is spanned by the first $d$ coordinate axes, and construct the simplicial partition of Theorem \ref{Theorem:Simplicial} of the projected set $P^\perp\subset \reals^d$. 



\medskip
\noindent{\bf Definition.} Let $P$ be an $n$-point set in a general position in $\reals^d$, and $s$ an integer that satisfies $1\leq s\leq n$.\footnote{Unless indicated otherwise, all partition parameters are chosen to be a very small, albeit fixed, degree of $1/\eps$.} We say that $\{(P_i,\Delta_i)\mid 1\leq i\leq s\}$ is the {\it vertical simplicial $s$-partition} (or, shortly, the vertical simplical partition, if the parameter $s$ is clear from the context) if 
$\{(P^\perp_i,\Delta_i)\mid 1\leq i\leq s\}$ is the simplicial $s$-partition $\P_{d-1}(P^\perp,s)$ of the projected set $P^\perp$ in $\reals^{d-1}$ that was described in Theorem \ref{Theorem:Simplicial}.

In particular, every cell $\Delta_i$ is a {\it $(d-1)$-dimensional} simplex in $\reals^{d-1}$, and every subset $P_i$ is contained in the vertical prism $\Delta^*_i$ that is raised over the respective $(d-1)$-dimensional simplicial cell $\Delta_i$.

\medskip
By definition, for each finite point set $P\subset \reals^d$ in general position, and each parameter $1\leq s\leq |P|$, there is a unique such vertical simplicial partition $\{(P_i,\Delta_i)\mid 1\leq i\leq s\}$, which we denote by $\V(P,s)$.
For each $1\leq i\leq s$, and each point $p\in P_i$, we refer to $\Delta_i$ as {\it the ambient cell} of $p$ in $\V(P,s)$, which we denote by $\Delta(p)$.


\medskip
\noindent{\bf Definition.} Let $\varepsilon>0$. We say that a convex set $K\subseteq \reals^d$ is  {\it $\varepsilon$-narrow} with respect to the vertical simplicial partition $\V(P,s)=\{(P_i,\Delta_i)\mid 1\leq i\leq s\}$ if there exists a $(d-2)$-plane $g$ within $\reals^{d-1}$ whose zone within the respective $(d-1)$-dimensional partition $\P_{d-1}(P^\perp,s)=\{(P_i^\perp,\Delta_i)\mid 1\leq i\leq s\}$ encompasses at least $\varepsilon n$ of the projected points of $(P\cap K)^\perp$. See Figure \ref{Fig:NarrowSpread} (left).	
If the convex set $K\subseteq \reals^d$ is not $\varepsilon$-narrow with respect to a vertical simplicial partition $\V(P,s)$, we say that $K$ is {\it $\varepsilon$-spread} in $\V(P,s)$.\footnote{In particular, every ``$\varepsilon$-light" convex set outside $\K(P,\eps)$ is $\varepsilon$-spread. However, most applications of the lemma involve convex sets $K$ that are at least $\varepsilon$-heavy with respect to the underlying point set.} (See Figure \ref{Fig:NarrowSpread} (right).)

\medskip
As a rule, the narrowness threshold $\varepsilon$ will be either $\Theta(\eps)$, or only marginally smaller than the parameter $\eps$ in Theorem \ref{Thm:Main}.

\subsection{Piercing the $\varepsilon$-narrow convex sets}\label{Subsec:Narrow} 

As was mentioned in Section \ref{Subsec:Overview}, the simplices in the $\varepsilon$-narrow sets of $\K(P,\eps)$ may prove hard for interception with few vertical lines. Fortunately, a seizable fraction of their points can be localized in only $O\left(s^{1-1/(d-1)}\right)$ parts $P_i$ of the vertical partition $\V(P,s)$. Hence such sets can still be pierced by a 
net whose cardinality is bounded by an essentially ``$O\left(1/\eps^{d-1}\right)$-type" recursive term which is close to
$O\left(s\cdot f\left(\eps\cdot s^{1/(d-1)}\right)\right)$.

\begin{lemma}\label{Lemma:Narrow}
Let $P$ be a finite point set in $\reals^d$, and $\varepsilon>0$. Let $s>0$ be an integer that satisfies $1\leq s\leq |P|$, and $\V(P,s)$ be the vertical simplicial $s$-partition of $P$ as described in Section \ref{Subsec:VerticalPartition}. Then there exists a set $N_\nar(P,s,\varepsilon)\subset \reals^d$, of cardinality $O\left(s\cdot f\left(\varepsilon\cdot s^{1/(d-1)}\right)\right)$, that pierces every convex set $K$ that is $\varepsilon$-narrow with respect to $\V(P,s)$.\footnote{For the sake of brevity, we keep suppressing the constant multiplicative factors within the recursive terms of the general form $a\cdot f\left(\eps\cdot b\right)$, provided that these factors are much larger than $1/b$ and $1/a$.}
\end{lemma}

\begin{proof}
For every $1\leq i\leq s$ we construct a net $N_i:=N\left(P_i,\tilde{\eps}\right)$ for the local problem $\K\left(P_i,\tilde{\eps}\right)$, where
$$
\tilde{\eps}:=c\cdot \varepsilon {s^{1/(d-1)}}=\Theta\left(\frac{\varepsilon n/\left(s^{1-1/(d-1)}\right)}{n/s}\right),
$$

\noindent and $c>0$ is a suitably small fixed constant that does not depend on $\eps$, $\varepsilon$ or $s$ (but may depend on $d$. Notice that each $N_i$ is a $d$-dimensional net of cardinality $|N_i|\leq f\left(\tilde{\eps}\right)$.

It suffices to show that the union $N_\nar(P,s,\varepsilon):=\bigcup_{1\leq i\leq s} N_i$ pierces all the $\varepsilon$-narrow convex sets. Indeed,  let $K$ be such a narrow convex set with respect to $\V(P,s)$, and $g$ be a ``witness" $(d-2)$-plane within $\reals^{d-1}$ that certifies the $\varepsilon$-narrowness of $K$. That is, at least $\varepsilon n$ points of $P\cap K$ belong to such sets $P_i$ whose ambient $(d-1)$-dimensional cells $\Delta_i\subset \reals^{d-1}$ are crossed by $g$. Since the number of the latter simplices is only $O\left(s^{1-1/(d-1)}\right)$ (via Theorem \ref{Theorem:Simplicial}), the pigeonhole principle yields such a set $P_i$ that contains $\Omega\left(\varepsilon n/\left(s^{1-1/(d-1)}\right)\right)$ of the points of $P\cap K$. A suitable choice of the constant $c>0$ guarantees that $K$ is pierced by the respective net $N_i$.
\end{proof}

\subsection{The Canonical Line-Selection Theorem}\label{Subsec:Canonical}

Specializing to the remaining sets, which are not pierced by the recursive net of Section \ref{Subsec:Narrow}, we obtain a so called {\it canonical set} $\L=\L(P,s)$ of $s^{O\left(d^2\right)}$ vertical lines, with the following property: for any $\Omega^*(\epsilon)$-spread convex set $K\in \K(P,\eps)$, there exists a subset of roughly $\Omega\left(\eps^dn^d\right)$ simplices of ${P_K\choose d}$ that can be crossed by a single line of $\L$. Before we proceed with the more formal definition of such canonical families $\L$ of lines, and the proof of their existence, further notation is in order.





\medskip
\noindent{\bf Definition.} We say that a vertical line $\ell$ in $\reals^d$ is {\it surrounded} by some $d$ sets $\Delta_1,\ldots,\Delta_d\subseteq \reals^{d-1}$ if the vertical intercept $\ell^\perp$ of $\ell$ in $\reals^{d-1}$ is surrounded by $\Delta_1,\ldots,\Delta_d$ according to the definition in Section \ref{Subsec:PointConfig}.

\medskip
Note that, if a vertical line $\ell$ is surrounded by some $d-1$ cells $\Delta_{i_1},\ldots,\Delta_{i_d}\subset \reals^{d-1}$ that arise in the vertical $r$-partition $\V(P,r)$ of Section \ref{Subsec:VerticalPartition}, for $1\leq i_1\neq i_2\neq \ldots \neq i_d\leq r$, then it pierces the interior of {\it every} simplex $\conv(p_1,\ldots,p_d)$ with $p_j\in P_{i_j}$.

\medskip
\noindent{\bf Definition.} Let $\vartheta>0$, and let $\beta_{d-1}$ denote the constant selection exponent that suits Theorem \ref{Theorem:SecondSelection1} in dimension $d-1$. (One can choose $0<\beta_{d-1}\leq d^5+O(d^4)$ \cite{SelectionsSODA}.)

Let $P$ be a set of $n$ points in general position in $\reals^d$, and $1\leq s\leq n$. 
We say that a family $\L$ of vertical lines is {\it $\vartheta$-canonical} with respect to the vertical simplicial partition $\V(P,s)$ if it satisfies the following property:

\smallskip
 {\it For any $0\leq \eps,\sigma\leq 1$, any $(\vartheta\sigma^{\beta_{d-1}}\eps)$-spread set $K\in \K(P,\eps)$ with a principal subset $P_K$ of exactly $\lceil \eps n\rceil$ points, and any subset $\Pi_K\subseteq {P_K\choose d}$ of at least $\sigma {\lceil \eps n\rceil \choose d}$ $(d-1)$-simplices, there is a line $\ell\in \L$ with a subset $\Pi_K(\ell)$ of at least $\vartheta\sigma^{\beta_{d-1}}{\lceil \eps n\rceil \choose d}$ simplices within $\Pi_K$, that are pierced by $\ell$. Furthermore, for every simplex $\{p_1,\ldots,p_d\}\in \Pi_{K}(\ell)$, this line $\ell$ is surrounded by the $d$ ambient cells $\Delta(p_1),\ldots,\Delta(p_d)$ in $\V(P,s)$.} See Figure \ref{Fig:MultipleSelection}.

\begin{figure}
    \begin{center}      
        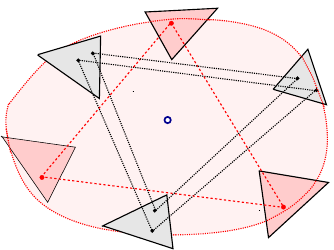
        \caption{\small The Canonical Line-Selection Theorem in dimension $d=3$ (view from above). 
The set  $K\in \K(P,\eps)$ is $(\vartheta\sigma^{\beta_{d-1}}\eps)$-spread in the vertical simplicial partition $\V(P,s)$. The $\vartheta$-canonical family $\L(P,s)$ contains a vertical line $\ell$ piercing at least $\vartheta\sigma^{\beta_{2}}{\lceil\eps n\rceil\choose 3}$ simplices $\{p_1,p_2,p_3\}\in \Pi_K$, such that the ambient cells $\Delta(p_1),\Delta(p_2),\Delta(p_3)\subset \reals^2$ of their vertices $p_1,p_2$, and $p_3$ surround $\ell$.}
        \label{Fig:MultipleSelection}
    \end{center}
\end{figure}

\begin{theorem}[{\it The Canonical Selection Theorem}]\label{Theorem:MultipleSelection}
There is a constant $\vartheta>0$ so that for any finite point set $P$, and any $1\leq s\leq |P|$, there exists a $\vartheta$-canonical set $\L$ of $O\left(s^{(d-1)^2}\right)$ vertical lines with respect to $\V(P,s)$. Furtheremore, such a set $\L=\L(P,s)$ can be constructed from $\V(P,s)$ in time $O\left(s^{(d-1)^2}\right)$.
\end{theorem}

In what follows, every finite point set $P$, and every integer $s$ that satisfies $1\leq s\leq |P|$, will be assigned a unique $\vartheta$-canonical family $\L(P,s)$.

\medskip
Before proceeding with the proof of Theorem \ref{Theorem:MultipleSelection}, let us demonstrate its usefulness to piercing the finer ``edge-constrained" sub-families $\K(P,\Pi,\eps,\sigma)$, which have been introduced in Section \ref{Subsec:RecursiveFramewk}. To this end, the subsets $\Pi_K$ will be set to $\Pi\cap {P_K\choose d}$, which clearly satisfies $|\Pi_K|\geq \sigma{\lceil \eps n\rceil\choose d}$. Therefore, specializing to the $(\vartheta\sigma^{\beta_{d-1}}\eps)$-spread sets $K$, the complete collection ${P\choose d}$ of the $(d-1)$-dimensional simplices can be replaced, for the sake of our $1$-dimensional reduction, by its potentially sparser subset $\Pi$. 
Combined with the use of Lemma \ref{Lemma:Narrow} for the $(\vartheta\sigma^{\beta_{d-1}}\eps)$-narrow sets, this yields a more accurate recursive bound on the quantity $f(\eps,\rho,\sigma)$ whenever the density parameter $\rho$ is significantly smaller than $1$ (and the restriction threshold $\sigma$ remains near-constant).


\begin{theorem}\label{Theorem:Sparse}
Let $0<\eps, \sigma\leq 1$, $P$ be an $n$-point set in $\reals^d$, $s$ an integer that satisfies $1\leq s\leq n$, and $\Pi\subseteq {P\choose d}$. 
Then the family $\K(P,\Pi,\eps,\sigma)$ can be pierced by a net of size
\begin{equation}\label{Eq:ThmSparseBound}
O\left(\frac{s^{(d-1)^2}|\Pi|}{\sigma^{\beta_{d-1}}\eps^dn^d}+s\cdot f\left(\sigma^{\beta_{d-1}}s^{1/(d-1)}\cdot \eps\right)\right).
\end{equation}

Thus, we have that

$$
f(\eps,\rho,\sigma)=O\left(\frac{s^{(d-1)^2}\rho}{\sigma^{\beta_{d-1}}\eps^d}+s\cdot f\left(\sigma^{\beta_{d-1}}s^{1/(d-1)}\cdot \eps\right)\right)
$$

\noindent for all $0<\eps,\sigma\leq 1$, $0\leq \rho\leq 1$, and $1\leq s\leq 1/\eps$.

\end{theorem}

\noindent {\it Remark.}  In particular, for $s$ that is an arbitrarily small (albeit fixed) degree of $1/\eps$, $\rho=\eps$, and $\sigma=\Theta(1)$, Theorem \ref{Theorem:Sparse} yields a ``$1/\eps^{d-1}$-type" estimate $f(\eps,\eps,\sigma)=O^*\left(\frac{1}{\eps^{d-1}}+s\cdot f\left(s^{1/(d-1)}\cdot \eps\right)\right)$.


\begin{proof}[Proof of Theorem \ref{Theorem:Sparse}.]
If $n\leq 1/\eps$, then the entire family $\K(P,\eps)$ is pierced by the net $N=P$. Hence, the second part of the theorem is an immediate corollary of its first part.

To see the first part, we construct the vertical simplicial partition $\V(P,s)=\{(P_1,\Delta_1),\ldots,(P_s,\Delta_s)\}$ that was described in Section \ref{Subsec:Narrow}.
We then invoke Lemma \ref{Lemma:Narrow}, with $\varepsilon=\vartheta\sigma^{\beta_{d-1}}\eps$, to construct a net $N_\nar(P,s,\varepsilon)$ of cardinality $O\left(s\cdot f\left(\varepsilon\cdot s^{1/(d-1)}\right)\right)$, and
that pierces all the $\varepsilon$-narrow sets $K\in \K(P,\Pi,\eps,\sigma)$ with respect to $\V(P,s)$. (To this end, we stick with the same constant $\vartheta>0$ as in Theorem \ref{Theorem:MultipleSelection}.)

Applying the line-selection of Theorem \ref{Theorem:MultipleSelection} to $\V(P,s)$ yields a $\vartheta$-canonical family $\L=\L(P,s)$ of at most $s^{(d-1)^2}$ vertical lines. 
For each line $\ell\in \L(P,s)$ we construct the set $X_\ell:=\{\conv(\tau)\cap \ell\mid \tau\in \Pi\}$, which consists of at most $|\Pi|$ intersection points of $\ell$ with the simplices of $\Pi$, and then select a $1$-dimensional strong $\nu$-net $N_\ell$ over the set $X_\ell$ with respect to intervals, with $\nu=c\sigma^{\beta_{d-1}}\eps^d n^d/|\Pi|$. Here $c>0$ is a suitably small constant that may depend on $d$ (but not on $\eps$).
Lastly, we define $N':=(\bigcup_{\ell\in \L(P,s)}N_\ell)\cup N_\nar(P,s,\varepsilon)$. Since the cardinality of $N'$ is obviously bounded by (\ref{Eq:ThmSparseBound}), it suffices to show that every convex set $K\in \K(P,\Pi,\eps,\sigma)$ is pierced by $N'$. 

Indeed, if $K$ is $\varepsilon$-spread with respect to $\V(P,s)$, there must be a vertical canonical line $\ell\in \L(P,s)$ that pierces $\Omega\left(\sigma^{\beta_{d-1}}{\lceil \eps n\rceil\choose d}\right)$ simplices of $\Pi_K:=\Pi\cap {P_K\choose d}$, each time at some point of $X_\ell$ within the segment $K\cap \ell$. Therefore, a suitable choice of the constant $c>0$ guarantees
that $K$ is pierced by the respective net $\nu$-net $N_\ell\subset N'$.

On the other hand, if $K$ is $\varepsilon$-narrow with respect to $\V(P,s)$, then it must be pierced by the net $N_\nar(P,s,\varepsilon)$ of Lemma \ref{Lemma:Narrow}.
\end{proof}

\subsection{A useful lemma}

The proof of Theorem \ref{Theorem:MultipleSelection} relies on the following property. 

\begin{lemma}\label{Theorem:SecondSelection1}
Let $(V,E)$ be a $(d+1)$-uniform hypergraph so that $|V|=n$ and $|E|\geq h{n\choose d+1}$, and $\chi:V\rightarrow \reals^d$ be a mapping of $V$ to some point set $X=\chi(V)$ in general position in $\reals^d$.\footnote{Though the image $X=\chi(V)$ is in general position, the map $\chi$ is not necessarily injective. Thus, some simplices $f\in E$ could be mapped to lower-dimensional simplices $\chi(f)\in {X\choose \leq d}$.} 
Then there exists a point $x\in \reals^d$ that satisfies the following conditions:
\begin{enumerate}
	\item  There exist $\Omega\left(h^{\beta_d}{n\choose d+1}\right)$ hyperedges $f\in E$ so that $x\in \conv\left(\chi(f)\right)$.
	\item There exist $k\leq d$ pairwise disjoint subsets $A_1,\ldots,A_k\subseteq X$, each set $A_i$ of cardinality $|A_i|\leq d$ points, with the property that $\{x\}=\bigcap_{i=1}^k\conv(A_i)$.
\end{enumerate}
\end{lemma}

\begin{proof}
We apply an arbitrary small generic perturbation to the mapping $\chi: V\rightarrow \reals^d$. To this end, we fix a sufficiently small $\delta>0$, and replace the image $\chi(v)$ of every vertex $v\in V$ by a point $\chi'(v)$ that is chosen uniformly, independently, and at random from the $\delta$-neighborhood $B(\chi(v),\delta)$. 
Then, with probability $1$, the perturbed image $\chi'(V)$  is comprised of $n$ distinct points in general position. Thus, Theorem \ref{Theorem:SecondSelection} yields a point $x\in \reals^d$, with $\Omega\left(h^{\beta_d}{n\choose d+1}\right)$ edges $f\in E$, so that $x\in \conv(\chi'(f))$.
Due to an arbitrary small choice of $\delta>0$, the standard compactness argument yields such a point $x$ that meets condition (1) with respect to the original mapping $\chi$.

Let $E'$ denote the subset of all such edges $f\in E$ that satisfy $x\in \conv(\chi(f))$. Then it can be assumed, with no loss of generality, that $x$ is a vertex of the polytope $\bigcap_{f\in E'}\conv\left(\chi(f)\right)$. 
As such, it is an intersection of some $k$ boundary faces $\mu_1,\ldots,\mu_k$ of some $k$ distinct simplices $\tau_1,\ldots,\tau_k\in \{\conv(\chi(f))\mid f\in E'\}$ whose dimensions may vary between $0$ and $d-1$. Choosing the minimal such representation (which involves the smallest possible overall number of points) guarantees that the faces $\mu_1,\ldots,\mu_k$ have pairwise disjoint vertex sets, which we respectively denote by $A_1,\ldots,A_k$, so that $x=\bigcap_{i=1}^k \conv(A_i)$. 
\end{proof}




\subsection{Proof of Theorem \ref{Theorem:MultipleSelection}}
We fix an $n$-point set $P$ with an integer $s$ that satisfies $1\leq s\leq n$, which together determine the vertical simplicial $s$-partition $\V(P,s)$ of $P$.

\medskip
\noindent{\bf Constructing the canonical line family $\L(P,s)$.} We pick a generic proxy point $x_i\in \reals^{d-1}$ within every $(d-1)$-dimensional cell $\Delta_i$ of the vertical simplicial partition $\V(P,s)$. This yields the set $X:=\{x_i\mid 1\leq i\leq s\}$, whose $s$ points in general position will be used to generate a significantly larger point set $Y$ within $\reals^{d-1}$.

To this end, for every $1\leq k\leq d-1$, and every choice $A_1,\ldots,A_k$ of $k\leq d-1$ pairwise disjoint sets $A_1,\ldots,A_k\in {X\choose \leq (d-1)}$ that satisfy $\sum_{i=1}^k(d-|A_i|)=d-1$ and $\bigcap_{i=1}^d\conv(A_i)\neq \emptyset$, we add to $Y$ the unique point $y$ so that $\{y\}=\bigcap_{i=1}^k \conv(A_i)$; see Figure \ref{Fig:UniversalSetMap} (left). 

\begin{figure}
    \begin{center}      
        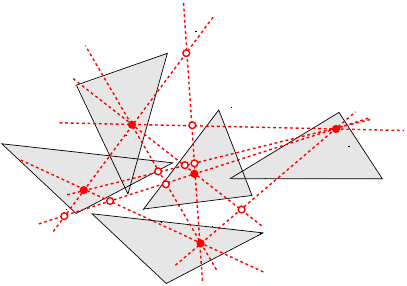\hspace{1cm}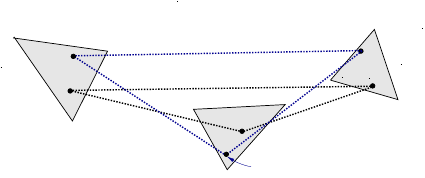
        \caption{\small Proof of Theorem \ref{Theorem:MultipleSelection} in dimension $d=3$. Left: The set $X=\{x_1,\ldots,x_r\}$ with (part of) the induced set $Y$ in $\reals^2$. Right: The triangle $\tau=\conv(p_1,p_2,p_3)$ in the vertical projection, and its representative $\chi(\tau)$ under the mapping $\chi: P_K\rightarrow \reals^2$.}
        \label{Fig:UniversalSetMap}
    \end{center}
\end{figure}

We then set $\L(P,s):=\{y^*\mid y\in Y\}$. Namely, $\L(P,s)$ is the set of all the vertical lines that result from lifting the points of $Y$. 

\medskip
\noindent {\bf Analysis.} By definition, we have $|\L(P,s)|=|Y|=O\left(s^{(d-1)^2}\right)$, so it remains to show that, given a suitably small constant $\vartheta>0$, the family $\L(P,s)$ is $\vartheta$-canonical with respect to $\V(P,s)$.
To this end, let us fix the parameters $0<\eps,\sigma\leq 1$. Let $K\in \K(P,\eps)$ be a convex set which is endowed with a principal set $P_K\subseteq P$ of $\lceil \eps n\rceil$ points, and a subset $\Pi_K\subseteq {P_K\choose d}$ of at least $\sigma{\lceil \eps n\rceil\choose d}$ $(d-1)$-simplices. 


To show that $\L(P,s)$ is $\vartheta$-canonical, with some $\vartheta>0$ that depends only on $d$, it is enough to show that at least one of the following statements holds with the previous choice of $K$, $P_K$, and $\Pi_K$: (i) there is a line $\ell\in \L(P,s)$ that pierces the relative interiors of at least $\Omega\left(\sigma^{\beta_{d-1}}{|P_K|\choose d}\right)$ simplices $\tau\in \{p_1,\ldots,p_d\}\in \Pi_K$ whose ambient cells $\Delta(p_i)$ surround $\ell$, or (ii) the set $K$ is $\Omega\left(\sigma^{\beta_{d-1}}\eps \right)$-narrow in $\V(P,s)$. To this end, it can be henceforth assumed that $|P_K|\geq d$ (or, else, the claim holds emptily with ${P_K\choose d}=\Pi_K=\emptyset$).

\medskip
\noindent {\it The mapping $\chi:P_K\rightarrow X$.} For each $1\leq i\leq s$, and each $p\in P_i\cap P_K$, we map every point $p\in P_K$ to the proxy point $\chi(p)=x_i$ within the ambient cell $\Delta_i=\Delta(p)\subseteq \reals^{d-1}$. See Figure \ref{Fig:UniversalSetMap} (right).
As a result, every simplex $\tau=\{p_1,\ldots,p_d\}\in \Pi_K$ maps to the set $\chi(\tau)=\{\chi(p_i)\mid 1\leq i\leq d\}\in {X\choose {\leq d}}$.
We say that this simplex $\tau\in \Pi_K$ is {\it crowded} if $|\chi(\tau)|<d$, and we say that $\tau$ is {\it split} otherwise.

\bigskip
\noindent{\bf Case 1.} At least $|\Pi_K|/3$ of the hyperedges in $\Pi_K$ are crowded. The pigeonhole principle yields a subset $\kappa=\{p_1,\ldots,p_{d-1}\}\in {P_K\choose d-1}$ with
at least $|\Pi_K|/\left(3{\lceil \eps n\rceil \choose d-1}\right)$ hyperedges $\tau=\{p_1,\ldots,p_d,p\}\in \Pi_K$ that satisfy $\chi(\tau)=\chi(\kappa)$. Up to relabeling, it can be assumed that at least $|\Pi_K|/\left(3d{\lceil \eps n\rceil \choose d-1}\right)=\Omega\left(\sigma \eps n\right)$ of the latter simplices $\tau$ satisfy $\chi(p_{d+1})=\chi(p_1)$. 
In other words, the ambient subset $P_i$ of $p_1$ in $\V(P,s)$ must encompass $\Omega(\sigma \eps n)$ points $p\in P_K$. Hence, the set $K$ must be $\Omega(\sigma\eps)$-narrow in $\V(P,s)$.

\medskip
\noindent {\bf Case 2.} Assume then that the previous scenario does not occur for $K$ and $\Pi_K$. Let $\Pi'_K$ denote the subset 
of all the split simplices in $\Pi_K$, which now must encompass at least $2|\Pi_K|/3$ simplices.
Applying Lemma \ref{Theorem:SecondSelection1} for the $d$-uniform hypergraph $(P_K,\Pi'_K)$, with the previous embedding function $\chi:P_K\rightarrow X\subseteq \reals^{d-1}$, yields a point $y\in Y$ and a subset $\Pi'_K(y)\subseteq \Pi'_K$ of $\Omega\left(\sigma^{\beta_{d-1}}{|P_K|\choose d}\right)$ simplices within $\Pi'_K$,
with the property that $y\in \conv\left(\chi(\tau)\right)$ for each $\tau\in \Pi'_K(y)$.

\medskip
Let $\Pi''_K(\ell)\subseteq \Pi'_K(y)$ denote the subset of all the $(d-1)$-simplices $\tau=\{p_1,\ldots,p_d\}$ in $\Pi'_K(y)$, so that $y$ is surrounded by the ambient cells $\Delta(p_1),\ldots,\Delta(p_d)\subset\reals^{d-1}$ of their $d$ vertices $p_1,\ldots,p_d$. 
 To complete the proof of Theorem \ref{Theorem:MultipleSelection}, we distinguish between two possible scenarios (see Figure \ref{Fig:MutipleCases}).

\medskip
\noindent{\bf Case 2a.} If $|\Pi''_K(y)|\geq |\Pi'_K(y)|/2=\Omega\left(\sigma^{\beta_{d-1}}{|P_K\choose d}\right)$ then $y^*$ constitutes the desired canonical line in $ \L(P,s)$ (with a suitably small choice of $\vartheta>0$). 

\begin{figure}
    \begin{center}      
        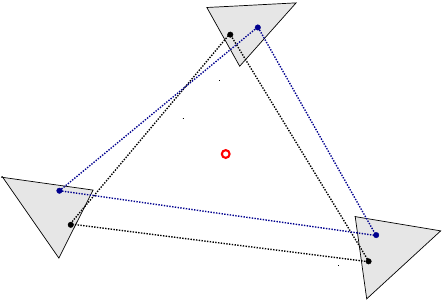\hspace{1cm}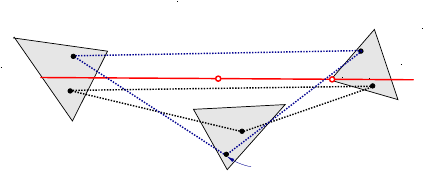
        \caption{\small Proof of Theorem \ref{Theorem:MultipleSelection} (in dimension $d=3$). Left: Case (2a) -- the simplex $\tau=\conv(p_1,p_2,p_3)$ (shown in the vertical projection) belongs to $\Pi''_K(y)$, so that $y$ is surrounded by $\Delta(p_1),\Delta(p_2)$, and $\Delta(p_3)$; in particular, $\tau$ is crossed by the vertical line $\ell=y^*$ over $y$. Right:  Case (2b)  -- the simplex $\tau=\conv(p_1,p_2,p_3)$ belongs to $\Pi'_K(y)\setminus \Pi''_K(y)$, so that $y$ is not surrounded by $\Delta(p_1),\Delta(p_2)$, and $\Delta(p_3)$. The line $g_\tau$ passes through $y$ and a vertex of $\Delta(p_1)$, and crosses $\Delta(p_2)$.}
        \label{Fig:MutipleCases}
    \end{center}
\end{figure}

\medskip
\noindent{\bf Case 2b.} Suppose, then, that the set $\Pi'_K(y)\setminus \Pi''_K(y)$ encompasses at least $|\Pi'_K(y)|/2=\Omega\left(\sigma^{\beta_{d-1}}{|P_K|\choose d}\right)$ simplices.
It suffices to show that such a convex convex set $K$ must be $\Omega\left(\sigma^{\beta_{d-1}}\eps\right)$-narrow in $\V(P,s)$.

To this end, we fix $\tau=\conv\left(p_1,\ldots,p_d\right)\in \Pi'_K(y)\setminus \Pi''_K(y)$. Notice that the point $y$ is {\it not} surrounded in $\reals^{d-1}$ by the $d$ ambient cells $\Delta(p_1),\ldots,\Delta(p_d)$ of the vertices of $\tau$, and yet it lies within $\conv\left(\bigcup_{i=1}^d \Delta(p_i)\right)$ (and, specifically, within the simplex $\conv(\chi(\tau))$, whose vertices $\chi(p_i)$ are chosen from the respective cells $\Delta(p_i)$). Hence, Lemma \ref{Prop:NotSurrounded} yields a $(d-2)$-plane $g$ through $y$, within $\reals^{d-1}$, that intersects some $d-1$ of the cells $\Delta(p_i)$, say $\Delta(p_1),\Delta(p_2),\ldots,\Delta(p_{d-1})$. Furthermore, applying Lemma \ref{Lemma:ExtremalHyperplane} with $\Sigma=\{\Delta(p_1),\ldots,\Delta(p_{d-1}),\{y\}\}$ yields such a $(d-2)$-plane $g=g_\tau$ that contains $y$ and, in addition, passes through some $d-2$ boundary vertices of $\Delta(p_1),\Delta(p_2),\ldots,\Delta(p_{d-1})$.

Notice that the supporting vertices of $g_\tau$ belong to at most $d-2$ simplices which, with no loss of generality, can be assumed to be amongst $\Delta(p_1),\ldots,\Delta(p_{d-2})$. Also note that, for a fixed $y\in Y$, the $(d-2)$-plane $g_\tau$ can be guessed from $p_1,\ldots,p_{d-2}$ in at most ${(d-2)d\choose d-2}$ different ways (from among the vertices of the $d-2$ ambient cells $\Delta(p_i)$). We, therefore, assign to $\tau$ this characteristic $(d-2)$-subset $S_\tau=\{p_1,\ldots,p_{d-2}\}\subset P_K$.

\smallskip
We repeat the above argument for each simplex $\tau\in \Pi'_K(y)\setminus \Pi''_K(y)$. Since the overall number of distinct characteristic subsets $S_\tau$, over all the simplices $\tau\in \Pi'_K(y)\setminus \Pi''_K(y)$, is at most ${\lceil \eps n \rceil\choose d-2}$, by the pigeonhole principle there must be such a $(d-2)$-subset $S\in {P_K\choose d-2}$ that has been assigned to $\Omega\left(\sigma^{\beta_{d-1}}{{\lceil \eps n\rceil}\choose 2}\right)$ of the $(d-1)$-simplices $\tau\in\Pi'_K(y)\setminus \Pi''_K(y)$ (which all satisfy $S_\tau=S$). For each of these simplices $\tau$, one or both of the two vertices in $\tau\setminus S$ must belong to such partition sets $P_i$ whose ambient cells $\Delta_i$ are crossed (in $\reals^{d-1}$) by the $(d-2)$-plane $g_\tau$. Since the $(d-2)$-plane $g_\tau$ can be ``guessed" from $S$ in at most ${d(d-2)\choose d-2}=O(1)$ ways, at least $\Omega\left(\sigma^{\beta_{d-1}}{{\lceil \eps n\rceil}\choose 2}\right)$ of these latter simplices $\tau$ must also ``share" the same $(d-2)$-plane $g=g_\tau$. Using that any point of $P_K$ belongs to at most $\lceil \eps n \rceil$ distinct pairs $\tau\setminus S$, with $\tau\in \Pi'_K(y)\setminus \Pi''_K(y)$, we conclude that the projections of at least $\Omega\left(\sigma^{\beta_{d-1}}\lceil \eps n\rceil\right)$ points in $P_K$ must fall, within the partition $\P_{d-1}(P^\perp,s)=\{(P_1^\perp,\Delta_1),\ldots,(P_s^\perp,\Delta_s)\}$, in the zone of the {\it same} $(d-2)$-plane $g$. In other words, the set $K$ must be $\Omega\left(\sigma^{\beta_{d-1}}\eps\right)$-narrow. $\Box$

\section{Proof of Theorem \ref{Thm:Main}}\label{Sec:MainRecurrence}


In this section we will use the machinery that has been assembled in the previous Sections \ref{Sec:Prelim} and \ref{Sec:MultipleSelection}, in order to establish the main result of this paper -- Theorem \ref{Thm:Main}. To this end, we fix an arbitrary constant $\gamma>0$, and derive a complete recurrence formula for the quantity $f(\eps)=f_d(\eps)$, which will ultimately solve to $f(\eps)=O\left(1/\eps^{\alpha_d+\gamma}\right)$. It can, furthermore, be assumed with no loss of generality that $\gamma<1/100$.

To facilitate the asymptotic analysis of the weak $\eps$-numbers $f(\eps)$, we let $\eta:=\gamma^{2d}/(100d)$, and adopt the following useful notation.
For $y>1$, we denote $x\lll_\eta y$ whenever $x=O\left(y^{\frac{\eta}{100d^4\beta_{d-1}}}\right)$, where $\beta_{d-1}\leq d^5(1+o(1))$ denotes the fixed selection exponent that Theorem \ref{Theorem:SecondSelection} yields in dimension $d-1$. For $0<y<1$, we accordingly denote $x\lll_\eta y$ whenever $1/y\lll_\eta 1/x$.

\subsection{A recurrence for $f(\eps,\rho,\sigma)$.} 
 
\noindent{\bf The setup.} As was mentioned in Section \ref{Subsec:RecursiveFramewk}, we first develop an efficient recurrence for the more general quantity $f(\eps,\rho,\sigma)$, with $0<\eps,\sigma,\rho\leq 1$. Recall that $f(\eps,\sigma,\rho)$ denotes the smallest possible number of points that would suffice to pierce {\it any} ``edge-constrained" family $
 \K=\K(P,\Pi,\eps,\sigma)$ which is determined by a finite point set $P$ in general position, and a subset $\Pi\subseteq {P\choose d}$ of at most $\rho{|P|\choose d}$ $(d-1)$-dimensional simplices.
 Unless stated otherwise, in what follows we keep using $\Pi_K$ to denote the sets ${P_K\choose d}\cap \Pi$ that are induced by the elements of $\K$. Hence, every family
 $\K=\K(P,\Pi,\eps,\sigma)$ under consideration is comprised of all such $\eps$-heavy convex sets $K\in \K(P,\eps)$ that satisfy $|\Pi_K|\geq \sigma{|P_K|\choose d}$.

In what follows, it can be assumed that $\eps>0$ is bounded from above by a sufficiently small absolute constant $\eps_0=\eps(\gamma,d)>0$, which is smaller than $1/100$.\footnote{This is done so as to guarantee that the arguments in all the recursive terms $a \cdot f(\eps\cdot b)$, with $b=b(\eps)\gg \log 1/\eps$, significantly increase with each invocation of the recurrence.}
Otherwise, if $\eps\geq \eps_0$, we can use, e.g., the bound $f(\eps,\rho,\sigma)\leq f(\eps)=O\left(1/\eps_0^{d+1}\right)=O(1)$ due to Alon {\it et al.} \cite{AlonSelections}.
 

\medskip
In addition to $\eps,\rho$ and $\sigma$, our recurrence formula for $f(\eps,\rho,\sigma)$ will involve five integer parameters $h,r,s,u$, and $t$, which depend on $\eps>0$ (and also on the constants $\gamma$ and $d$).
Specifically, $t$ is a fixed degree of $1/\eps$ which satisfies $\left\lceil 1/\eps^{d/2}\right\rceil\leq t\leq \left\lceil 1/\eps^{\alpha_d}\right\rceil$, and it will be
set in what follows to $(1/\eps)^{3d/4+o_d(1)}$ so as to ``optimize" the eventual recurrence for $f(\eps,\rho,\sigma)$.
The remaining four parameters, $h,r,s$ and $u$, are very small degrees of $1/\eps$ that satisfy

$$
h\lll_\eta r\lll_\eta s\lll_\eta u \lll_\eta 1/\eps. 
$$ 

\noindent To this end, we denote $\tilde{\eta}:=\eta/(100d^4\beta_{d-1})$, and set $h:=\left\lceil (1/\eps)^{\tilde{\eta}^4}\right\rceil$, $r:=\left\lceil (1/\eps)^{\tilde{\eta}^3}\right\rceil$, $s:=\left\lceil (1/\eps)^{\tilde{\eta}^2}\right\rceil$, and $u:=\left\lceil (1/\eps)^{\tilde{\eta}}\right\rceil$.

\medskip
As was explained in Section \ref{Subsec:RecursiveFramewk}, our recurrence launches with $f(\eps,1,1)=f_d(\eps)$, and has only constant depth in the second parameter $\rho$. 
Though the restriction threshold $\sigma>0$ may decrease in the course of this recurrence in $\rho$, for the rest of this paper it will be assumed that $\sigma$ is bounded from below by
a positive constant $\sigma_0$ that satisfies the inequality
$\sigma_0\leq 2^{-\lceil\log_h (1/\eps)\rceil}$ for any $\eps<1/100$. (The condition is clearly met by $\sigma_0\geq 2^{-cd^{40}/\eta^4}$, where $c>0$ denotes some absolute constant which does not depend on $\eta$ or $d$.)


\medskip
It can, furthermore, be assumed that $\rho>\eps$ for, otherwise, Theorem \ref{Theorem:Sparse} yields the estimate
\begin{equation}\label{Eq:SparseBound}
	f(\eps,\rho,\sigma)\leq f(\eps,\eps,\sigma)=O\left(\frac{s^{(d-1)^2}}{\sigma^{\beta_{d-1}}\eps^{d-1}}+s\cdot f\left(\sigma^{\beta_{d-1}}s^{1/(d-1)}\cdot \eps\right)\right)=
\end{equation}
$$
=O\left(\frac{s^{(d-1)^2}}{\eps^{d-1}}+s\cdot f\left(s^{1/(d-1)}\cdot \eps\right)\right).
$$

To facilitate our analysis, we also introduce two fractional parameters $0<\varepsilon\leq 1$ and $0<\delta\leq 1$, which too depend on $\eps$. The first parameter $\delta$ is a very small, yet positive, degree of $\eps$ which is given by

\begin{equation}\label{Eq:Narrow}
	\delta=\delta(\eps):=\frac{1}{r^{d(d+1)+1}}.
\end{equation}

\noindent The second parameter $\varepsilon>0$ is only marginally smaller than $\eps$, and satisfies

\begin{equation}\label{Eq:Spread}
	\varepsilon:=\frac{\vartheta^2\sigma_0^{\beta_{d-1}}\delta}{10^6\cdot 2^{\beta_{d-1}}d^2(d+1){(d-1)d\choose d-1}r}\cdot \eps=\Omega\left(\eps^{1+\eta}\right),
\end{equation}

\noindent where $\vartheta>0$ is the constant in the Canonical Line-Selection Theorem \ref{Theorem:MultipleSelection}.

\medskip
 To bound the quantity $f(\eps,\rho,\sigma)$ for  $\eps<\min\{\rho,\eps_0\}$, we fix a finite point set $P$ in a general position in $\reals^d$, with a subset $\Pi$ of at most $\rho {P\choose d}$ simplices within ${P\choose d}$, and construct a small-size net $N$ for the induced ``edge-constrained" sub-family $\K=\K(P,\Pi,\eps,\sigma)$ within $\K(P,\eps)$. The cardinality of $N$ will then serve as an upper bound on $f(\eps,\rho,\sigma)$. 
 

\smallskip 

In the rest of this proof, it can be assumed, with no loss of generality, that the ground point set $P$ is larger than a certain threshold $n_0(\eps)$ which satisfies
\begin{equation}\label{Eq:ManyPoints} 
n_0(\eps)=\max\left\{100\cdot s\cdot \lceil 1/\eps^{\alpha_d}\rceil,\frac{100ds}{\varepsilon}
\right\}\geq 100st.
\end{equation}

\noindent If this is not the case, then our net $N=P$ for $\K$ consists of $|P|\leq n_0(\eps)=O\left(1/\eps^{\alpha_d+\eta}\right)=O\left(1/\eps^{\alpha_d+\gamma}\right)$ points. In the sequel, it can be explicitly assumed that $|\Pi|\geq \eps {n\choose d}$ or, else, such a family $\K$ is again pierced by a net whose maximum cardinality $f(\eps,\eps,\sigma)$ has been estimated in (\ref{Eq:SparseBound}).





\bigskip
\noindent{\bf The vertical partition $\V(P,s)$, and the canonical line family $\L(P,s)$.} 
Since $|P|>n_0(\eps)\geq s$, the point set $P$ admits the vertical simplicial $s$-partition $\V(P,s)=\{(P_1,\Delta_1),\ldots,(P_s,\Delta_s)\}$, which has been described in Section \ref{Subsec:VerticalPartition}.
According to Theorem \ref{Theorem:MultipleSelection}, this partition $\V(P,s)$ is endowed with the $\vartheta$-canonical set $\L(P,s)$ of $O\left(s^{(d-1)^2}\right)$ lines. 


\medskip
\noindent{\bf Constructing the net $N$.} We say that a set $K\in \K(P,\eps)$ is {\it $\delta$-punctured} (resp., {\it $\delta$-hollow}) if its principal subset $P_K$ are $\delta$-punctured (resp., $\delta$-hollow) according to the definitions in Section \ref{Subsec:PrelimConfiguration}.

\medskip
Note that every convex sets $K$ in $\K=\K(P,\Pi,\eps,\sigma)$ belongs to one of the following three families.


\begin{enumerate}
	\item[(i)] The first family $\N_{\varepsilon}=\N_\varepsilon(P,\eps)$ is comprised of all such sets $K\in \K(P,\eps)$ that are $\varepsilon$-narrow with respect to the vertical $s$-partition $\V(P,s)$ of $P$.
	\item[(ii)] The second family $\K_{>\delta}=\K_{>\delta}(P,\eps)$ consists of all such $\delta$-punctured sets $K\in \K(P,\eps)$ that are $\varepsilon$-spread with respect to $\V(P,s)$.
	\item[(iii)] The third family $\K_{\leq \delta}=\K_{\leq \delta}(P,\Pi,\eps,\sigma)$ consists of all such $\delta$-hollow sets $K\in \K$ that are $\varepsilon$-spread with respect to $\V(P,s)$.
\end{enumerate}


In what follows, we obtain a separate net for each of these families.
Note that the first two of them, $\N_\varepsilon$ and $\K_{>\delta}$, may include additional $\eps$-heavy sets $K\in \K(P,\eps)$, which are not $(\eps,\sigma)$-restricted to $(P,\Pi)$. As a result, their respective nets will make no use of the subset $\Pi\subseteq {P\choose d}$.

\medskip
\noindent {\it 1. The net $N_\nar(P,s,\varepsilon)$ for $\N_\varepsilon$.}  Lemma \ref{Lemma:Narrow} yields a net $N_\nar(P,s,\varepsilon)$ whose cardinality satisfies
$$
|N_\nar(P,\eps,\varepsilon)|\leq s\cdot f\left(s^{1/(d-1)}\varepsilon\right)=s \cdot f\left(\eps \cdot \frac{s^{1/(d-1)}}{r^{d(d+1)+2}}\right)\ll s\cdot f\left(\eps \cdot s^{1/(d-1)-\eta}\right),
$$
\noindent and that pierces every convex set $K$ in $\N_\varepsilon$. 

\medskip

\medskip
\noindent{\it 2. The net $N_{>\delta}$ for $\K_{>\delta}$.} A small-size net $N_{>\delta}$ for the family $\K_{>\delta}$ is provided by the following statement, whose proof is postponed to Section \ref{Sec:VerticallyConvex}.

\begin{theorem}\label{Theorem:BoundPunct}
Let $d\geq 3$. Then one can fix a constant $d/2\leq \theta_d\leq \alpha_d$ so that, with the previous
 choice of the constant $\eta>0$, $0<\eps\leq 1$, $\delta=\delta(\eps)$, and the finite point set $P$ in $\reals^d$, the induced family $\K_{>\delta}=\K_{>\delta}(P,\eps)$ within $\K(P,\eps)$ can be pierced by a net $N_{>\delta}$ whose cardinality satisfies

$$
|N_{>\delta}|\leq \s^{1+\eta}\cdot f\left(\eps\cdot \s^{(1/\alpha_d)-\eta}\right)+{\u}^{1+\eta}\cdot f\left({\u}^{\frac{1}{d-1}-\eta}\cdot \eps\right)+O\left(1/\eps^{\alpha_d+\eta}\right).
$$

\noindent Here $t=\left\lceil (1/\eps)^{\theta_d}\right\rceil$, whereas $s$ and $u$ denote the previously selected (and small) degrees of $1/\eps$ which satisfy $s\lll_\eta u\lll_\eta 1/\eps$. 

\end{theorem}

 Let us spell out the key details of the rather involved proof of Theorem \ref{Theorem:BoundPunct}.
In view of Theorem \ref{Theorem:Sparse}, the key to an efficient recurrence lies in replacing the set ${P\choose d}$ with a {sparser} subset $\tilde{\Pi}$ that would still be sufficiently dense over all the principal subsets $P_K$ of the convex sets $K\in \K$ under our consideration. \footnote{Applying Theorem \ref{Theorem:Sparse} directly to $\K(P,\Pi,\eps,\sigma)$ would yield a sub-$(1/\eps^d)$-estimate only if the hypergraph $(P,\Pi)$ is already sufficiently sparse. Also note that the desired sparser subset $\tilde{\Pi}$ of ${P\choose d}$ need not necessary be a subset of the collection $\Pi$ which figures in the initial definition of our family $\K=\K(P,\Pi,\eps,\sigma)$. } 
 However, as was explained in Section \ref{Subsec:Overview}, such a subset of ${P\choose d}$ may not exist in absence of further restrictions on the {\it shape} of the sets in $\K$.

Specializing to the $\delta$-punctured sets $K\in \K_{>\delta}$, our case analysis will be facilitated by secondary {$d$-dimensional} simplicial partitions $\P_i=\P(P_i,t)$, which will be obtained via Theorem \ref{Theorem:Simplicial} for each part $P_i$ in the vertical partition $\V(P,s)$ of $P$. The fact that $s\lll_\eta t$ will guarantee that the overall partition $\P:=\bigcup_{i=1}^s \P_i$ is comprised of $O(s\cdot t)=O^*(t)$ parts, while each hyperplane crosses $O\left(st^{1-1/d}\right)=O^*\left(t^{1-1/d}\right)$ cells.
Specializing to the sets $K\in \K_{>\delta}$ that are flat with respect to $\P$ in the sense that was outlined in Section \ref{Subsec:Overview}, we will construct a family $\tilde{\Pi}\subseteq {P\choose d}$, of cardinality $O^*\left(\eps^{d-\alpha_d}n^{d}\right)$, and that contains a large fraction of every set ${P_K\choose d}$ that is induced by a flat convex set $K$. As a result, the flat convex sets will be pierced using Theorem \ref{Theorem:Sparse}.
All the ``left-out" $\delta$-punctured convex sets $K\in \K_{>\delta}$, which are either ``thick" with respect to the partition $\P$, or the points of their principal subsets $P_K$ are not in a sufficiently convex position, will dispatched via a combination of more elementary nets.

%



\medskip
\noindent {\it 3. The net $N_{\leq \delta}$ for $\K_{\leq \delta}$.}   Specializing to the remaining family $\K_{\leq \delta}$, which is comprised of all the $\delta$-hollow sets {within the ``edge-constrained" family $\K=\K(P,\Pi,\eps,\sigma)$}, in Section \ref{Sec:Surrounded} we obtain a small-size net whose properties are summarized in the following theorem.

\begin{theorem}\label{Theorem:BoundHollow} 
Let $d\geq 3$. Then with the previous choice of the constants $\eta$ and $\sigma_0$, of the fractions $0<\eps\leq 1$, $2\sigma_0\leq \sigma\leq 1$ and $\delta=1/r^{d(d+1)+1}$, and of the finite point set $P$ in $\reals^d$, along with $\Pi\subseteq {P\choose d}$, the induced family $\K_{\leq \delta}$
within $\K=\K(P,\Pi,\eps,\sigma)$ can be pierced by a net $N_{\leq \delta}$ whose cardinality satisfies

$$
|N_{\leq \delta}|\leq f(\eps,\rho/h,\sigma/2)+f\left(\eps\log r\right)+
$$

$$
+O\left(r^{\alpha_d+\eta}\cdot  f\left(\eps\cdot r^{1-\eta}\right)+\sum_{i=1}^l 2^iw^{\alpha_d+\eta}\cdot f\left(\eps \cdot 2^i w\right)+\frac{1}{\eps^{\alpha_d+\eta}}\right),
$$

\noindent Here $r$ and $s$ denote the previously selected and small degrees of $1/\eps$ which satisfy $r\lll_\eta s\lll_\eta 1/\eps$, $w=r^{\sqrt{d^2-2d}-d+2}$, and $l$ is a positive integer that satisfies $l=O(\log r)$. 

\end{theorem}

The case analysis in the proof Theorem \ref{Theorem:BoundHollow} will overly resemble that of Theorem \ref{Theorem:BoundPunct}, except for replacing Matou\v{s}ek's simplicial partitions $\P_i$ with a $\Theta\left(\frac{\log r}{r}\right)$-cutting $\D$ of the family $\HH=\HH(\Pi)$, which consists of all the hyperplanes that support the simplices of $\Pi$. As was explained in Section \ref{Subsec:Arrangements}, such a cutting can be obtained by triangulating the random arrangement $\A(\R)$ which is built over some $r$-sample $\R\subseteq \HH$.

Note, though, that removing the relatively ``thick" convex sets $K\in \K_{\leq \delta}$ is no longer enough to attain an even a weaker variant of flatness for the $\delta$-hollow sets $K$ at hand. Instead, we will revert to a simple recurrence in $P$ (and in $\eps$) that outlined in the beginning of Section \ref{Subsec:RecursiveFramewk}. An improved ``$O^*(1/\eps^{\alpha_d})$-type" recurrence for this special case will then stem from the following property: every sufficiently ``slim" set $K\in \K_{\leq \delta}$ can be assigned a so called {\it principal hyperplane} $H_K\in \HH(\Pi)$ whose zone within our cutting of $\HH(\Pi)$ encompasses a fixed fraction of the points of $P_K$. Together with the $\delta$-hollowness of $P_K$, this will guarantee that $\Omega(\eps n)$ of the points in $P_K$ must lie in the vicinity of the relative boundary of some {$(d-1)$-dimensional} convex polytope $T_K\subset H_K$. As a result, their distribution within $\A(\R)$ is close to {\it $(d-2)$-dimensional}. Therefore, such sets $K\in \K_{\leq \delta}$ can be relegated to recursive sub-problems whose ground subsets are cut out by the cells of the random arrangement $\A(\R)$.

\medskip
The complete net $N$ for the family $\K=\K(P,\Pi,\eps,\sigma)$ is, therefore, given by the combination

\begin{equation}
N:=N_\nar(P,s,\varepsilon)\cup N_{>\delta}\cup N_{\leq \delta},
\end{equation}

\noindent whose cardinality satisfies

\begin{equation}\label{Eq:NetDensity}
	|N|\leq f(\eps,\rho/h,\sigma/2)+F(\eps),
\end{equation}

\noindent with

$$
	F(\eps)=O\left(f(\eps\cdot \log r)+\r\cdot f\left(\eps \cdot \r^{\frac{1}{d-1}-\eta}\right)+
\s^{1+\eta}\cdot f\left(\eps\cdot \s^{(1/\alpha_d)-\eta}\right)+{\u}^{1+\eta}\cdot f\left({\u}^{\frac{1}{d-1}-\eta}\cdot \eps\right)\right)+
$$

$$
+O\left(r^{\alpha_d+\eta}\cdot  f\left(\eps\cdot r^{1-\eta}\right)+\sum_{i=1}^l 2^iw^{\alpha_d+\eta}\cdot f\left(\eps \cdot 2^i w\right)+1/\eps^{\alpha_d+\eta}\right).
$$

\noindent Here $t=\left\lceil 1/\eps^{\theta_d}\right\rceil$ is an integer that depends only on $\eps$, and satisfies $\left\lceil1/\eps^{d/2}\right\rceil\leq t\leq \left\lceil 1/\eps^{\alpha_d}\right\rceil$, $w=r^{\sqrt{d^2-2d}-d+2}$, and $l$ is a positive integer that satisfies $l=O(\log r)$. As was previously mentioned, the same upper bound must also hold for the quantity $f(\eps,\rho,\sigma)$, which thus satisfies

\begin{equation}\label{Eq:RecurseDensity}
f(\eps,\rho,\sigma)\leq f(\eps,\rho/h,\sigma/2)+F(\eps).
\end{equation}

\subsection{Solving the recurrence}

 To obtain a complete recurrence formula for $f(\eps)$, we begin with $f(\eps)=f(\eps,1,1)$ and invoke the bound (\ref{Eq:RecurseDensity}) $j_0=j_0(\eps,\eta):=\lceil \log_h(1/\eps)\rceil=O\left(d^{40}/\eta^4\right)$ times, each time only to the first term of the general form $f(\eps,\rho',\sigma')$, while leaving all the remaining recursive terms $f(\eps')$ intact. Since each iteration reduces the density $\rho$ by a factor of $h$, and $\sigma$ by a factor of $2$, this yields

\begin{equation}\label{Eq:IterateDensity}
	f(\eps)\leq f\left(\eps,1/h^{j_0},\sigma_0\right)+\sum_{i=0}^{j_0-1}F\left(\eps\right),
\end{equation}

\noindent where $\sigma_0>0$ denotes the aforementioned constant $2^{-cd^{12}/\eta^4}$ which satisfies $\sigma_0\leq 1/2^{j_0}$. Applying Theorem \ref{Theorem:Sparse} to the first term $f\left(\eps,1/h^{j_0},\sigma_0\right)$ of (\ref{Eq:IterateDensity}), with $\rho\leq 1/h^{j_0}\leq \eps$, yields 

\begin{equation}\label{Eq:FinalRecurrence}
	f(\eps)\leq j_0\cdot F(\eps)+O\left(\frac{s^{(d-1)^2}}{\eps^{d-1}}+s\cdot f\left(s^{1/(d-1)}\cdot \eps\right)\right)=
\end{equation}

$$
=O\left(f(\eps\cdot \log r)+\r\cdot f\left(\eps \cdot \r^{\frac{1}{d-1}-\eta}\right)+
\s^{1+\eta}\cdot f\left(\eps\cdot \s^{(1/\alpha_d)-\eta}\right)+{\u}^{1+\eta}\cdot f\left({\u}^{\frac{1}{d-1}-\eta}\cdot \eps\right)\right)
$$

$$
+O\left(r^{\alpha_d+\eta}\cdot  f\left(\eps\cdot r^{1-\eta}\right)+\sum_{i=1}^l 2^iw^{\alpha_d+\eta}\cdot f\left(\eps \cdot 2^i w\right)+1/\eps^{\alpha_d+\eta}\right),
$$

\noindent with implicit constants of proportionality that depend only on $\eta=\eta(\gamma)>0$, and the dimension $d\geq 3$.

Recall that the recurrence in $\eps>0$ bottoms out as soon as $\eps$ by-passes a certain upper threshold $\eps_0>0$. 
Since $\eta$ is much smaller than $\gamma$, fixing a sufficiently small constant threshold $\eps_0=\eps_0(d,\gamma)$, and following the standard induction argument which applies to all recurrences of this type (see, e.g., \cite{Envelopes3D,FOCS18,EnvelopesHigh} and \cite[Section 7.3.2]{SA}), yields

\begin{equation}\label{Eq:FinalFinal}
	f(\eps)\leq \frac{a_0}{\eps^{\alpha_d+\gamma}},
\end{equation}

\noindent where $a_0\geq f(\eps_0)$ is a suitably large constant which too depends on $d$ and $\gamma>0$.

For the sake of completeness, let us spell out some key details of this generic induction. We first fix a small enough auxiliary threshold $\eps'_0=\eps'_0(d,\gamma)>0$ so that all the recursive terms $f(\eps')$ in the right hand side of (\ref{Eq:FinalRecurrence}) involve arguments that satisfy $\eps'>2\eps$ whenever $\eps\leq \eps'_0$.\footnote{Since we have that $f(\eps)\leq f(\eps')\leq f(\eps/2)$ for all pairs $0<\eps',\eps\leq 1$ that satisfy $0<\eps/2\leq \eps'\leq \eps\leq 1$, condition (ii) guarantees that the value of $\eps$ increases with each invocation of (\ref{Eq:FinalRecurrence}) and, furthermore, the induction in $\eps$ can be restricted to $\eps=1/2^j$, with $j\in {\mathbb N}$.} Since the inequality (\ref{Eq:FinalFinal}) trivially holds for $\eps_0<\eps\leq 1$ whenever $a_0\geq f(\eps_0)$, it suffices to choose the constant threshold $\eps_0\leq \eps'_0(d,\gamma)$ that would facilitate the induction step for all the smaller values $\eps<\eps_0$. To this end, plugging the desired induction assumption (\ref{Eq:FinalFinal}), with so far unknown $a_0$, into (\ref{Eq:FinalRecurrence}), we obtain

$$
f(\eps)\leq \frac{a_0}{\eps^{\alpha_d+\gamma}}	\cdot J(\eps),
$$

\noindent where
$$
J(\eps)\leq a_1\cdot \frac{1}{\log^{\alpha_d+\gamma} r}+a_2\cdot \frac{s}{s^{(\alpha_d+\gamma)\left(\frac{1}{d-1}-\eta\right)}}+a_3\cdot \frac{t^{1+\eta}}{t^{(\alpha_d+\gamma)\left(\frac{1}{\alpha_d}-\eta\right)}}+
$$

$$
+a_4\cdot\frac{u^{1+\eta}}{u^{(\alpha_d+\gamma)\left(\frac{1}{d-1}-\eta\right)}}+a_5\cdot\frac{r^{\alpha_d+\eta}}{r^{(\alpha_d+\gamma)(1-\eta)}}+a_6\cdot\sum_{i=1}^l\frac{2^i w^{\alpha_d+\eta}}{2^{i\cdot (\alpha_d+\gamma)}w^{\alpha_d+\gamma}}+a_7\cdot \eps^{\gamma-\eta},
$$

\noindent and $a_1,\ldots,a_7$ are positive constants that do not depend on the choice of $\eps_0\leq \eps'_0(d,\gamma)$ or $a_0$. Using that (i) $s,t,u,r$ and $w$ are fixed a positive degrees of $1/\eps$, (ii) $\gamma<1/100$,  and (iii) $\eta\leq \gamma^{2d}/100$, a sufficiently small choice of $0<\eps_0\leq \eps'$ yields the inequality $J(\eps)\leq 1$ for all $\eps\leq \eps_0$.
Hence, the bound (\ref{Eq:FinalFinal}) holds, with $a_0=f(\eps_0)$, for all $0<\eps\leq 1$, which concludes the proof of Theorem \ref{Thm:Main}. $\Box$

\section{Piercing the $\delta$-punctured sets -- proof of Theorem \ref{Theorem:BoundPunct}}\label{Sec:Surrounded}
We stick with $0<\eps\leq 1$, the auxiliary integer parameters $r,s,u$, which have been defined in Section \ref{Sec:MainRecurrence} as very small yet positive degrees of $1/\eps$ which satisfy $r\lll_\eta s\lll_\eta u\lll_\eta 1/\eps$, and the $n$-point set $P$ within $\reals^d$, which satisfies $n=|P|\geq n_0(\eps)\geq 100s\lceil 1/\eps^{\alpha_d}\rceil\geq  100st$.  
The proof of Theorem \ref{Theorem:BoundPunct} will be facilitated by the underlying vertical simplicial partition $\V(P,s)$, and the $\vartheta$-canonical family $\L(P,s)$ of $O\left(s^{(d-1)^2}\right)$ vertical lines, whose choice is too described in Section \ref{Sec:MainRecurrence}.

To establish Theorem \ref{Theorem:BoundPunct}, we fix the exponent $d/2\leq \theta_d\leq \alpha_d$, along with the integer parameter $t=\lceil(1/\eps)^{\theta_d}\rceil$, which will be selected in the end of this section, and describe a small-size net $N_{>\delta}$ for the family $\K_{>\delta}$, which is comprised of all such $\delta$-punctured sets $K\in \K(P,\eps)$ that are $\varepsilon$-spread with respect to $\V(P,s)=\{(P_1,\Delta_1),\ldots,(P_s,\Delta_s)\}$. 
The construction of $N_{>\delta}$ begins with an empty set, and proceeds through four steps.

In the preparatory ``$0$-step", we establish the following crucial property of the $\delta$-punctured sets $K\in \K_{>\delta}$: each of them can be assigned an ordered $(d+1)$-tuple $\Psi=(P_{j_0},\ldots,P_{j_{d}})$ of subsets within $\V(P,s)$, so that (i) $P_{j_0}$ is surrounded by the remaining $d$ subsets $P_{j_i}$, with $1\leq i\leq d$, and (ii) each subset $P_{j_i}$ in $\Psi$ contains $\Omega(\delta\eps n/s)\gg \eps^{1+\eta}n$ points of $P_K$. 
In what follows, every convex set $K\in \K_{>\delta}$ will be assigned a unique such $(d+1)$-tuple $\Psi(K)=(P_{j_0},\ldots,P_{j_{d}})$, which will be referred to as the {\it principal $(d+1)$-tuple of $K$}.

In each of the subsequent three steps $1\leq i\leq 3$, we will define, and immediately add to $N_{>\delta}$, a new auxiliary net $N_{>\delta,i}$. Once a set $K\in \K_{>\delta}$ is pierced by $N_{>\delta,i}$, it is immediately removed from further consideration, whereas the surviving sets of $\K_{>\delta}$ are passed on to the subsequent steps. 
Upon the completion of the second step, every remaining set $K\in \K_{>\delta}$ will be sufficiently flat in the sense that was outlined in Section \ref{Subsec:Overview}. With these preparations, in the third step we will at last construct a sparse subset $\tilde{\Pi}\subseteq {P\choose d}$ that is sufficiently dense over every remaining $\lceil \eps n\rceil$-subset $P_K$ that is determined by some $K\in \K_{>\delta}$, so the third net $N_{>\delta,3}$ will be obtained through a straightforward application of Theorem \ref{Theorem:Sparse}.



\subsection{Step 0: Fixing a principal $(d+1)$-tuple}

\begin{lemma} \label{Lemma:CharacteristicTuple} 


Let $K\in \K_{>\delta}$. Then the partition $\V(P,s)=\{(P_1,\Delta_1),\ldots,(P_s,\Delta_s)\}$ contains an ordered $(d+1)$-tuple $\Psi=\left(P_{j_0},\ldots,P_{j_{d}}\right)$ with the following properties (see Figure \ref{Fig:Characteristic} (left)):
\begin{itemize}
\item[(P1)] The ambient cell $\Delta_{j_{0}}$ of $P_{j_{0}}$ is surrounded by the respective ambient cells $\Delta_{j_i}$ of the remaining $d$ sets $P_{j_i}$, for $1\leq i\leq d$. 

\item[(P2)] For all $0\leq i\leq d$, we have that 
$$
|P_K\cap P_{j_i}|\geq \frac{\delta|P_K|}{(4d+4)s}
$$

\end{itemize}
\end{lemma}


\begin{proof} 
Let us begin with introducing a few handy notations. Recall that, according to our definition in Section \ref{Subsec:PointConfig}, a collection of $k\geq d$ $(d-1)$-dimensional simplicial cells $\{\Delta_{i_1},\ldots,\Delta_{i_k}\}$ in $\V(P,s)$, for $1\leq j_1\neq j_2\neq \ldots\neq j_k\leq \r$, is {\it separated} if no $d$ of its members can be crossed by a single $(d-2)$-plane in $\reals^{d-1}$. 

 We then say, for $k\geq d$, that a $k$-size subset of points $\{p_1,\ldots,p_k\}\in {P_K\choose k}$ is {\it good} if their ambient cells $\Delta(p_i)$ in the vertical partition $\V(P,s)$ are all distinct, and form a separated collection $\{\Delta(p_1),\ldots,\Delta(p_k)\}$. In any other case, the $k$-size subset $\{p_1,\ldots,p_k\}$ will be called {\it bad}.

 For the rest of this proof, it can be assumed that $\lceil \eps n\rceil\geq 2d+2$ for, otherwise, there are no $\varepsilon$-spread convex sets in $\K_{>\delta}$, and there is nothing to prove. 
 
 \medskip
 Our proof of Lemma \ref{Lemma:CharacteristicTuple} relies on the following property.


\begin{figure}
    \begin{center}      
        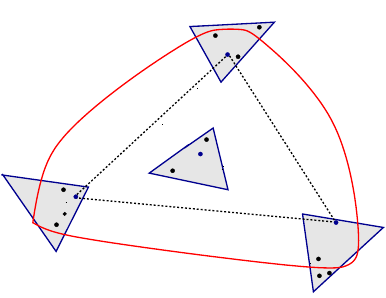\hspace{2cm}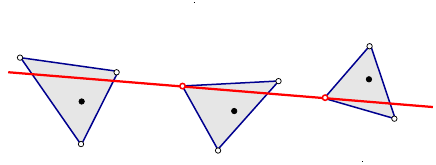
        \caption{\small Left: Lemma \ref{Lemma:CharacteristicTuple} in dimension $d=3$. The $4$-tuple $\Psi=\left(P_{j_0},P_{j_1},P_{j_2},P_{j_3}\right)$ is depicted in the vertical projection. We have that $|P_K\cap P_{j_i}|\geq \frac{\delta|P_K|}{(4d+4)s}$. The cell $\Delta_{j_0}$ is surrounded by the cells $\Delta_{j_1},\Delta_{j_2}$, and $\Delta_{j_3}$: for each selection $x_i\in \Delta_{j_i}$ we have that $x_0\in \conv(x_1,x_2,x_3)$. Right: Proof of Lemma \ref{Lemma:Regularity} in $\reals^3$. The depicted bad triple $A=\{p_1,p_2,p_3\}$ within $P_K$ is labeled with $(B,g)$, where $B=\{p_2,p_3\}$, and the 1-plane (i.e., line) $g$ passes through a pair of vertices of $\Delta(p_1)$ and $\Delta(p_2)$, and crosses $\Delta(p_3)$.}
        \label{Fig:Characteristic}
    \end{center}
\end{figure}

\begin{lemma}\label{Lemma:Regularity}
There exist at most $(\delta/(4(d+1)){|P_K|\choose d}$ bad $d$-size subsets in ${P_K\choose d}$.
Hence, there exist at most $(\delta/4){|P_K|\choose d+1}$ bad $(d+1)$-subsets in ${P_K\choose d+1}$.
\end{lemma}
\begin{proof}[Proof Lemma \ref{Lemma:Regularity}.]
Clearly, the second property follows from the first one, as any bad $(d+1)$-tuple in ${P_K\choose d+1}$ of points can be ``charged" to a bad $d$-tuple in ${P_K\choose d}$, whereas any bad $d$-tuple is part of at most $|P_K|-d$ bad $(d+1)$-tuples, for a total of at most
$$
\left(|P_K|-d\right)\left(\frac{\delta}{4(d+1)}\right)\frac{|P_K|\cdot(|P_K|-1)\ldots(|P_K|-d+1)}{d!}
$$

\noindent such bad $(d+1)$-tuples in ${P_K\choose d+1}$.

\smallskip
Assume for a contradiction that there exist more than $\frac{\delta}{4(d+1)}{P_K\choose d}$ bad subsets in ${P_K\choose d}$. Let $\B_K\subseteq {P_K\choose d}$ denote the collection of all such bad $d$-size subsets in ${P_K\choose d}$. Fix any $A\in \B_K$.
Then there must be a $(d-2)$-plane $g$ within $\reals^{d-1}$ that crosses every simplicial cell in the family $\Gamma_A:=\{\Delta(p_1),\ldots,\Delta(p_d)\}$, whose cardinality may be smaller than $d$. In particular, if $|\Gamma_A|\leq d-1$, so that $\Delta(p_j)=\Delta(p_{j'})$ for some $1\leq j\neq j'\leq d$, then {\it every} $(d-2)$-plane $g$ through $A^\perp\setminus \{p^\perp_i\}$ meets every simplex in $\Sigma$.

Applying Lemma \ref{Lemma:ExtremalHyperplane} to this family $\Gamma_A$, which is comprised of at least one $(d-1)$-dimensional simplicial cell in $\reals^{d-1}$, yields a $(d-1)$-subset $B\in {A\choose d-1}$,  and a $(d-2)$-plane transversal $g$ to $\Gamma_A$ that passes through some $d-1$ boundary vertices of the cells in the family $\Gamma_B=\{\Delta(p)\mid p\in B\}$. See Figure \ref{Fig:Characteristic} (right).
We thus label $A$ with such a pair $(B,g)$, and repeat the above argument to all the (at least) 
$\frac{\delta}{4(d+1)}{\lceil\eps n\rceil\choose d}$ such
bad $d$-size sets $A\in \B_K$. Since any $(d-1)$-set $B\in {P_K\choose d-1}$ participates in at most ${(d-1)d\choose d-1}$ of the labels $(B,g)$, by the pigeonhole principle, and the prior choice (\ref{Eq:Narrow}) of the threshold $\varepsilon$, there must be such a label $(B,g)$ that is shared by at least 
$$
\frac{\delta{\lceil \eps n\rceil \choose d}}{4(d+1){(d-1)d\choose d-1}{\lceil \eps n\rceil \choose d-1}} \geq \varepsilon n
$$ 
of the degenerate $d$-size sets $A\in \D_K$. The desired contradiction follows because each of these degenerate $d$-sets $A$ contributes a distinct point $p_A\in A\setminus B$ whose ambient cell $\Delta\left(p_A\right)$ is crossed by $H$, for a  total of at least $\varepsilon n$ of such points. Hence, $K$ must be $\varepsilon$-narrow in $\V(P,s)$ and, thereby, have previously been removed from consideration.
\end{proof}

\noindent {\it Back to the proof of Lemma \ref{Lemma:CharacteristicTuple}.} Since $K$ is $\delta$-punctured, the family ${P_K\choose d+1}$ must contain a sub-family $\E_K$ of at least $\delta{|P_K|\choose d+1}$ punctured $(d+1)$-subsets $A$.
By the previous Lemma \ref{Lemma:Regularity}, at most $(\delta/4){|P_K|\choose d+1}$ of these $(d+1)$-subsets in $\E_K$ can be bad. 

\medskip
Let
$$
\V':=\left\{P_i\mid 1\leq i\leq \r, |P_i\cap P_K|\leq \frac{\delta}{4\r(d+1)} |P_K|\right\}.
$$

\noindent Since the sets of $\V'$ encompass at total of at most $\delta |P_K|/(4d+4)$ points of $P_K$,
at most $(\delta/4){|P_K|\choose d+1}$ of the $(d+1)$-size sets $A\in \E_K$ can include a point of $P_K\cap \left(\bigcup \V'\right)$.

In particular, there must remain at least one good $(d+1)$-set $A=\{p_0,\ldots,p_{d}\}$ in $\E_K$, so that $A\subset \bigcup\left(\V(P,s)\setminus\V'\right)$.\footnote{With some abuse of notation, $\V(P,s)$ is treated as the collection $\{P_1,\ldots,P_\r\}$ of its parts.} Assume with no loss of generality that $p_{0}^\perp\in \conv\left(p_1^\perp,\ldots,p_d^\perp\right)$. Since the simplices $\Delta(p_0),\ldots,\Delta(p_d)$ are separated in $\reals^{d-1}$, and $\Delta(p_0)$ clearly intersects $\conv\left(\bigcup_{i=1}^d \Delta(p_i)\right)$ at the point $p_0$, it follows via Corollary \ref{Corol:Surrounded} that $\Delta\left(p_{0}\right)$ must be surrounded by the cells $\Delta\left(p_{i}\right)$, for $1\leq i\leq d$. 
Furthermore, as $A$ is contained in $\bigcup\P(s)\setminus \V'$, the respective $d+1$ subsets $P_{j_i}\ni p_i$ of $\V(P,s)$ are distinct and satisfy $|P_{j_i}|\geq \delta |P_K|/((4d+4)\r)$ for all $0\leq i\leq d$. Hence, the $(d+1)$-tuple $(P_{j_0},\ldots,P_{j_{d}})$ meets both criteria (P1) and (P2), which completes the proof of Lemma \ref{Lemma:CharacteristicTuple}. 
\end{proof}

\noindent{\bf Definition.} We assign to each $K\in \K_{>\delta}$ a {\it unique} ordered $(d+1)$-tuple $\Psi(K)=(P_{j_0},\ldots,P_{j_d})$ of subsets in $\V(P,s)$ that meet the criteria (P1) and (P2), and refer to $\Psi(K)$ as the {\it principal $(d+1)$-tuple} of $K$ in $\V(P,s)$.

\subsection{Step 1}\label{Subsec:MainPartition}\label{Subsec:NetSurrounded}

To facilitate the subsequent divide-and-conquer treatment of the family $\K_{>\delta}$, let us first refine the underlying vertical simplicial $s$-partition $\V(P,s)=\{(P_1,\Delta_1),\ldots,(P_s,\Delta_s)\}$ of $P$ by the means of two secondary (and mutually independent) sub-partitions $\P(P,s,t)$ and $\W(P,s,u)$.  
Note that, in the view of the previous assumption that $n\geq n_0(\eps)$, each non-empty part $P_i$ in $\V(P,s)$ satisfies $|P_i|\geq \lceil |P|/s\rceil\geq 10t\ggg_\eta u$.

\medskip
\noindent{\bf The secondary simplicial partition $\P(P,s,t)$.} 
 We invoke the simplicial partition of Theorem \ref{Theorem:Simplicial} for each part $P_i\in \V(P,s)$ so as to obtain the secondary $d$-dimensional simplicial $t$-partition $\P\left(P_i,t\right)=\{(P_{i,j},\varphi_{i,j})\mid 1\leq j\leq \s\}$. See Figure \ref{Fig:Secondary} (left). We then denote
$$
\P(P,s,t):=\bigcup_{i=1}^s \P\left(P_i,t\right)=\{\left(P_{i,j},\varphi_{i,j}\right)\mid 1\leq i\leq s,1\leq j\leq t\}.
$$ 

For each point $p\in P_i$, its {\it ambient cell $\varphi(p)$ in $\V(P,s,t)$} is defined as the ambient cell of $p$ in the secondary simplicial partition $\P(P_i,s)$ of $P_i$.

\begin{figure}
    \begin{center}
    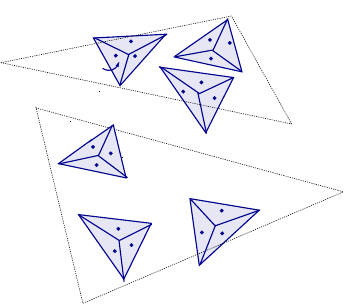 \hspace{2cm}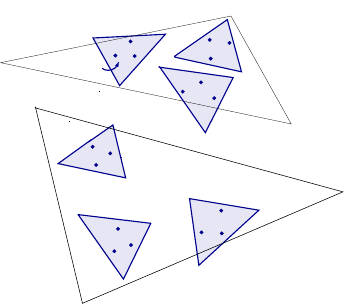       \caption{\small Left: The secondary partition $\P(P,s,t)$ of $P$ in $\reals^3$ -- view from above. 
    Each subset $P_i$ in $\V(P,s)$ is subdivided into pairwise disjoint parts $P_{i,j}$ of size at most $2\lceil|P_i|/t\rceil$, each of them enclosed by a cell $\varphi_{i,j}$ in $\reals^3$. Right: The secondary vertical partition $\W(P,s,u)$, which is obtained by constructing the vertical simplicial partition $\V(P_i,u)$ for each part $P_i$.}
      \label{Fig:Secondary}
    \end{center}
\end{figure}

\medskip
\noindent{\bf The secondary vertical partition $\W(P,s,u)$.}  We apply the vertical simplicial partition of Section \ref{Subsec:VerticalPartition} to each part $P_i\in \V(P,s)$, to obtain an auxiliary vertical $u$-partition $\V(P_i,u)=\{(W_{i,k},\Delta'_{i,k})\mid 1\leq k\leq u\}$ as detailed in Section \ref{Subsec:Narrow}. See Figure \ref{Fig:Secondary} (right). We then set

$$
\W(P,s,u):=\bigcup_{i=1}^s \V(P_i,u)=\{\left(W_{i,k},\Delta'_{i,k}\right)\mid 1\leq i\leq s,1\leq k\leq u\}.
$$

 For each point $p\in P_i$, the {\it ambient cell $\Delta'(p)\subset \reals^{d-1}$ in $\W(P,s,u)$} is accordingly defined as the ambient cell of $p$ in the secondary vertical partition $\V(P_i,s)$ of $P_i$.

\medskip
\noindent {\bf The net $N_{>\delta,1}$.} Similar to the recursive net $N_\nar(P,s,\varepsilon)$ in Section \ref{Subsec:Narrow}, the net $N_{>\delta,1}$ is a combination of several recursively defined nets, which are constructed over ground sets $P_{i,j}$ (resp., $W_{i,k}$) are cut out by the subsets in the secondary partitions $\P(P,s,t)$ and $\W(P,s,u)$. Their inclusion in $N_{>\delta}$ will guarantee, for all of the remaining sets $K\in \K_{>\delta}$, that the points of their principal subsets $P_K$ are sufficiently well-distributed among the parts of $\P(P,s,t)$ and $\W(P,s,u)$.

\medskip
To this end, let

\begin{equation}\label{Eq:1stEpsPunct}
	\epsilon_1:=\cdot\frac{\delta\eps}{st^{1-1/\alpha_d}}.
\end{equation}

\noindent The net $N_{>\delta,1}$ is comprised of the following ingredients:

\begin{enumerate}
	\item For each $1\leq i\leq s$, and each $1\leq j\leq t$, we construct the net $N\left(P_{i,j},\epsilon_1\cdot st\right)$, which pierces the recursive sub-family $\K\left(P_{i,j},\eps_1\cdot st\right)$. Notice that the overall contribution of these nets $N\left(P_{i,j},\epsilon_1\right)$ to $|N_{>\delta,1}|$ amounts to 
$$
\left|\bigcup_{1\leq i\leq s,1\leq j\leq t}N\left(P_{i,j},\eps_1\cdot st\right)\right|\leq st\cdot f\left(\epsilon_1\cdot st\right)=st\cdot f\left(\eps\cdot t^{1/\alpha_d}\right).
$$

\item For $1\leq i\leq s$ we construct the net $N_\nar(P_i,u,\varepsilon)$, as described in Lemma \ref{Lemma:Narrow}, which pierces all the convex sets $K\in \K_{>\delta}$ that are $\varepsilon$-narrow in the secondary vertical subdivision $\V(P_i,u)$ of the set $P_i\in \V(P,s)$. Clearly, the overall contribution of these recursive nets amounts to at most

$$
\left|\bigcup_{i=1}^s N_\nar(P_i,u,\varepsilon)\right|\leq su\cdot f\left(\varepsilon \cdot u^{1/(d-1)}\right).
$$

\end{enumerate}

\noindent{\bf Analysis.} The properties of $N_{>\delta,1}$ are summarized in the following statement (using that $s\lll_\eta u\lll_\eta 1/\eps\ll t$).

\begin{lemma}\label{Lemma:1stNetPunctured}
\begin{enumerate}
	\item The cardinality of the net 
	$$
N_{>\delta,1}:=\left(\bigcup_{1\leq i\leq s,1\leq j\leq t}N\left(P_{i,j},\eps_1\cdot st\right)\right)\cup \left(\bigcup_{i=1}^s N_\nar(P_i,u,\varepsilon)\right)
$$ 

\noindent satisfies
$$
N_{>\delta,1}\leq st\cdot f\left(\eps\cdot t^{1/\alpha_d}\right)+su\cdot f\left(\varepsilon \cdot u^{1/(d-1)}\right)=	
$$
$$=O\left(t^{1+\eta}\cdot f\left(\eps\cdot t^{1/\alpha_d}\right)+u^{1+\eta}\cdot f\left(\varepsilon \cdot u^{1/(d-1)}\right)\right).$$

\item Upon adding the points of $N_{>\delta,1}$ 
 to our net $N_{>\delta}$ and, thereby, removing all the convex sets that are pierced so far by it, it can be assumed that every remaining convex set $K\in \K_{>\delta}$ meets the following criteria
 
 \begin{itemize}
 	\item[(a)]   $K$ is $\varepsilon$-spread in each sub-partition $\V(P_i,u)$, and
 	\item[(b)] we have that 
 	$|P_K\cap P_{i,j}|\leq \eps_1 n$ for all $1\leq i\leq s$ and $1\leq j\leq t$.
 \end{itemize}
\end{enumerate}
\end{lemma}

\subsection{Step 2}\label{Subsec:2ndNetPunctured}

The second net $N_{>\delta,2}$ that we construct for the $\delta$-punctured sets $K\in \K_{>\delta}$, is a combination of (i) the strong net of Lemma \ref{Thm:StrongNet} with respect to convex $(2d)$-hedra, and (ii) a small number of $1$-dimensional nets, whose carefully defined ground sets are contained in certain canonical vertical lines. Together with $N_{>\delta,1}$, its inclusion in $N_{>\delta}$ will at last guarantee that all of the remaining sets $K\in \K_{>\delta}$, that are missed by $N_{>\delta,1}\cup N_{>\delta,2}$, are sufficiently flat with respect to our secondary simplicial partition $\P(P,s,t)$, and thus can be dispatched using the strategy outlined in the end of Section \ref{Sec:Prelim}.

The construction of $N_{>\delta,2}$ is facilitated by additional geometric constructs, which rely on the secondary sub-divisions $\P(P,s,t)$ and $\V(P,s,u)$ introduced in Step 1.

\medskip
\noindent{\bf The secondary canonical families $\L_i=\L(P_i,u)$.} 
For each $1\leq i\leq s$ we obtain the $\vartheta$-canonical collection $\L_i:=\L(P_i,u)$ which exists, according to Theorem \ref{Theorem:MultipleSelection}, for the secondary vertical partition $\V(P_i,u)$ of $P_i$ (as part of $\W(P,s,u)$). Altogether, this amounts to a total of 
$$
\left|\bigcup_{i=1}^s\L_i\right|=O\left(su^{(d-1)^2}\right)
$$ 
\noindent vertical lines.

\medskip
\noindent{\bf Clipped cells.} For each $p\in P$ we use $\varphi'(p)$ to denote the restriction of its ambient cell $\varphi(p)$ in $\P(P,s,t)$ to the vertical prism above $\Delta'(p)\cap \Delta(p)$. (Here $\Delta(p)$ denotes the ambient cell of $p$ in the primary vertical partition $\V(P,s)$, and $\Delta'(p)$ denotes the ambient cell of $p$ in the secondary vertical partition $\W(P,s,u)$ of Step 1.) In what follows, we refer to $\varphi'(p)$ as the {\it clipped cell around $p$}.

\smallskip
Notice that, for each $1\leq i\leq s$, and each $p\in P_i$, the clipped cell $\varphi'(p)$ is a convex polytope with at most $3d+1$ facets which encloses the point set $P_{i,j}\cap W_{i,k}$, where $P_{i,j}$ (resp., $W_{i,k}$) is the part of the secondary partition $\P(P_i,t)$ (resp., $\V(P_i,u)$) that contains $p$. 

\medskip
\noindent{\bf Definition.} In what follows, we use $V'(p)$ to denote the vertex set of $\varphi'(p)$, and we use
 $V'$ to denote the set $\bigcup_{p\in P}V'(p)$ of all the vertices of the clipped cells. Since $s\lll_\eta u\lll_\eta t$, we have that $|V'|=O(stu)=O\left(t^{1+\eta/{10d^2}}\right)$.

\medskip
\noindent{\bf The net $N_{>\delta,2}$.} Denote
\begin{equation}\label{Eq:SecondConstantPunctured}
\eps_2:=\frac{c_\part}{100(4d+4)}\cdot \frac{\delta\eps}{st^{1-1/d}},
\end{equation}

\noindent where $c_\part>0$ is the constant in Matou\v{s}ek's Theorem \ref{Theorem:Simplicial}.

\medskip
 Let $0<c_2\leq 1/100$ be a sufficiently small constant that will be determined in the sequel.
Our second net $N_{>\delta,2}$ for $\K_{>\delta}$ is comprised of the following non-recursive ingredients:

\begin{figure}
    \begin{center}
    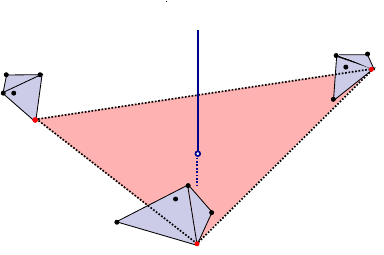  \caption{\small Step 2 in the definition of $N_{>\delta}$ in $\reals^3$ -- building the ground sets $X'_{\ell}$ for the 1-dimensional nets $N'_\ell$. Depicted are a secondary canonical line $\ell\in \L_i$, the clipped cells $\varphi'(p_i)$ around some three points $p_1,p_2,p_3\in P$, and a simplex $\phi$ spanned by some $3$ vertices of $\varphi'(p_1),\varphi'(p_2),\varphi'(p_3)$. Then the $\ell$-intercept $\phi\cap \ell$ is added to $X'_\ell$.}
      \label{Fig:SecondaryLines}
    \end{center}
\end{figure}

\begin{enumerate}
	
\item 
For each $1\leq i\leq s$, and each line $\ell\in \L_i$, we construct the set $X'_\ell=\{\phi\cap \ell\mid \phi\in {V'\choose d}\}$ of all the $\ell$-intercepts of the $(d-1)$-dimensional simplices $\phi\in {V'\choose d}$, which are spanned by the boundary vertices of the clipped cells $\varphi'(p)$ of the points $p\in P$. See Figure \ref{Fig:SecondaryLines}.

We then add to $N_{>\delta,2}$ the net $N'_\ell$ which includes every $\left\lceil c_2t^{d-d/\alpha_d}/u^{(d-1)^2}\right\rceil$-th point of $X'_\ell$. 

\item Lastly, we add to $N_{>\delta,2}$ the strong $(c_2\eps_2)$-net $N_\strong\left(P,c_2\eps_2\right)$ of Lemma \ref{Thm:StrongNet} that pierces every convex $(2d)$-hedron that is $(c_2\eps_2)$-heavy with respect to $P$. Note that $|N_\strong\left(P,c_2\eps_2\right)|=O\left(\left(1/\eps_2\right)\log \left(1/\eps_2\right)\right)$.

\end{enumerate}

\medskip
\noindent{\it Remark.} Though the choice of the thresholds $\eps_2$ and $c_2t^{d-d/\alpha_d}/u^{(d-1)^2}$ may seem arbitrary at this stage, it will become clearer in the subsequent analysis of Sections \ref{Subsec:3rdNetPunct} and \ref{Subsec:Flat}. For the time being, it is instructive to point out that every convex $K\in \K_{>\delta}$, that has been pierced by $N_{>\delta,2}$, must fall into one of the following categories: (i) its principal set $P_K$ is not in sufficiently convex position, or (ii) the set $K$ is ``fat" with respect to $\P(P,s,t)$,\footnote{This is because no hyperplane $H(\tau)$ through a simplex $\tau\in {P_K\choose d}$ can pass in a close vicinity of a point of $X'_\ell$, provided a generic enough position of $P$.}
 so its overall shape is far from a hyperplane.

\medskip
\noindent{\bf The analysis.}  The cardinality of $N_{>\delta,2}$ is bounded by the following statement.
\begin{lemma}\label{Lemma:2nsNetPunctured}
The net
$$
N_{>\delta,2}:=\left(\bigcup_{1\leq i\leq s}\bigcup_{\ell\in \L_i} N'_\ell\right)\cup N_\strong(P,c_2\eps_2)
$$

\noindent satisfies

$$
|N_{>\delta,2}|=O\left(s^{d+1}u^{(d-1)^2+d}t^{d/\alpha_d}+\frac{s t^{1-1/d}}{\delta\eps}\log \frac{1}{\eps}\right)=
O\left(t^{d/\alpha_d+\eta/(10d)}+\frac{t^{1-1/d}}{\eps^{1+\eta}}\right).
$$

\end{lemma}
\begin{proof}
By the previous choice of $\delta=1/r^{d(d+1)}$, and of $r\lll_\eta s\lll_\eta u \lll_\eta 1/\eps\ll t$, it follows that the overall cardinality of the nets $N'_\ell$ satisfies
 
$$
\left|\bigcup_{1\leq i\leq s}\bigcup_{\ell\in \L_i} N'_\ell\right|=O\left(\left(\sum_{i=1}^s|\L_i|\right)\cdot \frac{u^{(d-1)^2}\left|V'\right|^d}{t^{d-d/\alpha_d}}\right)=O\left(\frac{s\cdot u^{2(d-1)^2}\cdot {sut\choose d}}{t^{d-d/\alpha_d}}\right)
=O\left(t^{d/\alpha_d+\eta/(10d)}\right),
$$

\noindent while 
$$
|N_\strong(P,c_2\eps_2)|=O\left(\frac{1}{\eps_2}\log \frac{1}{\eps_2}\right)=O\left(\frac{s t^{1-1/d}}{\delta\eps}\log \frac{1}{\eps}\right)=O\left(\frac{t^{1-1/d}}{\eps^{1+\eta}}\right).
$$

\end{proof}

We add $N_{>\delta,2}$ to $N_{>\delta}$ and, thereby, remove from $\K_{>\delta}$ every convex set that is hit by $N_{>\delta,2}$. 

\subsection{Step 3} \label{Subsec:3rdNetPunct}

We are at last ready describe the third and final net $N_{>\delta,3}$, which uses Theorem \ref{Theorem:Sparse} to pierce all the remaining $\delta$-punctured $\varepsilon$-spread sets $K\in \K_{>\delta}$, that have been missed by the combination $N_{>\delta,1}\cup N_{>\delta,2}$ of the previous two nets. To this end, we will argue that, pending on a suitable choice of the constant $c_2>0$ in the previous Step 2, every remaining set $K\in \K_{>\delta}$ is sufficiently flat with respect to the secondary partition $\P(P,s,t)$ of $P$, in the sense that will be made formal shortly.


\bigskip
\noindent{\bf Definition.} Let $K\in \K_{>\delta}$ be a convex set with the principal $(d+1)$-tuple $\Psi(K)=(P_{i_0},\ldots,P_{i_d})$, where $P_{i_j}\in \P(P,s)$ for all $0\leq j\leq d$.
For the sake of brevity, we relabel the elements of $\Psi(K)$ so that $\Psi(K)=(P_0,\ldots,P_d)$ (where the cell $\Delta_0\supset P_0^\perp$ is surrounded by the remaining $d$ cells $\Delta_i\supset P_i^\perp$, with $1\leq i\leq d$).

\medskip
For each $0\leq i\leq d$, we select a subset $P_i(K)\subseteq P_K\cap P_{i}$ of cardinality 
$$
|P_i(K)|=\lceil \delta |P_K|/(4d+4)s\rceil.
$$

Note that the secondary $t$-partition $\P_0=\P(P_{0},t)=\{P_{0,j}\mid 1\leq j\leq t\}$ of the part $P_0$ induces the partition $\P_0(K)=\{P_{0,j}(K)\mid 1\leq j\leq t\}$ of the subset $P_0(K)$, with $P_{0,j}(K):=P_{0,j}\cap P_0(K)$, for all $1\leq j\leq t$. 

\medskip
\noindent{\bf Definition.} We say that the part $P_{0,j}\in \P_0$ is {\it full} with respect to $K$ (or, simply, that the subset $P_{0,j}(K)$ is {\it full}) if we have that $|P_{0,j}(K)|\geq \eps_2n$.

\medskip
It can be assumed, in what follows, that $\eps_2n\geq 1/c_2\geq 100$ for, otherwise, every singleton subset within $P_K$ is $(c_2\eps_2)$-heavy, and the entire family $\K(P,\eps)$ must have been pierced by the strong net $N_\strong(P,c_2\eps_2)=P$ in Step 2.

In view of Lemma \ref{Lemma:1stNetPunctured}, every full subset $P_{0,j}(K)\in \P_0(K)$ must then satisfy

\begin{equation}\label{Eq:Full}
100\leq 1/c_2\leq \eps_2n \leq |P_{0,j}(K)|\leq \eps_1n.
\end{equation}

\medskip
We say that a $2$-edge $\lambda\in {P\choose 2}$ is {\it short} if both its endpoints lie in the same part $P_{i,j}$ of the secondary partition $\P(P,s,t)$ of $P$.

\begin{figure}
    \begin{center}
             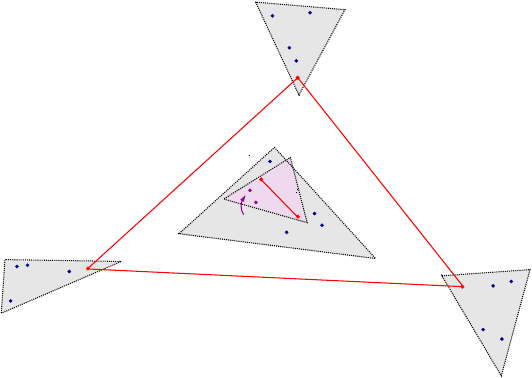
        \caption{\small The overall setup for the case analysis of Step 3 -- a schematic view from above in dimension $d=3$. Depicted are the sets $P_i(K)$ together with the enclosing cells $\Delta_i$ in $\V(P,s)$, a $(d-1)$-simplex $\tau\in \Sigma(K)$, and a short 2-edge $\lambda\in \Lambda_{0,j}(K)$.  Notice that $\lambda^\perp\subseteq \Delta_0\subseteq  \tau^\perp$ holds for any such pair $(\lambda,\tau)\in \Lambda_{0,j}(K)\times\Sigma(K)$. 
        }
        \label{Fig:SigmaSimplices}
    \end{center}
\end{figure}

\bigskip
\noindent{\bf Flatness.}  
At the heart of our construction lies the relation between the following two classes of simplices, which are determined by the remaining convex sets $K\in \K_{>\delta}$; see Figure \ref{Fig:SigmaSimplices}. 

\begin{enumerate}
	\item[A.] The $t$ families 
	$$
	\Lambda_{0,j}(K):={P_{0,j}(K)\choose 2},
	$$ 
	\noindent each comprised of such short 2-edges $\lambda$ that both of their vertices come from a given subset $P_{0,j}(K)$ in the induced partition $\P_0(K)$ of $P_0(K)$, with $1\leq j\leq t$.

\item[B.] The family $\Sigma(K)$ of all such simplices $\tau\in {P_K\choose d}$ whose $d$ vertices $p_i$ are chosen from the $d$ different subsets $P_i(K)$ of $P_K$, with $1\leq i\leq d$. Namely,
$$
\Sigma(K):=\left\{\conv\left(p_1,p_2,\ldots,p_d\right)
\mid p_1\in P_1(K),p_2\in P_2(K),\ldots,p_{d}\in P_{d}(K)\right\}.
$$  

\end{enumerate}

\noindent{\it Remark.} Since the ambient cell $\Delta_{0}$ of $P_0$ is surrounded by the $d$ ambient cells $\Delta_i$ of the remaining sets $P_i$, with $1\leq i\leq d$, it follows that every pair $(\lambda,\tau)\in \Lambda_{0,j}(K)\times \Sigma(K)$, with $1\leq j\leq t$, must satisfy $\lambda^\perp\subset \Delta_0\subset \tau^\perp$. See Figure \ref{Fig:SigmaSimplices}. In other words, the intersection of the supporting hyperplane $H(\tau)$ of $\tau$ with the vertical prism $\Delta_0^*$ over $\Delta_0$, must lie entirely within $\tau$, so that every point of $P_{0}(K)$ lies either above or below $\tau$. 



\smallskip
To gain a little intuition, suppose that the convex set $K\in \K_{>\delta}$ at hand is a {\it hyperplane}. Then the respective set $P_0(K)$ (along with all the short 2-edges in $\biguplus_{i=1}^t\Lambda_{0,j}(K)$) is contained in {\it every} simplex $\tau\in \Sigma(K)$. In particular, the points of $P_{0}(K)$ must lie in at most $c_\part t^{1-1/d}$ simplicial cells $\varphi_{0,j}$, which are crossed by every simplex $\tau\in \Sigma(K)$. Therefore, by the pigeonhole principle the majority of the points in $P_0(K)$ must belong to the full parts $P_{0,j}(K)\in \P_0(K)$, which satisfy 
$$
|P_{0,j}(K)|\geq \frac{|P_0(K)|}{100\cdot c_\part\cdot t^{1-1/d}}\geq \eps_2n.
$$

\noindent A suitable infinitesimal perturbation of $P$ puts the points of $P_K$ in a convex (and ``sufficiently generic") position, while placing the entire subset $P_0(K)$ either above, or below, of all the simplices $\tau\in \Sigma(K)$. As a result, every short edge $\lambda\in \biguplus_{i=1}^t\Lambda_{0,j}(K)$ is now ``parallel" to every simplex $\tau\in \Sigma(K)$, in the sense that its supporting line $\aff(\lambda)$ misses $\tau$.

\medskip
In what follows, we show that a very similar set of properties can be established for all the remaining convex sets $K\in \K_{>\delta}$ if (some of) the points in $\biguplus_{i=1}^d P_i(K)$ are replaced by their ambient simplicial cells $\varphi(p)$ in $\P(P,s,t)$, and the possible ``co-planarities" involving such cells $\varphi(p)$ are taken into account.

\bigskip
\noindent{\bf Definition.} For every convex set $K\in \K_{>\delta}$, and every simplex $\tau=\{p_1,\ldots,p_d\}\in \Sigma(K)$, we use $\Phi(\tau)$ to denote the convex hull $\conv\left(\bigcup_{i=1}^d\varphi(p_i)\right)$ of the ambient cells of the $d$ vertices of $\tau$, and refer to $\Phi(\tau)$ as the {\it envelope} of $\tau$ in $\P(P,s,t)$.\footnote{Notice that these $d$ ambient cells $\varphi(p_i)$ must be distinct, as each of them comes from a distinct secondary partition $\P_i=\P(P_i,t)$.} 

\bigskip
\noindent{\bf Definition.} Let $K\in \K_{>\delta}$ be a $\delta$-punctured set with $\Psi(K)=(P_0,\ldots,P_d)$, $\tau=\conv(p_1,\ldots,p_d)$ be a simplex in $\Sigma(K)$, and $\lambda$ be short 2-edge in $\Lambda_{0,j}(K)$. 
\begin{enumerate}
\item We say that $\lambda$ and $\tau$ are {\it parallel} if both $\lambda$ and $\tau$ are faces of the convex $(d+2)$-hedron $\conv(\tau\cup\lambda)$; see Figure \ref{Fig:ParallelSimplices}. 

Notice that the fact that $\lambda^\perp\subset \tau^\perp$ implies that the 2-edge $\lambda$ must lie either entirely above, or entirely below $\tau$, so that $\tau$ is always a facet of $\conv(\tau\cup \lambda)$.

\begin{figure}
    \begin{center}
            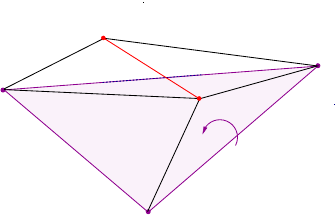
        \caption{\small The 2-edge $\lambda\in \Lambda_{0,j}(K)$ is parallel to the simplex $\tau\in \Sigma(K)$ in $\reals^3$, so they comprise facets of the convex pentahedron $\conv(\tau\cup \lambda)$. 
        }
        \label{Fig:ParallelSimplices}
    \end{center}
\end{figure}

\item We say that the ordered pair $(\lambda,\tau)$ is {\it tight} if $\lambda$ and $\tau$ are parallel and, in addition, $\lambda$ is contained in the envelope $\Phi(\tau)$ of $\tau$; see Figure \ref{Fig:TightLoose3D} (left). (Notice that the $d$ vertices of $\tau=\{p_1,\ldots,p_d\}$ determine $d$ distinct simplicial cells $\varphi(p_i)$, each within the respective secondary partition $\P_i=\P(P_i,t)$.)

If such a pair $(\lambda,\tau)\in \Lambda_{0,j}(K)\times\Sigma(K)$ is only parallel yet not tight, we say that it is {\it loose}, in which case at least one of the two vertices of $\lambda$ lies either above or below $\Phi(\tau)$; see Figure \ref{Fig:TightLoose3D} (right). 
\end{enumerate}

\begin{figure}
    \begin{center}
             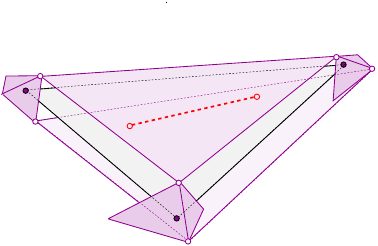\hspace{2cm}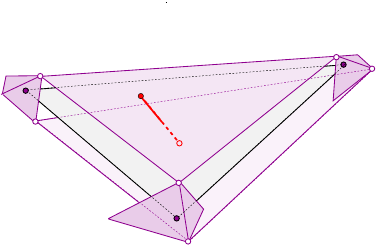
        \caption{\small Left: A tight pair $(\lambda,\tau)$ in $\reals^3$. The short 2-edge $\lambda\in \Lambda_{0,j}(K)$ is contained in the envelope $\Phi(\tau)=\conv(\varphi(p_1)\cup\varphi(p_2)\cup \varphi(p_3))$. Right: A loose pair $(\lambda,\tau)$ in $\reals^3$. At least one of the vertices of the short 2-edge $\lambda$ lies either above, or below $\Phi(\tau)$.} 
        \label{Fig:TightLoose3D}
    \end{center}
\end{figure}

\bigskip
\noindent{\bf Definition.} Let $K\in \K_{>\delta}$ be a $\delta$-punctured set with $\Psi(K)=(P_0,\ldots,P_d)$. 

\begin{enumerate}

\item We say that a full subset $P_{0,j}(K)$ in the induced partition $\P_0(K)$ of $P_0(K)$ is {\it tight} if
at least $|\Lambda_{0,j}(K)|\cdot|\Sigma(K)|/80$ of the pairs $(\lambda,\tau)\in \Lambda_{0,j}(K)\times\Sigma(K)$ are tight.

\item Let $\T_0(K)$ to denote the subset of all the tight parts $P_{0,j}(K)$ within $\P_0(K)$.

\item We say that the set $K$ is {\it flat} if the tight parts $P_{0,j}(K)\in \T_0(K)$ encompass at least $|P_0(K)|/5$ points of $P_0(K)$.

\end{enumerate}

\medskip
\noindent{\bf The net $N_{>\delta,3}$.} Our small-size net the remaining sets $K\in \K_{>\delta}$ is an easy corollary of the following two properties, whose respective proofs are postponed to Sections \ref{Subsec:Flat} and \ref{Subsec:SparsePi}.

\begin{lemma}\label{Lemma:Flat}
Given a sufficiently small choice of the constant $c_2>0$ in Step 2, every remaining convex set $K\in \K_{>\delta}$, that was missed by the combination $N_{>\delta,1}\cup N_{>\delta,2}$, must be flat.
\end{lemma}

\begin{lemma}\label{Lemma:SparsePi}
There is a subset $\tilde{\Pi}=\tilde{\Pi}(P,\eps,\delta,s,t)\subseteq {P\choose d}$ with the following properties:

\begin{enumerate}
	\item[(i)] We have that
$$
|\tilde{\Pi}|=O\left(\frac{n^d}{\delta\eps t^{2/d}}\right).
$$ 

\item[(ii)] 
 Every flat convex set $K\in \K_{>\delta}$ satisfies
\begin{equation}\label{Eq:LocallyDense}
	\left|{P_K\choose d}\cap \tilde{\Pi}\right|=\Omega\left((\delta\eps n)^d/\r^d\right).
\end{equation}

\end{enumerate}
\end{lemma}

According to Lemma \ref{Lemma:Flat} every remaining set $K\in \K_{>\delta}$ is flat and, by the second part of Lemma \ref{Lemma:SparsePi}, satisfies (\ref{Eq:LocallyDense}).
Hence, the final net $N_{>\delta,3}$ for the all remaining sets in $\K_{>\delta}$ is provided by the following lemma, which is established by plugging $\tilde{\Pi}$ into Theorem \ref{Theorem:Sparse}, with the threshold $\sigma=\Omega\left(\delta^d/s^d\right)$, and then using the upper bound on $|\tilde{\Pi}|$ in the first part of Lemma \ref{Lemma:SparsePi}. (Notice that the last inequality follows by the choice of $s\lll_\eta u\lll_\eta t$, and $\delta=1/r^{d(d+1)+1}$, which in particular implies that $\delta^d/s^d$ is much larger than $1/u^{\eta}$.)  
 
\begin{lemma}\label{Lemma:ThirdNetPunctured}
	Suppose that the choice of the constant $c_2>0$ in Step 2 meets the criteria of Lemma \ref{Lemma:Flat}. 
	Then the remaining sets of $\K_{>\delta}$ can be pierced by a net $N_{>\delta,3}$ of cardinality 
	\begin{equation}\label{Eq:3rdNetPunctured}		
	\left|N_{>\delta,3}\right|=O\left(u\cdot f\left(\eps\cdot u^{1/(d-1)}\cdot \delta^{d-1}/s^{d-1}\right)+\frac{u^{(d-1)^2}\cdot s^{d\beta_{d-1}}}{\delta^{d\beta_{d-1}+1}\eps^{d+1} t^{2/d}}\right)=
\end{equation}

	$$
	=O\left(u\cdot f\left(\eps\cdot u^{1/(d-1)-\eta}\right)+\frac{1}{\eps^{d+1+\eta}t^{2/d}}\right).
	$$
		
\end{lemma}

\noindent{\it Remark.} As was explained in Section \ref{Subsec:RecursiveFramewk}, the recursive term in (\ref{Eq:3rdNetPunctured}) is of the ``$O^*\left(1/\eps^{d-1}\right)$-type", while the non-recursive term is smaller than $1/\eps^d$ for any value of $t$ that is considerably larger than $(1/\eps)^{d/2}$. To complete the proof of Theorem \ref{Theorem:BoundPunct}, in Section \ref{Subsec:PuncturedWrapUp} we will combine the bounds in Lemmas \ref{Lemma:1stNetPunctured}, \ref{Lemma:2nsNetPunctured}, and \ref{Lemma:ThirdNetPunctured}, and choose a suitable $t=(1/\eps)^{3d/4+o_d(1)}$.

\subsection{Proof of Lemma \ref{Lemma:Flat}}\label{Subsec:Flat}

The proofs of both Lemmas \ref{Lemma:Flat} and \ref{Lemma:SparsePi} will use the following basic property of tangent hyperplanes to convex polytopes.

\begin{lemma}\label{Lemma:TangentHyperplanes} Let $\varphi_1,\ldots,\varphi_{d}$ be a collection of $d$ convex polytopes in $\reals^d$ whose vertices are in general position, and $\Delta$  be a subset of $\reals^{d-1}$ that is surrounded by their respective vertical projections $\varphi_1^\perp,\ldots,\varphi_d^\perp$ (see Figure \ref{Fig:TangentHyperplanes} (left)). Then there is a unique hyperplane $H^+$ (resp., $H^-$) that is tangent
to each $\varphi_i$ from above (resp., below) at the respective point $y_i\in \varphi_i$ (resp., $z_i\in \varphi_i$), and so that $\Delta\subset \conv(y_1,\ldots,y_d)^\perp\cap \conv(z_1,\ldots,z_d)^\perp$.
\end{lemma}

In other words, any vertical line through $\Delta$ must cross the boundary of the polytope $\Phi:=\conv\left(\varphi_1\cup\ldots\cup \varphi_d\right)$ at the facets $\conv(z_1,\ldots,z_d)$ and $\conv(y_1,\ldots,y_d)$, which support, respectively, $H^-$ and $H^+$, and in this increasing order of the $d$-th coordinate.

\begin{proof}[Proof of Lemma \ref{Lemma:TangentHyperplanes}.]
Since $\Delta$ is surrounded by $\varphi_i^\perp$,  for $1\leq i\leq d$, it follows via Corollary \ref{Corol:Surrounded} that 
the family $\Gamma:=\{\varphi^\perp_1,\ldots,\varphi^\perp_d,\Delta\}$ of $d+1$ compact convex sets in $\reals^{d-1}$ is separated.
In particular,
$\Delta$ cannot intersect any of the ``$(d-1)$-wise" hulls $\conv\left(
\bigcup_{j\in [d]\setminus \{i\}}\varphi_i^\perp\right)$, for $1\leq i\leq d$ (or, else, some $d$ elements of $\Gamma$, including $\Delta$, would be crossed by the same $(d-2)$-plane in $\reals^{d-1}$).

\begin{figure}
    \begin{center}
    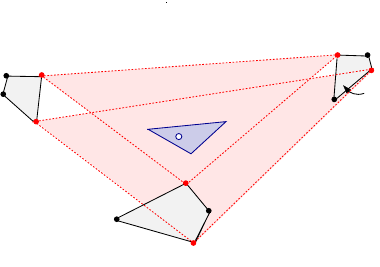 \hspace{2cm}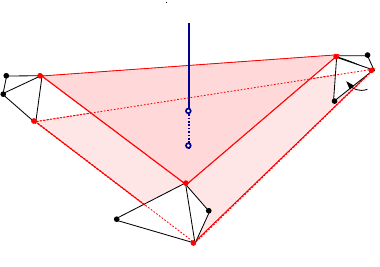  \caption{\small Lemma \ref{Lemma:TangentHyperplanes}. Left: $\Delta$ is a subset of $\reals^2$ which is surrounded by the vertical projections $\varphi^\perp_1,\varphi_2^\perp$ and $\varphi^\perp_3$ of the 3-dimensional polytopes $\varphi_1,\varphi_2$ and $\varphi_3$. Right: Any vertical line $x_0^*$ over a point $x_0\in \Delta$ crosses the same facets $\phi^+$ and $\phi^-$ which support, respectively, the upper and the lower  common tangents to $\varphi_1,\varphi_2$ and $\varphi_3$.}
      \label{Fig:TangentHyperplanes}
    \end{center}
\end{figure}

 Fix a point $x_{0}\in \Delta$. The preceding discussion implies that the vertical line $x_0^*$ through $x_{0}$ must cross the boundary of the polytope $\Phi=\conv\left(\varphi_1\cup\ldots\cup \varphi_d\right)$ at some pair of facets $\phi^+$ and $\phi^-$
whose supporting vertices come from $d$ {\it distinct} polytopes $\varphi_i$, as depicted in Figure \ref{Fig:TangentHyperplanes} (right). Namely, we have that $\phi^-=(y_1,\ldots,y_d)$ and $\phi^+=(z_1,\ldots,z_d)$, where $y_i,z_i\in \partial\varphi_i$ for all $1\leq i\leq d$, and 
 $\Phi$ lies below $H\left(\phi^+\right)$ and above $H\left(\phi^-\right)$.

Since $\Delta$ is surrounded by $\varphi_1^\perp,\ldots,\varphi_d^\perp$, it must be that $\Delta\subseteq (\phi^+)^\perp\cap (\phi^+)^\perp$, and, furthermore, the simplex $\phi^+$ (resp., $\phi^-$) supports the unique common upper (resp., lower) tangent $H^+=H(\phi^+)$ (resp., $H^-=H(\phi^-))$ to the polytopes $\varphi_1,\ldots,\varphi_d$, which does not depend on the choice of $x_0$.\footnote{A more general study of the space of tangent hyperplanes, due to Cappell {\it et al.} \cite{Cappell}, implies that {\it any} family of $d$ convex sets in $\reals^d$, whose elements cannot be simultaneously crossed by a $(d-2)$-plane, must be ``sandwiched" between a pair of distinct oriented external hyperplane tangents. The condition of Cappell {\it et al.} is clearly satisfied by $\{\varphi_1,\ldots,\varphi_d\}$, due to the separation of $\{\varphi'_1,\ldots,\varphi'_d\}$ in $\reals^{d-1}$.}  
\end{proof}





To establish Lemma \ref{Lemma:Flat}, let us fix a convex set $K\in \K_{>\delta}$, with the principal $(d+1)$-tuple $\Psi(K)=(P_0,P_1,\ldots,P_d)$, and assume that $K$ is {\it not} flat.
That is, the tight subsets $P_{0,j}(K)\in \T_0(K)$ encompass fewer than $|P_0(K)|/5$ points of $P_0(K)$, whereas the remaining subsets $P_{0,j}(K)\in \P_0(K)\setminus\T_0(K)$ encompass more than $\frac{4}{5}|P_0(K)|$ points of $P_0(K)$.
 To establish the lemma, it suffices to show that, pending a sufficiently small choice of the constant $c_2>0$ in Step 2, such a set $K$ must have been pierced by the combination $N_{>\delta,1}\cup N_{>\delta,2}$ of previous nets, and subsequently removed from $\K_{>\delta}$.

In view of our choice of $\varepsilon$ in (\ref{Eq:Spread}), and the lower bound $n_0(\eps)$ on $n$ in (\ref{Eq:ManyPoints}), we have that

\begin{equation}\label{Eq:ManyPointsLocal}
|P_0(K)|\geq \frac{\delta\lceil \eps n\rceil}{(4d+4)s} \geq \frac{10^6\varepsilon \lceil n/s\rceil}{\vartheta}\geq \max\left\{10^5\cdot \frac{\varepsilon|P_0|}{\vartheta},10^6d\right\}.
\end{equation}

Also note that Lemma \ref{Lemma:TangentHyperplanes} yields a unique pair of, respectively, upper and lower common tangent hyperplanes $H^+$ and $H^-$ to the $d$ polytopes $\conv\left(P_i(K)\right)$, with $1\leq i\leq d$. Specifically, $H^+$ (resp., $H^-$) meets the polytopes $\conv\left(P_i(K)\right)$ at the respective vertices $a_i\in P_i(K)$ (resp., $b_i\in P_i(K)$), which determine 
the simplex $\tau^+=\conv(a_1,\ldots,a_d)$ (resp., $\tau^-=\conv(b_1,\ldots,b_d)$) in $\Sigma(K)$. 
Also note that the ``inner" cell $\Delta_0$, which encloses the projection of $P_0(K)$, lies in the intersection $(\tau^+)^\perp\cap (\tau^-)^\perp$ of the respective projections of $\tau^+$ and $\tau^-$. 
Hence, every point $p\in P_0(K)$ is either contained in the $(2d)$-hedron $I_0(K)=\conv(\tau^+\cup \tau^-)\subseteq K$, or it lies to the same side (either above or below) of both simplices $\tau^+$ and $\tau^-$.

\medskip
The preceding discussion implies that any part $P_{0,j}(K)$ in the induced partition $\P_0(K)$ of $P_0(K)$ must belong to at least one of the following families
(see Figure \ref{Fig:TypesOfSubsets}):
\begin{figure}
    \begin{center}
        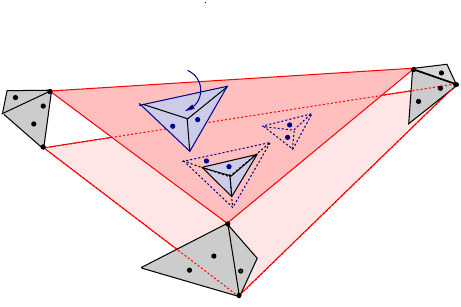  \caption{\small Proof of Lemma \ref{Lemma:Flat} in $\reals^3$. The simplex $\tau^+$ (resp.., $\tau^-$) supports the upper (resp., lower) common tangent hyperplane to $\conv\left(P_1(K)\right),\conv(P_2(K))$, and $\conv\left(P_3(K)\right)$. Three classes of subsets $P_{0,j}\in \P_0(K)$ are illustrated, together with their enclosing cells $\varphi_{0,j}$.}
      \label{Fig:TypesOfSubsets}
    \end{center}
\end{figure}

\begin{itemize}
	\item the family $\I_0(K)$, which consists of all such subsets $P_{0,j}(K)\in \P_0(K)$ that are contained in the aforementioned $(2d)$-hedron $I_0(K)=\conv\left(\tau^+\cup\tau^-\right)$ within $K$,
	\item the family $\Z_0(K)$, which is comprised of all such subsets $P_{0,j}(K)$ that fall in the zone of either one, or both, of the hyperplanes $H^+$ and $H^-$ (i.e., their ambient cells $\varphi_{0,j}$ in $\P_0$ intersect $H^+\cup H^-$),\footnote{Notice that the described subdivision of $\P_0(K)$ is not necessarily disjoint due to the possible overlap between the families $\I_0(K)$ and $\Z_0(K)$.} or
	\item the family $\P_0^+(K)$ (resp., $\P_0^-(K)$) which is comprised of all such subsets $P_{0,j}(K)$ that lie above $\tau^+$ (resp., below $\tau^-$) and, furthermore, their enclosing cells $\varphi_{0,j}$ in $\P_0$ lie entirely above $H^+$ (resp., entirely below $H^-$). 
\end{itemize}








To proceed, we distinguish between three basic scenarios for $K\in \K_{>\delta}$, based on the distribution of the points of $P_0(K)$ within its partition $\P_0(K)$ that is induced by $\P_0=\P(P_0,t)$.

\medskip
\noindent {\bf Case (a).} If at least $c_2\eps_2n$ points of $P_0(K)$ belong to $\bigcup \I_0(K)$ (in particular, they are ``sandwiched" between $\tau^+$ and $\tau^-$),
then the polytope $I_0(K)$ is $(c_2\eps_2)$-heavy, so $K$ must have been pierced in Step 2 by the net $N_\strong(P,c_2\eps_2)$. 

\medskip
Hence, it can be assumed, from now on, that case (a) does not occur, that is, 

\begin{equation}\label{Eq:FewInside}
	\left|\bigcup \I_0(K)\right|<c_2\eps_2n.
\end{equation}





\medskip
\noindent {\bf Case (b).} If the {\it non-full} parts of $\P_0(K)$ encompass at least $|P_0(K)|/10$ points of $P_0(K)$,
we argue that, pending on a suitable choice of $c_2>0$, such a convex set $K\in \K_{>\delta}$ must have been pierced by some $1$-dimensional net $N'_{\ell_0}\subseteq N_{>\delta,2}$, which has been constructed Step 2 for some line $\ell_0$ in the secondary $\vartheta$-canonical subset $\L_0=\L(P_0,u)$ of $P_0$.

To this end, note that, according to Theorem \ref{Theorem:Simplicial}, the zones of the hyperplanes $H^+$ and $H^-$ together encompass at most $2c_\part t^{1-1/d}$
non-full parts of $\Z_0(K)$, whose ambient cells $\varphi_{0,j}$ cross $H^+\cup H^-$,
so that the union of these non-full parts accounts for fewer than $2c_\part t^{1-1/d}\eps_2n$ points of $P_0(K)$. Together with (\ref{Eq:ManyPointsLocal}) and (\ref{Eq:FewInside}) (and with our prior choice (\ref{Eq:SecondConstantPunctured}) of $\eps_2$, and of $c_2<1/100$), this implies that the non-full parts $P_{0,j}(K)$ of $\P_0^+(K)\uplus \P_0^-(K)$ account for at least

$$
\frac{|P_0(K)|}{10}-c_2\eps_2n-2c_\part t^{1-1/d}\eps_2n\geq \frac{|P_0(K)|}{10}-\frac{\delta\eps n\left(c_2+2c_\part \cdot t^{1-1/d}\right)}{c_\part\cdot (100d+100)st^{1-1/d}}\geq \frac{|P_0(K)|}{20}
$$

\noindent points of $P_0(K)$. 

It thus can be assumed, with no loss of generality, that the non-full parts of $\P_0^+(K)$ encompass a subset $P_0^+(K)$ that lies entirely above $\tau^+$, and whose cardinality satisfies
$$
|P_0^+(K)|\geq |P_0(K)|/40\geq \max\left\{\varepsilon|P_0|/\vartheta,10^4d\right\},
$$
with the last inequality following from the previous lower estimate (\ref{Eq:ManyPointsLocal}) on $|P_0|$.
Also note that $P^+_0(K)$ determines, within ${P_0\choose d}$, the non-empty sub-family $\Pi_0^+(K):={P_0^+(K)\choose d}$ whose cardinality satisfies
$|\Pi_0^+(K)|\geq 80^{-d}{|P_0(K)|\choose d}$.

Consider the secondary partition $\V_0=\V(P_0,u)$ of $P_0$, and the secondary $\vartheta$-canonical line family $\L_0=\L(P_0,u)$ which has been constructed in Step 2 over $\V(P_0,u)$.
Since the set $K$ is $\varepsilon$-spread in $\V(P_0,u)$, and at least $\varepsilon|P_0|/\vartheta$ of its points belong to $P_0$, there must exist a line $\ell_0\in \L_0$, with a subset $\Pi^+_{\ell_0}\subseteq\Pi^+_0(K)$ of at least $\vartheta |\Pi^+_0(K)|$ simplices, that together meet the following criteria (see Figure \ref{Fig:CaseB-Setup}):
 	\begin{enumerate}
	\item[(i)] For each simplex $\pi=\{p_1,\ldots,p_d\}\in \Pi^+_{\ell_0}$, the line $\ell_0$ is surrounded, within the secondary vertical partition $\V_0=\V(P_0,u)$, by the ambient  $(d-1)$-dimensional cells $\Delta'(p_1),\ldots,\Delta'(p_d)$ of the vertices of $\pi$.  
	\item[(ii)] $\ell_0$ also crosses $\tau^+$ at a point which lies below all its intersections with the simplices of $\Pi^+_{\ell_0}$. (This is because $\ell_0^\perp$ lies in the subset $\conv\left(P_0^\perp\right)\subseteq \Delta_0$, which is contained in $(\tau^+)^\perp$, while all of the simplices $\pi$ in $\Pi_{\ell_0}^+$ are determined by the set $\bigcup\P_0^+(K)$, which lies above $\tau^+$.)

\end{enumerate}

\begin{figure}
    \begin{center}
        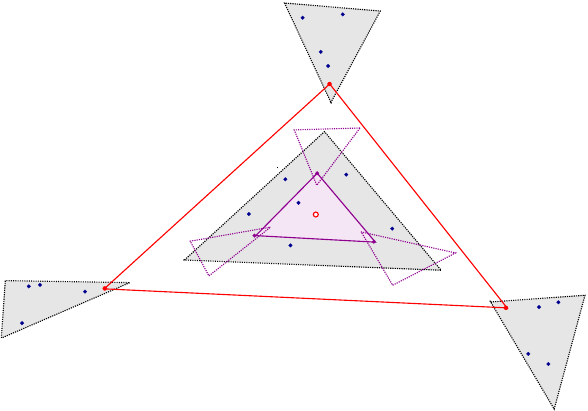  \caption{\small Proof of Lemma \ref{Lemma:Flat} in $\reals^3$. Case (b) -- a schematic view from above. Depicted are the vertical line $\ell_0\in \L(P_0,u)$, a simplex $\pi=\conv(p_1,p_2,p_3)\in \Pi^+_{\ell_0}$. Since all of the vertices of $\pi$ lie above $\tau^+$, the intersection $\ell_0\cap \tau^+$ lies below $\ell_0\cap \pi$. In the secondary vertical partition $\V(P_0,u)$ of $P_0$, the line $\ell_0$ is surrounded by the ambient cells $\Delta'(p_1),\Delta'(p_2)$ and $\Delta'(p_3)$ of, respectively, $p_1,p_2$ and $p_3$.}
      \label{Fig:CaseB-Setup}
    \end{center}
\end{figure}

 Fix a simplex $\pi=\conv(p_1,\ldots,p_d)$ of $\Pi^+_{\ell_0}$. Note that, by the choice of $\P_0^+(K)$, and of the point set $P_0^+(K)\subseteq \bigcup \P_0^+(K)$, the $d$ ambient clipped cells $\varphi'_1=\varphi'(p_1),\ldots,\varphi'_d=\varphi'(p_d)$ of the vertices of $\pi$, must lie above $\tau^+$, as each of them is contained in the prism $\Delta^*_0$ over $\Delta_0$. Furthermore, the vertical projections of $\varphi'_1,\ldots,\varphi'_d$ must surround $\ell_0$, as each of these cells $\varphi'_i$ has been ``clipped" to the vertical prism $(\Delta'(p_i))^*$ over the ambient $(d-1)$-dimensional cell $\Delta'(p_i)$ of $p_i$ in $\V(P_i,u)$.\footnote{In particular, this implies that the clipped cells $\varphi'(p_1),\ldots,\varphi'(p_d)$ are distinct.} 

Consider the polytope $\Phi'(\pi):=\conv\left(\varphi'_1\cup \ldots \cup \varphi'_d\right)$. Applying Lemma \ref{Lemma:TangentHyperplanes} (with $\Delta=\{\ell^\perp_0\}$) yields a boundary $(d-1)$-simplex $\phi^-_\pi$ of $\Phi'(\pi)$
whose supporting hyperplane is tangent to $\varphi'_1,\ldots,\varphi'_d$ from below, and whose vertical projection $(\phi^-_\pi)^\perp$ contains $\{\ell_0^\perp\}$. 
In particular, this simplex $\phi_\pi^-$ is too crossed by the vertical line $\ell_0$. See Figure \ref{Fig:CaseB}.

Furthermore, recalling the notation of Section \ref{Subsec:2ndNetPunctured}, the simplex $\phi^-_\pi$ belongs to $V'(p_1)\times\ldots\times V'(p_d)$, thus meeting 
$\ell_0$ at a point of $X'_{\ell_0}$.
Since (i) all the vertices of $\phi_\pi^-\subset \conv\left(\varphi'_1\cup \ldots \cup \varphi'_d\right)$ lie above $H^+$, and (ii) all the vertices of $\pi$ lie above the lower tangent hyperplane $H\left(\phi^-_\pi\right)$ to $\Phi'(\pi)$, this intersection point $\phi^-_\pi\cap \ell_0\in X'_{\ell_0}$ must lie between the points $\ell_0\cap \tau^+=\ell_0\cap H^+$ and $\ell_0\cap \pi$ and, thus, within the interval $K\cap \ell_0$ (as both $\tau$ and $\pi$ belong ${P_K\choose d}$).

\begin{figure}
    \begin{center}
        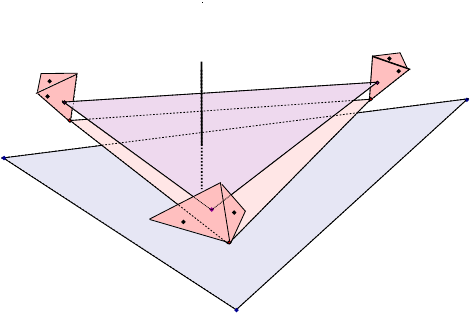  \caption{\small Proof of Lemma \ref{Lemma:Flat} in $\reals^3$ -- case (b). Showing that the vertical line  $\ell_0\in \L_i$ crosses $\phi^-_\pi$ within $K\cap \ell_0$. The line $\ell_0$ crosses the simplices $\pi\in \Pi^+_{\ell_0}$, $\phi^-_\pi$, and $\tau^+$ in this decreasing order of the last coordinate.}
      \label{Fig:CaseB}
    \end{center}
\end{figure}

To recap, every simplex $\pi=\conv(p_1,\ldots,p_d)\in \Pi^+_{\ell_0}$ has been assigned the unique simplex $\phi^-_\pi\in {V'\choose d}$ which supports the common lower tangent $H(\phi^-_\pi)$ to $\varphi'(p_1),\ldots,\varphi'(p_d)$, and intersects $\ell_0$ at the point $\phi^-_\pi\cap \ell_0\in X'_{\ell_0}\cap K$.
Since any simplex $\phi=\conv(v_1,\ldots,v_d)\in {V'\choose d}$ is ``assigned" in this manner to at most $\lfloor\eps_2n\rfloor^d$ simplices $\pi=(p_1,\ldots,p_d)\in \Pi^+_{\ell_0}$, which satisfy $\phi^-_\pi=\phi$ (as the choice of $v_1,\ldots,v_d$  uniquely determines the $d$ ambient {\it non-full} parts of $p_1,\ldots,p_d$ within $\P_0^+(K)$), the interval $\ell_0\cap K$ must contain at least

$$
\frac{|\Pi^+_{\ell_0}|}{\left(\lfloor\eps_2 n \rfloor\right)^d}=\Theta\left(\frac{(\delta\eps n/s)^d}{\left(\delta\eps n/\left(st^{1-1/d}\right)\right)^d}\right)=\Theta\left(t^{d-1}\right)\gg \frac{t^{d-d/\alpha_d}}{u^{(d-1)^2}}
$$

\noindent points $X_{\ell_0}$. Thus, given a sufficiently small choice of the constant $c_2>0$ in Step 2, which underlies the definition of the nets $N'_{\ell}$, the interval $\ell_0\cap K$ must contain a point of the net $N'_{\ell_0}\subset N_{>\delta,2}$.

\medskip
\medskip
\noindent {\bf Case (c).}  It, therefore, can be assumed, from now on, that the convex set $K\in \K_{>\delta}$ at hand falls into neither of the previous scenarios (a) and (b). 
Therefore, the full parts $P_{0,j}(K)$ of $\P_0(K)$ must encompass at least $9|P_0(K)|/10$ points of $P_0(K)$. 

Since the polytope $I_0(K)=\conv\left(\tau^-\cup \tau^+\right)$, within $K$, encompasses at most $c_2\eps_2n$ points of $P_0(K)$ (or, else, such a set $K$ would have been hit in Step 2 by the net $N_\strong(P,c_2\eps_2)$), each of the full parts $P_{0,j}(K)$ must contain at least $|P_{0,j}(K)|-c_2\eps_2n\geq 2|P_{0,j}(K)|/3$ such points that lie outside $I_0(K)$ (either above $\tau^+$ or below $\tau^-$).
For each of these full parts $P_{0,j}(K)$, it must be that either (i) at least $|P_{0,j}(K)|/3$ of its points lie above the simplex $\tau^+$ (and, in particular, above the supporting hyperplane $H^+=H(\tau^+)$), or (ii) at least $|P_{0,j}(K)|/3$ of its points lie below $\tau^-$ (and, in particular, below $H^-=H(\tau^-)$).

It can thus be assumed, with no loss of generality, that at least $2|P_0(K)|/5$ of the points of $P_0(K)$ fall in the full subsets $P_{0,j}(K)\in \P_0(K)$ of the former kind (so that at least $|P_{0,j}(K)|/3$ of their points lie above $\tau^+$ and $H^+$).
Let $\F_0(K)$ denote the collection of these subsets $P_{0,j}(K)$, so that

\begin{equation}\label{Eq:FullPartsCover}
	\left|\bigcup \F_0(K)\right|\geq 2|P_0(K)|/5.
\end{equation}

\noindent For every part $P_{0,j}(K)\in \F_0(K)$, let $P_{0,j}^+(K)$ denote its subset of at least 
$$
|P_{0,j}(K)|/3\geq \eps_2n/3\geq 30
$$ points that lie {\it above} $H^+$ and $\tau^+$ and, therefore, above every simplex $\tau\in \Sigma(K)$.
 (The second inequality uses (\ref{Eq:Full}).)

Any such full part $P_{0,j}(K)$ in $\F_0(K)$ determines, within $\Lambda_{0,j}(K)$, the sub-family 
$$
\Lambda^+_{0,j}(K):={P^+_{0,j}(K)\choose 2}
$$ 
\noindent which is comprised of
 at least $|\Lambda_{0,j}(K)|/20$ short 2-edges that lie above $\tau^+$. 
If such a part $P_{0,j}(K)$ belongs to $\P_0(K)\setminus \T_0(K)$ then, in particular, fewer than $\frac{1}{4}|\Lambda^+_{0,j}(K)|\cdot|\Sigma(K)|$ of the pairs $(\lambda,\tau)\in \Lambda_{0,j}^+(K)\times \Sigma(K)$ are tight.
As a result, at least one of the following two scenarios has to occur for $P_{0,j}(K)$:

\begin{enumerate}
	\item[(i)] at least  $|\Lambda_{0,j}^+(K)|\cdot |\Sigma(K)|/2$ of the pairs $(\lambda,\tau)\in \Lambda_{0,j}^+(K)\times \Sigma(K)$ are not parallel, or
	\item[(ii)]  at least $|\Lambda_{0,j}^+(K)|\cdot |\Sigma(K)|/4$ of the pairs in $(\lambda,\tau)\in \Lambda_{0,j}^+(K)\times \Sigma(K)$ are parallel, yet not tight.

\end{enumerate}

Let $\F'_0(K)$ (resp., $\F''_0(K)$) denote the sub-family of {\it all} such full subsets $P_{0,j}(K)\in \F_0(K)$ that satisfy condition (i) (resp., criterion (ii)).\footnote{For the sake of our subsequent case analysis, it is immaterial whether the families $\F'_0(K)$ and $\F''_0(K)$ are contained in $\T_0(K)$, as long as they cover sufficiently many points of $P_0(K)$.} 
Using that $|\bigcup\T_0(K)|\leq |P_0(K)|/5$, $\F_0(K)\setminus\T_0(K)\subseteq \F'_0(K)\cup \F''_0(K)$, and (\ref{Eq:FullPartsCover}), we conclude that

$$
\left|\bigcup \F'_0(K)\cup \F''_0(K)\right|\geq \left|\bigcup\F_0(K)\right|-\left|\bigcup\T_0(K)\right|\geq |P_0(K)|/5.
$$

To complete the proof of Lemma \ref{Lemma:Flat}, we distinguish between two possible sub-scenarios for $K$.


\medskip
\noindent {\bf Case (c1).} If at least one subset $P_{0,j}(K)\in \F_0(K)$ belongs to $\F'_0(K)$, we again argue that such a set $K$ must have been pierced by the net $N_\strong(P,c_2\eps_2)\subseteq N_{>\delta,2}$ (again, pending on a suitably small choice of the constant $c_2>0$ in Step 2).

\begin{figure}
    \begin{center}
        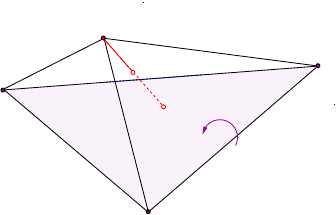  \caption{\small Proof of Lemma \ref{Lemma:Flat} in $\reals^3$ -- case (c1). The 2-edge $\lambda$ lies above $\tau$. Since $\lambda$ and $\tau$ are {\it not} parallel, a vertex $q$ of $\lambda$ must lie in the convex hull of $\tau$ and the remaining vertex of $\lambda$.}
      \label{Fig:CaseC1}
    \end{center}
\end{figure}

\medskip
To this end, fix a non-parallel pair $(\lambda,\tau)$, with $\lambda\in \Lambda^+_{0,j}(K)$ and $\tau\in \Sigma(K)$.
Since $\tau$ is still a boundary facet of the polytope $\conv(\tau\cup \lambda)$ (as the short 2-edge $\lambda$ lies in the halfspace above $H(\tau)$), there must be a vertex $p\in \lambda$ so that $\conv(\tau\cup \kappa)=\conv(\tau\cup \{p\})$; see Figure \ref{Fig:CaseC1}. Thus, $\conv(\tau\cup \kappa)$ must contain the remaining vertex $q$ of $\lambda$.    
 We, therefore, assign any such non-parallel pair $(\lambda,\tau)\in \Lambda^+_{0,j}(K)\times \Sigma(K)$ the ``smaller" configuration $(p,\tau)$, which involves only $d+1$ vertices, so that the ``missing" vertex $q$ of the short 2-edge $\lambda$ lies in the interior of $\conv\left(\tau\cup \{p\}\right)$. 
 
 Since the set $P_{0,j}^+(K)\times \Sigma(K)$ encompasses only $O\left(|P_{0,j}^+(K)|\cdot |\Sigma(K)|\right)$ possible configurations, which we have assigned to a total of $\Theta\left(|P_{0,j}^+(K)|^2\cdot |\Sigma(K)|\right)$ pairs $(\lambda,\tau)$, the pigeonhole principle yields such a configuration $(p,\tau)$ that has been assigned to $\Omega\left(|P_{0,j}^+(K)|\right)=\Omega(\eps_2n)$ distinct pairs $(\lambda=\{p,q\},\tau)$, each adding a distinct vertex $q\in P_{0,j}^+(K)$ into the interior of the simplex $\conv\left(\tau\cup \{p\}\right)$.
Given a small enough choice of $c_2>0$ in Step 2,  this simplex must contain at least $c_2\epsilon_2n$ points of $P_K$. However, such a convex set $K\in \K_{>\delta}$ must have been pierced by the strong net $N_\strong(P,c_2\eps_2)\subseteq N_{>\delta,2}$, and removed in Step 2.

\bigskip
\noindent {\bf Case (c2).} Let us assume, at last, that $\F'_0(K)=\emptyset$ and, therefore,
$$
\left|\bigcup \F''_0(K)\right|\geq |P_0(K)|/5.
$$

\noindent Once again, we will find a canonical line $\ell_0\in \L_0=\L(P_0,u)$ so that, given a suitably small constant $c_2>0$ in Step 2, the interval $K\cap \ell_0$ must contain a point of the net $N'_{\ell_0}\subset N_{>\delta,2}$.

To this end, fix a simplex  $\tau=\conv(p_1,\ldots,p_d)\in \Sigma(K)$, so that $p_i\in P_i(K)$ for all $1\leq i\leq d$. Refer to Figure \ref{Fig:CaseC2-Setup}.
Note that the vertices of $\tau$``mark" $d$ ambient clipped cells $\varphi'(p_1),\ldots,\varphi'(p_d)$, whose respective projections are contained in the respective cells $\Delta_1,\ldots,\Delta_d$, which surround $\Delta_0$ (so that $\Delta_0$ is, in particular, surrounded by the projections $\varphi'(p_1)^\perp,\ldots,\varphi'(p_d)^\perp$). 

Once again, let us consider the ``clipped envelope" of $\tau$ 
$$
\Phi'(\tau)=\conv\left(\varphi'(p_1)\cup\ldots\cup \varphi'(p_d)\right),
$$ 

\noindent which is obviously contained in $\Phi(\tau)=\conv\left(\bigcup_{i=1}^d\varphi(p_i)\right)$. According to Lemma \ref{Lemma:TangentHyperplanes}, $\Phi'(\tau)$ has a unique facet $\phi^+_\tau\in {V'\choose d}$ which supports the common upper tangent $H^+_\tau=H(\phi_\tau)$ to $\varphi'(p_1),\ldots,\varphi'(p_d)$, with the property that $\Delta_0\subseteq (\phi_\tau^+)^\perp$ (hence, every point of $P_0(K)$ lies either above or below the simplex $\phi^+_\tau$). We thus assign $\phi^+_\tau$ to $\tau$, and use $P_{0}^+(K;\tau)$ to denote the subset of all such points in $P_{0}^+(K)$ that lie above the respective simplex $\phi^+_\tau\in {V'\choose d}$ of $\tau$.

 \begin{figure}
    \begin{center}
    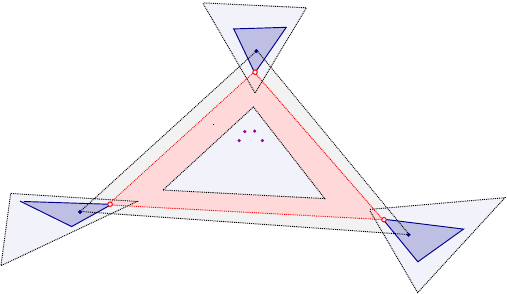\caption{\small Proof of Lemma \ref{Lemma:Flat} in $\reals^3$. Case (c2) -- a schematic view from above. Depicted is a simplex $\tau\in \Sigma(K)$.
    The projections of the clipped cells $\varphi'(p_1),\varphi'(p_2)$ and $\varphi'(p_3)$ surround the simplicial cell $\Delta_0$, which contains the set $P_0(K)$.
    Depicted are the ``upper" facet $\phi^+_\tau$ of $\Phi'(\tau):=\conv\left(\varphi'(p_1)\cup\varphi'(p_2)\cup\varphi'(p_3)\right)$, and the subset $P^+_0(K;\tau)$ of all such points in $P^+_0(K)\subseteq P_0(K)$ that lie above $\phi^+_\tau$.}
      \label{Fig:CaseC2-Setup}
    \end{center}
\end{figure}

In view of (\ref{Eq:Full}), any simplex of $\phi\in {V'\choose d}$ can be assigned in this manner at most $\lfloor \eps_1 n\rfloor^d$ simplices $\tau\in \Sigma(K)$, which all satisfy $\phi=\phi_\tau^+$.\footnote{The somewhat weaker bound, than in case of (b), is because the ambient clipped cells $\varphi'(p_i)$ of the vertices $p_i$ of $\tau$, for $1\leq i\leq d$, can be full and, thus, each contain up to $\lfloor\eps_1n\rfloor$ points of the respective subset $P_i(K)$.}
At the center of our argument lies the following property.

\begin{lemma}\label{Lemma:Above}
With the previous choice of the convex set $K\in \K_{>\delta}$ that falls into case (c2), the family $\Sigma(K)$ contains a subset $\tilde{\Sigma}(K)$ of at least
$|\Sigma(K)|/10^5$ simplices $\tau\in \Sigma(K)$ that each satisfy $|P_0^+(K;\tau)|\geq |P_0(K)|/10^5$.
 \end{lemma}
 \begin{proof}[Proof of Lemma \ref{Lemma:Above}]
By the pigeonhole principle, it is enough to find $|\Sigma(K)|\cdot|P_0(K)|/10^4$ such pairs $(\tau,p)\in \Sigma(K)\times P^+_0(K)$ that satisfy $p\in P_0^+(K;\tau)$.
Furthermore, as the parts of $\F''_0(K)$ cover at least $|P_0(K)|/5$ points of $P_0(K)$, it suffices to show that any part $P_{0,j}(K)\in \F_0(K)$ determines at least $|\Sigma(K)|\cdot |P_{0,j}(K)|/2000$ pairs $(\tau,p)\in \Sigma(K)\times P_{0,j}(K)$, in which $p$ lies above $\phi_\tau^+$. 

To see the latter claim, we keep the subset $P_{0,j}(K)\in \F''_0(K)$ fixed, and use $\Upsilon_{j}$ to denote the associated subset of at least $|\Lambda^+_{0,j}(K)|\cdot |\Sigma(K)|/4$ 
parallel, albeit loose pairs $(\lambda,\tau)\in \Lambda_{0,j}^+(K)\times \Sigma(K)$.
The crucial observation is that, in any pair $(\lambda,\tau)\in\Upsilon_{j}$, at least one of the vertices of $\lambda$ must lie above $\Phi'(\tau)$ and, thereby, also above $\phi^+_\tau$, since (i) $\lambda$ lies above $\tau$, (ii) $\Phi'(\tau)$ is contained in the envelope $\Phi(\tau)$ of $\tau$, which does not contain $\lambda$ (as the pair $(\lambda,\tau)$ is loose), and (iii) we have that $\lambda^\perp\subseteq \Delta_0^\perp\subseteq (\phi^+_\tau)^\perp$ (so that every vertex of $\lambda$ that lies outside $\Phi'(\tau)\supset \tau$, has to lie above $\phi^+_\tau$). 

 By the pigeonhole principle, there is a sub-collection $\tilde{\Sigma}_j(K)$ of at least $|\Sigma(K)|/16$ simplices $\tau\in \Sigma(K)$, so that each of them is part of at least $|\Lambda^+_{0,j}(K)|/16$ loose pairs $(\lambda,\tau)\in \Upsilon_j$.
Fix $\tau\in \tilde{\Sigma}_j(K)$. Then any short 2-edge $\lambda\in \Lambda_{0,j}^+(K)$, that completes a loose pair $(\lambda,\tau)\in \Upsilon_j$ with $\tau$, must include at least one vertex of $P^+_0(K;\tau)$ (as depicted in Figure \ref{Fig:TightLoose3D} (left)). 
Therefore, to allow at least $|\Lambda^+_{0,j}(K)|/16={|P_{0,j}^+(K)|\choose 2}/16$ choices of $\lambda\in \Lambda_{0,j}^+(K)$ as above, the subset $P_{0,j}^+(K;\tau)$ must contain at least $|P_{0,j}^+(K)|/32\geq |P_{0,j}(K)|/96$ points which lie above $\phi_\tau^+$. 
Repeating this argument for at least $|\Sigma(K)|/16$ simplices $\tau\in \tilde{\Sigma}_j(K)$ yields a total of at least $|\Sigma(K)|\times |P_{0,j}(K)|/2000$ pairs $(\tau,p)\in \Sigma(K)\times P_{0,j}^+(K)$ in which $p\in P_0^+(K,\tau)$.
\end{proof}

 \begin{figure}
    \begin{center}
 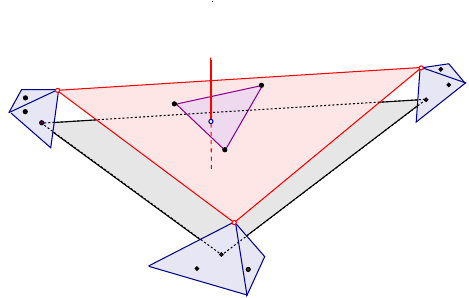  \caption{\small Proof of Lemma \ref{Lemma:Flat} in $\reals^3$ -- case (c2). Arguing that the intersection $\phi_\tau^+\cap \ell_0$ lies in $K\cap \ell_0$. The line $\ell_0$ crosses the triangles $\pi\in {P^+_0(K,\tau)\choose 3}$, $\phi_\tau^+$, and $\tau$ in this decreasing order of the last coordinate. 
 }
      \label{Fig:CaseC2}
    \end{center}
\end{figure}

In view of the estimate (\ref{Eq:ManyPointsLocal}), it follows, for each of the simplices $\tau\in \tilde{\Sigma}(K)$ in Lemma \ref{Lemma:Above}, the cardinality of its respective subset $P^+_0(K;\tau)$ satisfies
$$
|P_0^+(K;\tau)|\geq |P_0(K)|/10^5\geq \max\{\varepsilon|P_0|/\vartheta,d\}.
$$
Let $\tau\in \tilde{\Sigma}(K)$, and consider the family $\Pi_0^+(K;\tau):={P_0^+(K;\tau)\choose d}$ of $(d-1)$-simplices over the set $P_0^+(K;\tau)$.
Using that $K$ is $\varepsilon$-spread within $\V(P_0,u)$ in a manner that is similar to case (b), it follows that the $\vartheta$-canonical family $\L_0=\L(P_0,u)$ must include a line $\ell$ piercing at least {\it one} simplex $\pi\in \Pi_0^+(K;\tau)$. 
Furthermore, since any such simplex $\pi\in \Pi^+_0(K;\tau)$ lies above the simplex $\phi^+_\tau$ (whose supporting hyperplane $H\left(\phi^+_\tau\right)$ lies above $\Phi'(\tau)\supset \tau$), this line $\ell$ must meet
$\pi$, $\phi_\tau^+$, and $\tau$ in this decreasing order of the $d$-th axis, and within the interval $K\cap \ell$, as depicted in Figure \ref{Fig:CaseC2}.

We thus have found a subset $\tilde{\Sigma}(K)$ of at least $|\Sigma(K)|/10^4=\Theta\left(\left(\delta \eps n/s\right)^d\right)$ simplices $\tau$, and assigned each of them a simplex $\phi^+_\tau\in {V'\choose d}$ meeting some line $\ell\in \L_0$ at a point of $X'_\ell\cap K$. As any simplex $\phi=(v_1,\ldots,v_d)\in {V'\choose d}$ has been assigned in this manner to at most $\left(\lfloor\eps_1n\rfloor\right)^d$ simplices $\tau=(p_1,\ldots,p_d)\in \tilde{\Sigma}(K)$, there must be such a line $\ell_0\in \L_0$ whose cross-section $\ell_0\cap K$ with $K$ encompasses at least

$$
\frac{|\tilde{\Sigma}(K)|}{\left(\lfloor\eps_1 n \rfloor\right)^d\cdot |\L_0|}=\Theta\left(\frac{(\delta\eps n/s)^d}{\left(\delta\eps n/\left(st^{1-1/\alpha_d}\right)\right)^d\cdot |\L_0|}\right)=\Theta\left(\frac{t^{d-d/\alpha_d}}{u^{(d-1)^2}}\right)$$

\noindent points of $X'_{\ell_0}$. (Here the first bound uses the definition (\ref{Eq:1stEpsPunct}) of $\eps_1>0$ in Step 1.) Thus, provided a sufficiently small choice of $c_2>0$ in the definition of the nets $N'_{\ell}$, the interval $\ell_0\cap K$ must again contain a point of the net $N'_{\ell_0}\subset N_{>\delta,2}$, which completes the analysis of the last scenario (c2).

\medskip
In conclusion, any convex set $K\in \K_{>\delta}$ that fails to meet the criteria of flatness in Step 3, must fall into at least one of the four scenarios (a), (b), (c1), or (c2), that arise in our analysis, and, pending on a suitably small choice of the constant $c_2>0$ in Step 2, be dispatched by the combination $N_{>\delta,1}\cup N_{>\delta,2}$ of their previous nets. $\Box$

\subsection{Proof of Lemma \ref{Lemma:SparsePi}}\label{Subsec:SparsePi}
The proof of Lemma \ref{Lemma:SparsePi} comes down to describing a sufficiently sparse subset $\tilde{\Pi}\subseteq {P\choose d}$ that would still be ``locally dense" over every principal subset $P_K$ that is cut out by some flat convex set $K\in \K_{>\delta}$. 

\medskip
\noindent{\bf The subset $\tilde{\Pi}=\tilde{\Pi}(P,\eps,\delta,s,t)$.}  To obtain such a subset $\tilde{\Pi}\subseteq {P\choose d}$, which would not depend on the flat convex set $K\in \K_{>\delta}$ at hand, let us first consider the family 
 $$
\Lambda(P,s,t):=\biguplus_{P_{i,j}\in \P(P,s,t)} {P_{i,j}\choose 2},
$$

\noindent which is comprised of all such 2-edges in ${P\choose 2}$ whose two vertices come from the {\it same part $P_{i,j}$} in our secondary partition $\P(P,s,t)$ of Step 1. Notice that

\begin{equation}\label{Eq:Lambda}
|\Lambda(P,s,t)|=O\left(\sum_{1\leq i\leq s}\sum_{1\leq j\leq t}{\lceil n/(st)\rceil \choose 2}\right)=O\left(\frac{n^2}{st}\right).
\end{equation}

\medskip
\noindent{\bf Definition.}  For every point $p\in P$, let us denote

\begin{center}
$\displaystyle \Pi_p:=\left\{\pi\in {P\choose d}\mid p\in \pi\right\}$, {and} $\displaystyle\Lambda_p:=\{\lambda\in \Lambda(P,s,t)\mid p\in \lambda\}$.
\end{center}

Namely, the set $\Pi_p$ (resp., $\Lambda_p$) is comprised of all the $d$-dimensional simplices in ${P\choose d}$ (resp., 2-edges in $\Lambda(P)$) that are adjacent to $p$. Note that every $\pi\in \Pi_p$ must be of the form $\conv(p\cup\pi_p)$, where $\pi_p$ denotes some $(d-2)$-simplex in ${P\setminus \{p\}\choose d-1}$.

\bigskip
The subset $\tilde{\Pi}(P,\eps,\delta,s,t)\subseteq {P\choose d}$ will be determined by the means of  a delicate ``local" relation between the simplices of $\Pi_p$ and the short 2-edges of $\Lambda_p$, that holds for each $p\in P$. 

\medskip
\noindent {\bf Definition.} Let $p\in P$.
We say that a short 2-edge $\lambda\in \Lambda_p$ and a $(d-1)$-simplex $\pi=\conv(p\cup \pi_p)\in \Pi_p$ are {\it friendly} at their shared vertex $p$ if 
there exists a hyperplane $H$ through $\lambda$ whose zone within the secondary partition $\P(P,s,t)$ encompasses the $d-1$ vertices $p_i$ of $\pi_p=\conv(p_1,\ldots,p_{d-1})$; namely, $H$ must contain both vertices of $\lambda$, and cross the ambient cells $\varphi(p_i)$, for $1\leq i\leq d-1$ (which need not be distinct). See Figure \ref{Fig:Friendly}.

\begin{figure}
    \begin{center}
        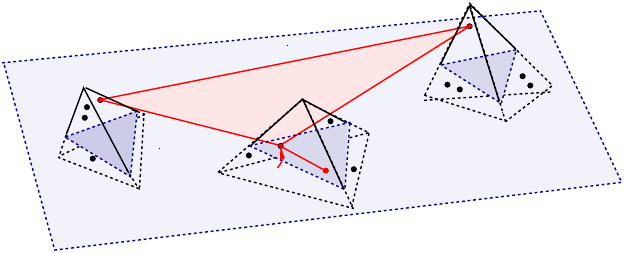  \caption{\small Proof of Lemma \ref{Lemma:SparsePi} in $\reals^3$. The 2-edge $\lambda\in \Lambda_p$ and the triangle $\pi\in \Pi_p$ are friendly at the shared vertex $p$ in $\reals^3$. Depicted is the plane $H$ through $\lambda$ that crosses the ambient cells $\varphi(p_i)$ of the remaining $2$ vertices of $\pi$.}
      \label{Fig:Friendly}
    \end{center}
\end{figure}

\medskip
Let us now fix a sufficiently small constant $c_3>0$, which will be determined in the end of this proof (with implicit dependence on the constant $c_2>0$ in Step 2), and say that a simplex $\pi\in \Pi_p$ is {\it rich} at its vertex $p$ if there exist at least 

\begin{equation}\label{Eq:Popular}
c_3\cdot \frac{\delta\eps n}{st^{1-1/d}}\gg \eps_2 n
\end{equation}

\noindent simplices $\lambda\in \Lambda_p$ that are friendly with $\pi$ at $p$.

\medskip
For each $p\in P$, let $\tilde{\Pi}_p\subseteq {P\choose d}$ denote the subset of all such simplices $\pi\in {P\choose d}$ that are rich at $p$. We then define

$$
\tilde{\Pi}(P,\eps,\delta,s,t):=\bigcup_{p\in P} \tilde{\Pi}_p.
$$

It, therefore, remains to show that the just constructed subset $\tilde{\Pi}(P,\eps,\delta,s,t)\subseteq {P\choose d}$ meets both criteria of Lemma \ref{Lemma:SparsePi} (pending on a sufficiently small choice of the constant $c_3>0$ in (\ref{Eq:Popular})).

\bigskip
\noindent {\bf Part (i) -- bounding $|\tilde{\Pi}(P,\eps,\delta,s,t)|$.}  To express the cardinality of $\tilde{\Pi}(P,\eps,\delta,s,t)$ in the terms of the much smaller quantity $|\Lambda(P,s,t)|$, we devise a {\it charging scheme} which assigns every $\pi\in \tilde{\Pi}(P,\eps,\delta,s,t)$ to multiple short 2-edges in $\Lambda(P,s,t)$ that share a vertex with $\pi$.

For each point $p\in P$, we {\it assign} every simplex $\pi\in \tilde{\Pi}_p$ that is rich at $p$, to every short 2-edge $\lambda\in \Lambda_p$ that is friendly with $\pi$ at $p$; each of these short 2-edges $\lambda$ {\it pays} 1 unit of charge to $\pi$, via their shared vertex $p$.
By definition, every $\pi\in \tilde{\Pi}(P,\eps,\delta,s,t)$ is rich at (at least) one of its $d$ vertices $p\in \tau$ and, therefore, is assigned to at least
$c_3\cdot \frac{\delta\eps n}{st^{1-1/d}}$ 
\noindent short 2-edges $\lambda\in \Lambda_p$. It, therefore, remains to bound the maximum number of simplices in $\tilde{\Pi}(P,\eps,\delta,s,t)$ that have been assigned in the above manner, to any given short 2-edge $\lambda\in \Lambda(P,s,t)$.

\begin{lemma}\label{Prop:TimesCharged}
Any short 2-edge $\lambda\in \Lambda(P,s,t)$ is assigned a total of $O\left(n^{d-1}/\s^{1/d}\right)$ simplices of $\tilde{\Pi}(P,\eps,\delta,s,t)$.
\end{lemma}
\begin{proof}
Fix a short 2-edge $\lambda=\{p,q\}\in \Lambda(P,s,t)$. Note that, for every simplex $\pi=\conv(p,p_1,\ldots,p_{d-1})=\conv\left(p,\pi_p\right)\in \tilde{\Pi}(P,\eps,\delta,s,t)$, that has been assigned to $\lambda$ via $p$, 
the remaining $d-1$ vertices $p_i$ of $\pi_p=(p_1,\ldots,p_{d-1})$ lie the zone of some hyperplane $H$ through $\lambda$. (In other words, $H$ must contain both endpoints of $\lambda$, and also cross the ambient cell $\varphi(p_i)$ of every vertex $p_i\in \pi_p$.) 

Furthermore, according to Lemma \ref{Lemma:ExtremalHyperplane}, there must exist such a hyperplane $H$ that passes through $p,q$, and some $d-2$ vertices $v_1,\ldots,v_{d-2}$ of the (not necessarily distinct) cells $\varphi(p_i)$, with $1\leq i\leq d-1$. In other words, $H$ is fully determined by the ``signature" which is comprised of $p,q$ and some $d-2$ vertices of the simplicial cells in the secondary partition $\P(P,s,t)$ of Step 2. 
It can, furthermore, be assumed, up to relabeling, that the vertices $v_1,\ldots,v_{d-2}$ come entirely from the boundaries of the cells $\varphi(p_1),\ldots,\varphi(p_{d-2})$. Hence, given $p$, $q$, and the first $d-2$ vertices $p_1,\ldots,p_{d-2}$ of $\pi_p$, the hyperplane $H$ can be ``guessed" in up to $O({(d+1)(d-2)\choose d-2})$ distinct ways via the respective cells $\varphi(p_1),\ldots,\varphi(p_{d-2})$. 

Lastly, since any hyperplane crosses only $O(st^{1-1/d})$ simplicial cells in $\P(P,s,t)$, fixing the points $p,q,p_1,\ldots,p_{d-2}$, together with $H$, leaves only $O\left(\frac{n}{st}\cdot st^{1-1/d}\right)=O\left(n/t^{1/d}\right)$ possible choices for the remaining vertex $p_{d-1}$ from within the zone of $H$. Thus, the number of the possible choices of $\pi$ cannot exceed
$O\left(n^{d-2}\cdot n/t^{1/d}\right)=O\left(n^{d-1}/t^{1/d}\right)$.
\end{proof}

\noindent In view of the just established Lemma \ref{Prop:TimesCharged}, and the upper bound (\ref{Eq:Lambda}) on the cardinality of $\Lambda(P,s,t)$, the short 2-edges $\lambda\in \Lambda(P,s,t)$ transfer via their vertices a total of 
$$
O\left(|\Lambda(P,s,t)|\cdot \frac{n^{d-1}}{t^{1/d}}\right)=O\left(\frac{n^{d+1}}{st^{1+1/d}}\right)
$$
\noindent units of charge to the friendly rich simplices $\pi\in \tilde{\Pi}$ (each time via a common vertex $p$ of $\lambda$ and $\pi$, which may not be unique). Since any simplex of $\tilde{\Pi}$ receives at least
$c_3\cdot \frac{\delta\eps n}{st^{1-1/d}}$
of these units of charge, via one or more of its vertices, it follows that 

$$
|\tilde{\Pi}(P,\eps,\delta,s,t)|=O\left(\frac{n^{d+1}}{st^{1+1/d}}\cdot \frac{st^{1-1/d}}{\delta\eps n}\right)=O\left(\frac{n^d}{\delta\eps t^{2/d}}\right),
$$

\noindent which concludes the proof of the first part of Lemma \ref{Lemma:SparsePi}.


\bigskip
\noindent{\bf Part (ii).} It, therefore, remains to show, for any flat and $\delta$-punctured convex set $K\in \K_{>\delta}$, that its induced set ${P_K\choose d}$ encompasses sufficiently many simplices of the just defined family $\tilde{\Pi}=\tilde{\Pi}(P,\eps,\delta,s,t)\subseteq {P\choose d}$. 

To this end, recall that every set $K$ as above is endowed with the principal $(d+1)$-tuple $\Psi(K)=(P_0,\ldots,P_{d})$. Since $K$ is flat, the secondary $t$-partition $\P_0=\P(P_0,t)$ of $P_0$ (or, more precisely, the partition $\P_0(K)$ of $P_0(K)$ by $\P_0$) yields the subset $\T_0(K)\subseteq \P_0(K)$, whose tight parts $P_{0,j}(K)$ together encompass at least $|P_0(K)|/5$ points of $P_0(K)$. Specifically, every part $P_{0,j}(K)\in \T_0(K)$ is full (that is, it satisfies $|P_{0,j}(K)|\geq \eps_2 n\geq 100$), and yields at least $|\Lambda_{0,j}(K)|\cdot|\Sigma(K)|/80$ tight pairs $(\lambda,\tau)\in \Lambda_{0,j}(K)\times\Sigma(K)$.

\smallskip
At the heart of our argument lies the following property.

\begin{lemma}\label{Lemma:TightCrosses}
Let $K\in \K_{>\delta}$ be a flat convex set, and $P_{0,j}(K)\in \T_0(K)$ a tight subset, and $(\lambda,\tau)$ a tight pair in $\Lambda_{0,j}(K)\times\Sigma(K)$. Then there exists a hyperplane through $\lambda$ that crosses at least $d-1$ among the ambient cells $\varphi(p_i)$ of the vertices of $\tau=\conv(p_1,\ldots,p_d)$. 
\end{lemma}

\begin{proof}
	Refer to Figure \ref{Fig:CrossSidesTight}. Recall that, by the definition of a tight pair, both vertices of the short 2-edge $\lambda=\{p,q\}$ must lie in the envelope $\Phi(\tau)=\conv\left(\varphi(p_1)\cup\ldots\cup \varphi(p_d)
	\right)$ (which contains the $d$ {\it distinct} simplicial cells $\varphi(p_i)$ within the respective secondary partitions $\P_i=\P(P_i,t)$, for $1\leq i\leq d$). Furthermore, $\tau$ is missed by the line $\ell_{\lambda}=\aff(\lambda)$. Projecting the cells $\varphi(p_i)$ in the direction of $\ell_\lambda$ yields $d$ convex polytopes $\overline{\varphi}(p_i)$ within $\reals^{d-1}$, while $\ell_\lambda$ similarly projects to a point $\overline{\ell}_\lambda\in \reals^{d-1}$.

	\begin{figure}
    \begin{center}
        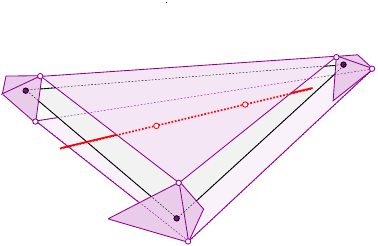\hspace{2cm}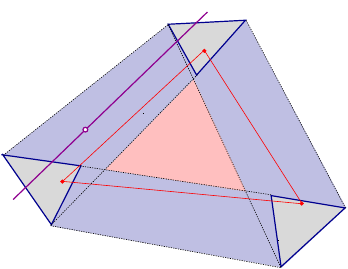  \caption{\small Lemma \ref{Lemma:TightCrosses} in $\reals^3$. Left: The pair $(\lambda,\tau)$ is tight, so the line $\ell_\lambda$ through the 2-edge $\lambda$ must cross the convex hull of some 2 simplicial cells among $\varphi(p_1),\varphi(p_2),\varphi(p_3)$. Right: Proof of the lemma. Depicted are the projections  $\overline{\varphi}(p_i)$ of the cells $\varphi(p_i)$ in the direction of $\ell_\lambda$, along with the intercept $\overline{\ell}_\lambda$ of $\ell_\lambda$. Since $\overline{\ell}_\lambda\not\in \overline{\tau}$, there exists a $1$-plane (i.e., a line) $\overline{H}$ through $\overline{\ell}_\lambda$ that intersects some $2$ among the projected sets $\overline{\varphi}(p_i)$.}
      \label{Fig:CrossSidesTight}
    \end{center}
\end{figure}

	 Note that $\overline{\ell}_\lambda$ cannot be surrounded by $\overline{\varphi}(p_1),\ldots,\overline{\varphi}(p_d)$. Indeed, in the contrapositive scenario $\overline{\ell}$ would lie within the projection $\overline{\tau}$ of $\tau$, whose $d$ vertices are selected from the $d$ distinct polytopes $\overline{\varphi}(p_i)$, in which case the line $\ell_\lambda$ would intersect $\tau$.
	 
	Since $\overline{\ell}_\lambda$ clearly lies within the projection $\overline{\Phi}(\tau)=\conv\left(\overline{\varphi}(p_1)\cup\ldots\cup\overline{\varphi}(p_d)\right)$ of $\Phi(\tau)$ (as $\Phi(\tau)$ contains $\lambda\subset \ell_\lambda$), Lemma \ref{Prop:NotSurrounded} yields a $(d-2)$-plane $\overline{H}$ within $\reals^{d-1}$ that contains $\overline{\ell}_\lambda$, and intersects some $d-1$ among the $d$ polytopes $\overline{\varphi}(p_i)$, say, $\overline{\varphi}(p_1),\ldots,\overline{\varphi}(p_{d-1})$. (See Figure \ref{Fig:CrossSidesTight} (right).) Then the respective $d-1$ cells $\varphi\left(p_1\right),\ldots,\varphi\left(p_{d-1}\right)$ must be intersected  in $\reals^d$ by the hyperplane $H$, through $\lambda$, that is obtained by ``lifting" $\overline{H}$ in the direction of $\ell_\lambda$. \end{proof}

\medskip
\noindent{\bf Definition.} Let us fix a flat convex set $K\in \K_{>\delta}$. Denote $\tilde{P}(K):=\biguplus_{i=1}^d P_i(K)$. For every tight subset $P_{0,j}(K)\in \T_0(K)$, let $\Pi_{j}(K)$ denote the family $P_{j}(K)\ast {\tilde{P}(K)\choose d-1}\subset {P_K\choose d}$ which is comprised of all such simplices that exactly one of their vertices belongs to $P_{0,j}(K)$, whereas the rest $d-1$ vertices are chosen from $\tilde{P}(K)$. 

\medskip
The asserted claim for $K$ will follow from the fact that, provided a suitably small choice of $c_3>0$ in (\ref{Eq:Popular}), every tight subset $P_{0,j}(K)$ within $\T_0(K)$ contributes to $\tilde{\Pi}\cap {P_K\choose d}$ a subfamily $\tilde{\Pi}_j(K)$ of 
$\Theta\left(|\Pi_j(K)|\right)=\Theta\left(|P_{0,j}(K)|\cdot |\tilde{P}(K)|^{d-1}\right)\gg \eps_2n\cdot \left(\delta\eps n/s\right)^{d-1}$ such simplices $\pi\in \Pi_j(K)$ that are rich at their only vertex within $P_{0,j}(K)$. 
The following lemma implies that an average simplex $\pi\in \Pi_j(K)$ is indeed friendly with $\Omega\left(|P_{0,j}(K)|\right)\gg \eps_2n$ short edges $\lambda\in \Lambda_{0,j}(K)$, at their shared vertex within $P_{0,j}(K)$. 

\begin{lemma}\label{Claim:ManyFriendly}
	Let $K\in \K_{>\delta}$ be a tight convex set, and $P_{0,j}(K)\in \T_0(K)$ be a tight subset. Then there exist 
$$
\Omega\left(|\Lambda_{0,j}(K)|\cdot |\tilde{P}(K)|^{d-1}\right)=\Omega\left(\left|P_{0,j}(K)\right|^2\cdot \left(\delta\eps n/s\right)^{d-1}\right)
$$ 
friendly pairs $(\lambda,\pi)$ with $\lambda\in \Lambda_{0,j}(K)$ and $\pi\in \Pi_{j}(K)$. 	
\end{lemma}

\begin{proof}
Since the subset $P_{0,j}(K)$ is tight, there exist at least 
$|\Lambda_{0,j}(K)\times \Sigma(K)|/80$
tight pairs $(\lambda,\tau)\in \Lambda_{0,j}(K)\times \Sigma(K)$. 
For any short 2-edge $\lambda\in \Lambda_{0,j}(K)$, let $\Sigma_\lambda(K)$ denote the subset of all such simplices $\tau\in \Sigma(K)$ that yield a tight pair $(\lambda,\tau)$.
By the pigeonhole principle, $\Lambda_{0,j}(K)$ must contain a subset $\tilde{\Lambda}_{0,j}(K)$ of at least $\frac{1}{400}|\Lambda_{0,j}(K)|$ such short 2-edges $\lambda$ that satisfy $|\Sigma_\lambda(K)|\geq \frac{1}{400}|\Sigma(K)|$. 

We claim that any short 2-edge $\lambda=\conv(p,q)\in \tilde{\Lambda}_{0,j}(K)$ must lie in a hyperplane $H_\lambda$ whose zone, within the secondary partition $\P(P,s,t)$, encompasses a subset $\tilde{P}_\lambda(K)\subseteq \tilde{P}(K)$ of cardinality $\Omega\left(|\tilde{P}(K)|\right)$.
To see this, let $\tau\in \Sigma_\lambda(K)$. By the previous Lemma \ref{Lemma:TightCrosses}, there is a hyperplane $H$ through $\lambda$ that crosses at least $d-1$ ambient cells $\varphi(p_i)$ of the vertices of $\tau=\conv(p_1,\ldots,p_d)$. Furthermore, Lemma \ref{Lemma:ExtremalHyperplane} yields such a hyperplane $H=H_{\lambda,\tau}$ that, in addition, passes through $p,q$, and some $d-2$ vertices of the cells $\varphi(p_i)$. It thus can be assumed, with no loss of generality, that all of these vertices come from the boundaries of the first $d-2$ cells $\varphi\left(p_1\right),\ldots,\varphi\left(p_{d-2}\right)$. 
Hence, the hyperplane $H$ can be ``guessed" from $p,q$ and $p_1,\ldots,p_{d-2}$ in $O\left({d(d-2)\choose d-2}\right)$ ways. 

We keep $\lambda\in \tilde{\Lambda}_{0,j}(K)$ fixed, and assign to each simplex $\tau\in \Sigma_\lambda(K)$ a unique signature $S_{\tau}=\{p_1,\ldots,p_{d-2}\}$ in ${\tilde{P}(K)\choose d-2}$. The pigeonhole principle then yields such a signature $S=\{p_1,\ldots,p_{d-2}\}\in {\tilde{P}(K)\choose d-2}$ that is shared by a subset $\Sigma_S\subseteq \Sigma_\lambda(K)$ of $\Theta\left(|\tilde{P}(K)|^2\right)$ simplices $\tau$ (which all satisfy $S_\tau=S$). Note that, for each of these simplices $\tau\in \Sigma_{S}$, at least one of the two ``free" vertices in the pair $\tau\setminus S$ must fall in the zone of the hyperplane $H=H_{\lambda,\tau}$. 

Repeating this argument for all $\tau\in \Sigma_S$, and using the pigeonhole principle one more time, yields (i) a hyperplane $H$ that contains $\lambda$ along with some $d-2$ vertices chosen from the boundaries of $\varphi(p_1),\ldots,\varphi(p_{d-2})$, and (ii) a subset $\tilde{P}_\lambda(K)$ of $\Omega\left(|\tilde{P}(K)|^2\right)$ such pairs in ${\tilde{P}(K)\choose 2}$ that at least one of their vertices lies in the zone $H$. As a result, the zone of $H$ must contain a subset $\tilde{P}_\lambda(K)\subseteq \tilde{P}(K)$ whose cardinality is $\Theta\left(|\tilde{P}(K)|\right)$.

The asserted claim now stems from the following two facts: (i) any combination of $\lambda=\{p,q\}\in \tilde{\Lambda}_{0,j}(K)$ with a $(d-1)$-size subset $\mu\in {\tilde{P}_\lambda(K)\choose d-1}$ yields $2$ friendly pairs $(\lambda,\tau)\in \Lambda_{0,j}(K)\times \Pi_j(K)$, with $\tau\in \{\mu\cup \{p\},\mu\cup \{q\}\}$, and (ii) a friendly pair $(\lambda,\tau)$ is derived from at most $O(1)$ combinations $(\lambda,\mu)$, all of which involve the same set $\lambda\cup \tau$ of $d+1$ points.
\end{proof}

To complete the proof of Lemma \ref{Lemma:SparsePi}, let us fix a tight subset $P_{0,j}(K)\in \T_0(K)$. 
For every simplex $\pi\in \Pi_{j}(K)$ let $\Lambda_{0,j}(K;\pi)$ denote the subset of all such simplices $\lambda\in \Lambda_{0,j}(K)$ that are friendly with $\pi$ at its only vertex that belongs to $P_{0,j}(K)$.
Since $\left|\Pi_{j}(K)\right|=\Theta\left(|P_{0,j}(K)|\cdot\left(\delta\eps n/s\right)^{d-1}\right)$ and $|P_{0,j}(K)|\geq \epsilon_2 n=\Omega\left(\delta\eps n/\left(st^{1-1/d}\right)\right)$, Lemma \ref{Claim:ManyFriendly} yields, via the pigeonhole principle, a subset $\tilde{\Pi}_{j}(K)\subset \Pi_{j}(K)$ of $\Theta\left(|P_{0,j}(K)|\cdot\left(\delta\eps n/s\right)^{d-1}\right)$ 
such simplices $\pi\in \Pi_{j}(K)$ that satisfy $|\Lambda_{0,j}(K;\pi)|=\Theta\left(|P_{0,j}(K)|\right)=\Omega\left(\delta\eps n/\left(st^{1-1/d}\right)\right)$. 

\smallskip
A sufficiently small choice of the constant $c_{3}>0$ in (\ref{Eq:Popular})
guarantees that every $\pi\in \tilde{\Pi}_{j}(K)$ satisfies
$$
|\Lambda_{0,j}(K;\pi)|\geq c_{3}\cdot \frac{\delta\eps n}{st^{1-1/d}};
$$
\noindent that is, every $\pi\in \tilde{\Pi}_{j}(K)$ is rich at its only vertex within $P_{0,j}(K)$. As the subsets $\tilde{\Pi}_{j}(K)$ are clearly disjoint for all $1\leq j\leq t$, every tight subset $P_{0,j}(K)\in \T_0(K)$ contributes $|\tilde{\Pi}_{j}(K)|=\Theta\left(|\Pi_{j}(K)|\right)$ simplices to $\tilde{\Pi}\cap {P_K\choose d}$, for a total amount of

$$
\sum_{P_{0,j}(K)\in \T_0(K)}|\tilde{\Pi}_{j}(K)|\gg  \sum_{P_{0,j}(K)\in \T_0(K)}|\Pi_{j}(K)|\gg \sum_{P_{0,j}(K)\in \T_0(K)}|P_{0,j}(K)|\cdot \left(\frac{\delta\eps n}{s}\right)^{d-1}\gg
$$

$$
\gg |P_0(K)|\cdot \left(\frac{\delta\eps n}{s}\right)^{d-1}=\Theta\left(\frac{\delta^d\eps^dn^d}{s^d}\right).
$$

\noindent Here the third bound uses that $|\biguplus \T_0(K)|\geq |P_0(K)|/5$ for every flat convex set $K\in \K_{>\delta}$. This concludes the proof of Lemma \ref{Lemma:SparsePi}. $\Box$

\subsection{Wrap up for the $\delta$-punctured sets} \label{Subsec:PuncturedWrapUp}
To finish the proof of Theorem \ref{Theorem:BoundPunct}, let us first point out that, according to Lemma \ref{Lemma:ThirdNetPunctured}, every $\varepsilon$-spread and $\delta$-punctured set $K\in K_{>\delta}$ is pierced by the net 
$N_{>\delta}:=N_{>\delta,1}\cup N_{>\delta,2} \cup N_{>\delta,3}$. Putting together the estimates in Lemmas \ref{Lemma:1stNetPunctured}, \ref{Lemma:2nsNetPunctured}, and \ref{Lemma:ThirdNetPunctured} yields

\begin{equation}\label{Eq:PunctWrapUp}
|N_{>\delta}|\leq t^{1+\eta}\cdot f\left(\eps\cdot t^{1/\alpha_d-\eta}\right)+u^{1+\eta}\cdot f\left(\eps \cdot u^{1/(d-1)-\eta}\right)+
\end{equation}
$$
+O\left(t^{d/\alpha_d+\eta/(10d)}+\frac{t^{1-1/d}}{\eps^{1+\eta}}+\frac{1}{\eps^{d+1+\eta}t^{2/d}}\right),
$$

\noindent where $t=t_d$ is set to $\left\lceil 1/\eps^{\theta_d}\right\rceil$, for some constant $d/2\leq \theta_d\leq \alpha_d$ that is left to be determined for all $d\geq 3$. 

Since all the recursive terms in (\ref{Eq:PunctWrapUp}) are of the ``right" sort $a^{1+\eta}\cdot f\left(\eps\cdot b^{1-\eta}\right)$, with $b\geq a^{\frac{1}{\alpha_d}}$, it suffices to fix such an exponent $d/2\leq \theta_d\leq \alpha_d$ that would render all the three non-recursive terms $O\left((1/\eps)^{\alpha_d+\eta}\right)$.

For $d=3$, fixing $\theta_3:={2.165}$ renders all the non-recursive terms $O\left(1/\eps^{2.558+\eta}\right)=O\left(1/\eps^{\alpha_d+\eta}\right)$.

For $d=4$, we use $\theta_4:=3.02$, and similarly check that all the non-recursive terms are $O\left((1/\eps)^{3.48+\eta}\right)$.

In any higher dimension $d\geq 5$, we choose $\theta_d:=d(d+1-\alpha_d)/2=d(3/4+o_d(1))$ so as to directly render the third non-recursive term $O\left((1/\eps)^{\alpha_d+\eta}\right)$. 
It suffices to check the first two non-recursive terms are $O\left(1/\eps^{\alpha_d+\eta}\right)$.

Indeed, as we have that
$$
\theta_d=d(d+1-\alpha_d)/2=\frac{d}{4}\left(d+2-\sqrt{d^2-2d}\right)
$$

\noindent and, in particular, $\theta_d\ll d-1$ for all $d\geq 5$, the second non-recursive term is $O\left((1/\eps)^{\frac{(d-1)^2}{d}+1+\eta}\right)=O\left((1/\eps)^{d-1+1/d+\eta}\right)\ll (1/\eps)^{\alpha_d+\eta}$ for all $d\geq 5$. 

Lastly, the first non-recursive term satisfies
$$
t^{d/\alpha_d+\eta/(10d)}\ll (1/\eps)^{\frac{d^2\cdot\left(d+2-\sqrt{d^2-2d}\right)}{2d+2\sqrt{d^2-2d}}+\eta},
$$

\noindent where the exponent in the right hand side is smaller than $\alpha_d+\eta=\frac{1}{2}(d+\sqrt{d^2-2d})+\eta$ for all $d\geq 5$. 
 $\Box$

\section{Piercing the $\delta$-hollow sets -- proof of Theorem \ref{Theorem:BoundHollow}}\label{Sec:VerticallyConvex}

\medskip
\noindent{\bf Setup.} Throughout this section, we will stick with the quantities $P$, $\Pi$, $\eps$, $\sigma$, $n=|P|$, and the auxiliary parameters that were set in the beginning of Section \ref{Sec:MainRecurrence}. In particular, we will use positive integers $h,r,s$, and $u$, that are very small, albeit fixed and positive, degrees of $1/\eps$, which satisfy
$$
h\lll_\eta r\lll_\eta s\lll_\eta u\lll_\eta 1/\eps,
$$

\noindent the fraction $\delta=1/r^{d(d+1)+1}$, and the ``spread threshold" $\varepsilon=\Omega(\eps^{1+\eta})$ whose value  (\ref{Eq:Spread}) has been selected in the beginning of Section \ref{Sec:MainRecurrence}.

In the rest of this section, we denote $m:=|\Pi|$, and use $\HH$ to denote the family $\HH(\Pi)=\{H(\tau)\mid \tau\in \Pi\}$ of the $m$ hyperplanes that support the simplices of $\Pi$. 
Recall that we have that $\eps{n\choose d}\leq |\Pi|\leq \rho{n\choose d}$, and that our ``edge-constrained" family $\K=\K(P,\Pi,\eps,\sigma)$ consists of all such convex sets $K\in \K(P,\eps)$ that are
 $(\eps,\sigma)$-restricted to the hypergraph $(P,\Pi)$ and, thereby, satisfy $|\Pi_K|\geq \sigma{|P_K|\choose d}$. (As before, $\Pi_K$ is used to denote the subset ${P_K\choose d}\cap \Pi$ of $\Pi$ that is induced by the points of $P_K$.) Also recall that the fraction $0<\sigma\leq 1$ is bounded from below by some constant $\sigma_0>0$ that was described in Section \ref{Sec:MainRecurrence}.
  
To establish Theorem \ref{Theorem:BoundHollow} in Section \ref{Sec:MainRecurrence} and, therefore, complete the proof of Theorem \ref{Thm:Main}, we will describe a small-size net $N_{\leq \delta}$ for the sub-family $\K_{\leq \delta}$, which is comprised of all such $\delta$-hollow convex sets in $\K$ that are $\varepsilon$-spread within the vertical partition $\V(P,s)$ of $P$, which too has been introduced in Section \ref{Sec:MainRecurrence}. 

 Similar to Section \ref{Sec:Surrounded}, our divide-and-conquer treatment of the family $\K_{\leq \delta}$ will use a decomposition of the ground point set $P$, along with the ambient space $\reals^d$. However, for reasons that will become clear towards the end of this section, this decomposition will be based on the random arrangement $\A(\R)$ of an $r$-sample $\R\subseteq \HH=\HH(\Pi)$, and a subsequent refinement of $\A(\R)$ to an $O\left(\frac{\log r}{r}\right)$-cutting $\D$ \cite{Cuttings1,Cuttings} of the family $\HH=\HH(\Pi)$. 

\medskip
\noindent{\bf The random arrangement $\A(\R)$.} 
Combining the assumption $m\geq\eps{n\choose d}$ with the lower bound (\ref{Eq:ManyPoints}) on $n$, we obtain that $|\HH|=m\geq r$.
Consider a random $r$-sample $\R$ of $\HH$. Note that the hyperplanes of $\R$ are not necessarily in general position, because those of $\HH$ are, most likely, {\it not} -- any point of $P$ can lie in up to ${n\choose d-1}$ hyperplanes of $\HH$. Nevetheless, any hyperplane in $\HH$ contains exactly $d$ points of $P$. 
 As was mentioned in Section \ref{Subsec:Arrangements}, the hyperplane collection $\R$ determines the {\it arrangement} $\A(\R)$ whose $d$-dimensional faces constitute a decomposition of $\reals^d\setminus (\bigcup R)$ into $O\left(r^d\right)$ open polyhedral {\it cells}, of overall complexity $O(r^d)$ \cite[Section 6]{JirkaBook}. See Figure \ref{Fig:RandomArrangement}.

\begin{figure}
    \begin{center}
        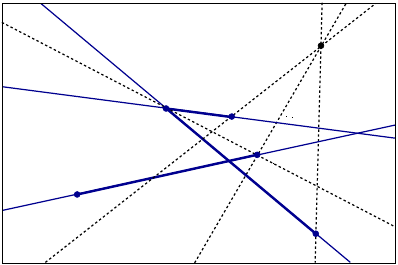  \caption{\small The random arrangement $\A(\R)$ illustrated in $\reals^2$. The family $\HH=\HH(\Pi)$ is comprised of the hyperplanes $H(\tau)$ that support the simplices of $\tau\in \Pi$. The hyperplanes of the $r$-sample $\R\subset \HH(\Pi)$ are blue and solid, whereas the rest of the hyperplanes of $\HH(\Pi)$ are black and dashed. }
      \label{Fig:RandomArrangement}
    \end{center}
\end{figure}

As a result, the ground point set $P$ is subdivided into two parts $P'$ and $P''$. The larger part $P'=P\setminus \left(\bigcup \R\right)$ is comprised of all the points in $P$ that lie in the interiors of the cells of $\A(\R)$, whereas the complementary subset  $P''$ consists of at most $rd$ points that lie on one or more of the sample hyperplanes of $\R$. By again using the lower bound (\ref{Eq:ManyPoints}) on $n$, it follows that the principal subset $P_K$ of every convex set $K\in \K_{\leq \delta}$, must encompass at least $\lceil\eps n\rceil-rd\geq \lceil \eps n\rceil/2$ points of $P'$. 

\medskip
\noindent{\bf The cutting $\D=\D(\R,P)$.} Applying the bottom-vertex triangulation of Section \ref{Subsec:Arrangements} to the cells of $\A(\R)$ yields a simplicial decomposition $\D(\R)$, that is, a collection of $O\left(r^d\right)$ pairwise disjoint open $d$-dimensional  simplices, so that every point of $P'$ lies in the interior of a unique simplicial cell $\varphi$ of $\D(\R)$. See Figure \ref{Fig:CuttingD} (left). As has been noted in Section \ref{Subsec:Arrangements}, we have that, with probability at least $1/2$, the interior of such cell $\varphi$ of $\D(\R)$ is crossed by $O\left(\frac{m}{r}\log r\right)$ hyperplanes of $\HH$. Hence, in the broadly accepted terminology  \cite{Cuttings1,ChazelleBook,Cuttings}, the decomposition $\D(\R)$ already yields a $O(\log r/r)$-cutting of $\HH$.


\begin{figure}
    \begin{center}
        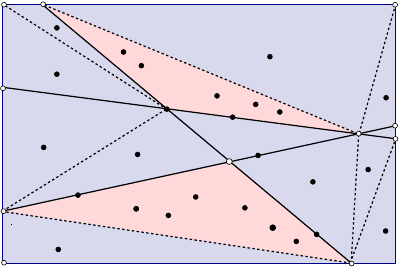\hspace{1cm}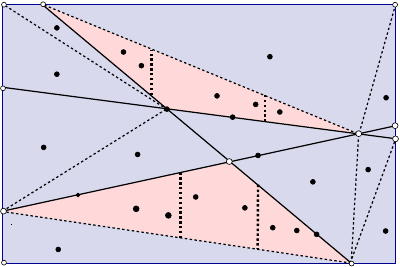  
        \caption{\small The cutting $\D(\R,P)$ illustrated in $\reals^2$. We begin with a bottom-vertex triangulation $\D(\R)$ of the cells of $\A(\R)$ (left). 
        The ``overfull" cells $\varphi\in \D(\R)$, with at least $\lceil n/k\rceil$ points in their interiors, are called {\it red}; each of them is split, by the means of vertical walls, into several red cells $\varphi_i$ within $\D=\D(\R,P)$ (right).}
      \label{Fig:CuttingD}
    \end{center}
\end{figure}

 To further refine the described cutting $\D(\R)$, let us fix another auxiliary parameter $k:=r^{(d+\sqrt{d^2-2d})/2}\leq r^{\alpha_d}$, which satisfies $r^{d-1}\leq k\leq r^{d-1/2}$. 

\medskip
\noindent{\bf Definition.} Let us denote $P_\varphi:=P\cap \varphi$ for each cell $\varphi\in \D(\R)$.
We say that a simplicial cell $\varphi\in \D(\R)$ is {\it red} if $|P_\varphi|\geq \lceil n/k\rceil$, and we say that it is {\it blue} otherwise. 
By definition, there exist at most $k$ red cells in $\D(\R)$.

\medskip
We further subdivide every red cell $\varphi$ of $\D(\R)$ into $\lceil|P_\varphi| /\lfloor\frac{n}{k}\rfloor\rceil$ open polyhedral cells $\varphi'$, which satisfy $|\varphi'\cap P|\leq\lfloor n/k\rfloor$. To this end, we introduce at most $\lceil|P_\varphi| /\lfloor\frac{n}{k}\rfloor\rceil-1\leq 2|P_\varphi| \frac{k}{n}-1$ walls which are orthogonal to the first coordinate axis. (None of these walls can pass through a point of $P$.) Let $\D(\R,P)$ denote this more refined decomposition of $\reals^d$, as illustrated in Figure \ref{Fig:CuttingD} (right).

\medskip
\noindent{\bf Definition.} For every red simplex of $\D(\R)$ we refer to {\it all} its off-shoots in $\D(\R,P)$ as {\it red}. If a cell of $\D(\R,P)$ is not red, we say that it is {\it blue}, in which case it has been inherited en toto from the bottom-vertex triangulation $\D(\R)$.

\medskip
Note that the refined decomposition $\D:=\D(\R,P)$ still encompasses $O\left(r^d\right)$ cells, including at most $2k$ red cells; each of these cells $\varphi$ is a convex polytope with at most $d+2$ facets, at most ${d+1\choose 2}+2{d\choose 2}=O(d^2)$ edges, and $O(d^2)$ vertices. 
As $\D$ is a refinement of the $O\left(\frac{1}{r}\log r\right)$-cutting $\D(\R,P)$, any open cell of $\D$ is still crossed by $O\left(\frac{m}{r}\log r\right)$ hyperplanes of $\HH$, which yields the following immediate corollary.

\begin{observation}\label{Prop:Cutting} 
There is a constant $c_{\it cut}>0$, which depends only on the dimension $d$, so that
 the interior of every cell $\varphi\in \D(\R,P)$ is crossed by at most $c_{\it cut}\left(\frac{m}{r}\log r\right)$ hyperplanes of $\HH$.
\end{observation}

\medskip
The net $N_{\leq \delta}$ will combine the following key ingredients:
	(i) $\tilde{\eps}$-nets which are defined over smaller ground sets $\tilde{P}$, each of them cut out by either a cell of $\A(\R)$, or a red cell $\varphi$ of $\D$,    
(ii) $1$-dimensional $\Omega(\eps^{d-1+\eta})$-nets which will be constructed within the lines of the canonical family $\L=\L(P,s)$, that was defined in Section \ref{Sec:MainRecurrence} over the vertical partition $\V(P,s)$.

Similar to its predecessor $N_{>\delta}$ in Section \ref{Sec:Surrounded}, the net $N_{\leq \delta}$ will be defined in three incremental steps. Initially, the set $N_{\leq 
\delta}$ is empty. In each subsequent step $1\leq i\leq 3$ we add to $N_{\leq \delta}$ a relatively simple partial net $N_{\leq \delta,i}$ that pierces one or several specific categories of the convex sets, to be immediately removed from $\K_{\leq \delta}$. By the end of Step 3, no set in $\K_{\leq \delta}$ will be left unpierced by $N_{\leq \delta}=\bigcup_{i=1}^3N_{\leq \delta,i}$.

The first two nets $N_{\leq \delta,1}$ and $N_{\leq \delta,2}$ will overly resemble their counterparts in Section \ref{Sec:Surrounded}, and make no use of the $\delta$-hollowness of the sets $K\in \K_{\leq \delta}$. 

By the end of Step 2, it will be shown that every remaining convex set $K\in \K_{\leq \delta}$, that have been ``missed" by the combination $N_{\leq \delta,1}\cup N_{\leq \delta,2}$, is sufficiently ``thin" with respect to both the random arrangement $\A(\R)$, and its finer cutting $\D=\D(\R,P)$. In particular, a large fraction of its principal subset $P_K$ will lie in the zone of a certain {\it principal hyperplane $H_K\in \HH$}, which crosses only $O^*\left(\frac{k}{r}\right)$ among the red cells of $\D(\R,P)$.

In Step 3, the $\delta$-hollowness will be used, at last, to dispose of the remaining convex sets $K\in \K_{\leq \delta}$ by the means of the basic recurrence in $\eps>0$, that was sketched in the beginning of Section \ref{Subsec:RecursiveFramewk}.

\subsection{Step 1}\label{Subsec:1stNetHollow}

\noindent{\bf Definition.} Let $\Pi''$ denote the subset of all the simplices $\tau\in \Pi$ that meet at least one of the following criteria:
\begin{enumerate}
	\item[(i)]  one or more of the vertices of $\tau$ belong to $P''$ and, therefore, lie on a hyperplane of $\R$.
  \item[(ii)] the supporting hyperplane $H(\tau)\in \HH$ crosses at least $4c_{\it cut}\cdot \frac{hk}{r}\log r$ red cells of $\D=\D(\R,P)$.
 \end{enumerate}

We then denote $\Pi':=\Pi\setminus \Pi''$.

\medskip
\noindent {\bf The net $N_{\leq \delta,1}$.} We set

\begin{equation}\label{Eq:Sparser}
	N_{\leq \delta,1}:=N(P,\Pi'',\eps,\sigma/2),
\end{equation}

\noindent where the right hand side denotes the net of cardinality at most $f(\eps,|\Pi''|/{n\choose d},\sigma/2)$, that exists for the family
$\K(P,\Pi'',\eps,\sigma/2)$, which is
comprised of all such convex sets $K\in \K(P,\eps)$ that are $(\eps,\sigma/2)$-restricted to the hypergraph $(P,\Pi'')$ (and thus satisfy $|{P_K\choose d}\cap \Pi''|\geq (\sigma/2){|P_K|\choose d}$).

\medskip
\noindent{\bf The analysis.} The following claim shows that the recursive instance (\ref{Eq:Sparser}) involves a considerably sparser subset $\Pi''$.

\begin{lemma}
We have that $|\Pi''|\leq m/h\leq \rho{n\choose d}/h$. 
\end{lemma}
\begin{proof}

Using the lower bound (\ref{Eq:ManyPoints}) on $n=|P|$ in the beginning of Section \ref{Sec:MainRecurrence}, that $h\lll_\eta r\lll_\eta 1/\eps$, and that $m=|\Pi|\geq \eps {n\choose d}$, it follows that
at most $rd {n \choose d-1}\leq r {n\choose d}/n\leq rm/(\eps n)\leq  m/(2h)\leq \rho{n\choose d}/(2h)$ simplices of $\Pi''$ can fall in the first category. 

Furthermore, in view of Observation \ref{Prop:Cutting}, there exist at most
$2c_{\it cut}\cdot \frac{mk}{r}\log r$ ``incidences" between the (at most $2k$) red cells of $\D$, and the $m$ hyperplanes of $\HH$ that intersect their interiors. Thus, by the pigeonole principle, at most 
$m/(2h)\leq \rho{n\choose d}/(2h)$ simplices $\tau\in \Pi$ can fall in the second category.
\end{proof}

The properties of $N_{\leq \delta,1}$ are summarized in the following lemma.

\begin{lemma}\label{Lemma:1stNetHollow}
\begin{enumerate}
	\item The net $N_{\leq \delta,1}$ has cardinality at most $f(\eps,\rho/h,\sigma/2)$.
	\item Upon upon including the points of $N_{\leq \delta,1}$ in $N_{\leq \delta}$ and, accordingly removing from $\K_{\leq \delta}$ all the sets that
are pierced by $N_{\leq \delta,1}$, every remaining set $K\in \K_{\leq \delta}$ satisfies $|{P_K\choose d}\cap \Pi'|\geq (\sigma/2){\lceil \eps n\rceil\choose d}$. 

Namely, at least $\left\lceil (\sigma/2) {\lceil\eps n\rceil\choose d}\right\rceil$ of the simplices in the family $\Pi_K={P_K\choose d}\cap \Pi$ belong to ${P'\choose d}$, and their
supporting hyperplanes $H(\tau)$ in $\HH$ cross the interiors of fewer than $4c_{\it cut}\cdot \frac{hk}{r}\log r$ red cells within $\D$. 
\end{enumerate}
 \end{lemma}

\medskip

In the sequel, we denote $P'_K:=P_K\cap P'$, and $\Pi'_K:=\Pi'\cap {P_K\choose d}$, for every $K\in \K_{\leq \delta}$. 

\subsection{Step 2}\label{Subsec:2ndNetHollow}

Our second net $N_{\leq \delta,2}$ for the family $\K_{\leq \delta}$ will be comprised of $1$-dimensional $\Omega\left(\eps^{d-1+\eta}\right)$-nets $N_\ell$, which will be restricted to the $O\left(s^{(d-1)^2}\right)$ vertical lines $\ell$ of the $\vartheta$-canonical family $\L(P,s)$ over the vertical partition $\V(P,s)$.  

Upon adding the points of $N_{\leq \delta,2}$ to $N_{\leq \delta}$, and thereby removing from $\K_{\leq \delta}$, every convex set that is pierced by them, every remaining convex set $K\in \K_{\leq \delta}$ will be assigned a so called {\it principal hyperplane $H_K\in \HH\left(\Pi'_K\right)$}, whose zone in $\D$ will contain $\Omega(\eps n)$ points of $P_K$. 



\medskip

\noindent{\bf The net $N_{\leq \delta,2}$.} The net $N_{\leq \delta,2}$ is attained through several explicit constructions using the $\vartheta$-canonical line set $\L(P,s)$, the sample $\R\subseteq \Pi$ of $r$ hyperplanes, and the cutting $\D=\D(\R,P)$, which has been defined in the beginning of this section.

\begin{enumerate}
	\item For each vertical line $\ell\in \L(P,s)$, we consider the set
$$
X_\ell:=\{\ell\cap H\mid H\in \R\}
$$ 

\noindent of all the $\ell$-intercepts of the hyperplanes in $\R$, and let

$$
X:=\bigcup_{\ell\in \L(P,s)} X_\ell.
$$

We include in the latter set $X$ every vertex of $\D$.
Since $r\lll_\eta s$, we have that 
$$
|X|\leq s^{(d-1)^2}r+O(r^d)=O\left(s^{(d-1)^2}r\right).
$$

\item We then use the points of $X$ to define a smaller auxiliary family $\Pi_X$ of ``synthetic" $(d-1)$-simplices

$$
\Pi_X:={P\choose d-1}\ast X=\{\conv(\kappa\cup \{x\})\mid \kappa\in {P\choose d-1},x\in X\}
$$ 

\noindent which determine the family $\HH_X:=\HH(\Pi_X)$ of supporting hyperplanes. Notice that

$$
|\HH_X|=|\Pi_X|=O\left(n^{d-1}|X|\right)=O\left(n^{d-1}s^{(d-1)^2}r\right).
$$

\item For each canonical line $\ell\in \L(P,s)$ we construct the point set 

$$
Y_\ell=:\{\ell\cap H\mid H\in \HH_X\}
$$ 
\noindent which encompasses all the $\ell$-intercepts of the hyperplanes of $\HH_X$, and note that $|Y_\ell|=O\left(n^{d-1}s^{(d-1)^2}r\right)$.

\item Lastly, for each line $\ell\in \L(P,s)$ we construct the net $N_\ell$ which includes every $\left\lceil (\vartheta/4)(\sigma_0/2)^{\beta_{d-1}}{\lceil \eps n\rceil \choose d-1}\right\rceil$-th point of $Y_\ell$, and let

$$
N_{\leq \delta,2}:=\bigcup_{\ell\in \L(P,s)}N_\ell.
$$

\end{enumerate}

We add $N_{\leq\delta,2}$ to $N_{\leq \delta}$, and remove from $K_{\leq \delta}$ every convex set already pierced by $N_{\leq \delta,2}$.

\medskip
\noindent{\bf The analysis.} Using that $r\lll_\eta s\lll_\eta 1/\eps$ (and that $\sigma_0$ is a constant that depends only on $\eta$ and the dimension $d$), we obtain the following immediate bound on the cardinality of $N_{\leq \delta,2}$.
\begin{lemma}\label{Lemma:2ndNetHollow}
	The previously defined net $N_{\leq \delta,2}$ has cardinality 
$$
O\left(\frac{s^{2(d-1)^2}r}{\sigma_0^{\beta_{d-1}}\eps^{d-1}}\right)=O\left(\frac{1}{\eps^{d-1+\eta}}\right).
$$
\end{lemma}

\noindent {\bf Definition.} For any point $p\in \reals^d$, and any non-vertical hyperplane $H$ in $\reals^d$, we use $[p]_H$ to denote the vertical projection of $p$ onto $H$, and we use $e(p,H)$ to denote the closed vertical segment between $p$ and $[p]_H$.

We then say that $p$ {\it sees} the hyperplane $H$ in $\A(\R)$ if $e(p,H)$ does not cross any hyperplane in $\R$. (In other words, both $p$ and $p_H$ have to lie in the same open cell of $\A(\R)$. See Figure \ref{Fig:WitnessZone}.)

\begin{lemma}\label{Lemma:PrincipalHyperplane}
    Let 
    $$
	c'_2:=\frac{d\vartheta}{2^{\beta_{d-1}+5}}.
    $$
	Then every convex set $K\in \K_{\leq \delta}$ that is missed by $N_{\leq \delta,1}\cup N_{\leq \delta,2}$, can be assigned a hyperplane $H_K\in \HH\left(\Pi'_K\right)$, and a
	subset $A_K\subseteq P_K$, with the following properties: 
	
	\begin{enumerate}[label=(\roman*)]
	\item We have that $|A_K|\geq c'_2 \sigma_0^{\beta_{d-1}} \lceil\eps n \rceil$.

	\item $A_K$ is contained in the zone of $H_K$ in $\D(\R,P)$.
	\item Every point $p\in A_K$ sees $H_K$; see Figure \ref{Fig:WitnessZone}.
\end{enumerate}
	\end{lemma}

 \medskip
\noindent {\bf  Definition.} 
We assign to each remaining convex set $K\in \K_{\leq \delta}$ a {\it unique} hyperplane $H_K\in \HH\left(\Pi'_K\right)$ with a {\it unique} subset $A_K\subseteq P_K$, that together meet the criteria (i) -- (iii) of Lemma \ref{Lemma:PrincipalHyperplane},
and refer to $H_K$ as the {\it principal hyperplane} of $K$.

\begin{figure}
    \begin{center}      
        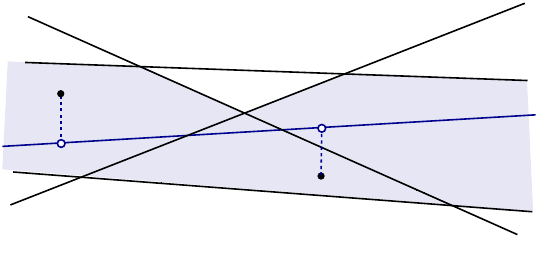
        \caption{\small The point $p\in P_K$ {\it sees} the hyperplane $H\in \HH$ in $\A(\R)$: its $H$-projection $[p]_H$ lies in the same cell of $\A(\R)$. In contrast, the point $p'$ {\it does not see} $H$ in $\A(\R)$, as the vertical segment $e(p',H)$ between $p'$ and $[p']_H$ is crossed by a hyperplane $H'\in \R$. Notice that both $p$ and $p'$ lie in the shaded zone of $H$ in $\A(\R)$.}
        \label{Fig:WitnessZone}
    \end{center}
\end{figure}

\begin{proof}[Proof of Lemma \ref{Lemma:PrincipalHyperplane}.]	
We fix a convex set $K\in \K_{\leq \delta}$, and suppose, for a contradiction, that there does not exist any such combination of a hyperplane $H_K\in \HH\left(\Pi'_K\right)$, with a subset $A_K\subseteq P_K$, that together meet the criteria (i) -- (iii). 
We show that $K$ is already pierced by the net $N_{\leq \delta,2}$.

\smallskip
Indeed, since (a) $K$ is $\varepsilon$-spread in $\V(P,s)$, with the the threshold $\varepsilon>0$ in (\ref{Eq:Spread}) which is much smaller than $\vartheta(\sigma_0/2)^{\beta_{d-1}}\eps$, and (b) we have that $|\Pi'_K|\geq \frac{\sigma}{2} {\lceil \eps n\rceil\choose d}\geq \frac{\sigma_0}{2} {\lceil \eps n\rceil\choose d}$ (by Lemma \ref{Lemma:1stNetHollow}), the $\vartheta$-canonical family $\L(P,s)$ must include a line $\ell_0\in \L(P,s)$ that pierces a subset $\Pi'_K(\ell_0)$ of $\vartheta\left(\frac{\sigma_0}{2}\right)^{\beta_{d-1}}{\lceil\eps n\rceil\choose d}$ 
simplices within $\Pi'_K$. 

\medskip
To proceed with the proof of Lemma \ref{Lemma:PrincipalHyperplane}, let us introduce further notation.

\medskip
\noindent{\bf Definition.} 
Let $\kappa$ be a $(d-2)$-simplex in ${P'_K\choose d-1}$. 

\begin{enumerate}
	\item We say that a simplex $\tau\in \Pi'_K(\ell_0)$ is {\it adjacent} to $\kappa$ if $\kappa$ contains $d-1$ vertices of $\tau$.
We then say that a point $p\in P'_K$ is a {\it neighbor} of $\kappa$ if the $(d-1)$-dimensional simplex $\conv(\kappa\cup\{p\})$ belongs to $\Pi'_K(\ell_0)$. See Figure \ref{Fig:GoodSimplex}.

\item We use $\Gamma'_K(\kappa)\subseteq P_K$ to denote the subset of all the neighbours of $\kappa$. (Notice that the points of $\Gamma'_K(\kappa)$ correspond to the simplices of $\Pi'_K(\ell_0)$ that are adjacent to $\kappa$.)

 \item We say that $\kappa$ is {\it good} if we have that
 
\begin{equation}\label{Eq:Good}
|\Gamma'_K(\kappa)|\geq \frac{\vartheta}{4d}\cdot\left(\frac{\sigma_0}{2}\right)^{\beta_{d-1}}\left(\lceil\eps n\rceil-d+1\right)\geq \frac{\vartheta}{8d}\cdot\left(\frac{\sigma_0}{2}\right)^{\beta_{d-1}}\lceil\eps n\rceil,\end{equation}

\noindent where the last inequality follows from the lower bound (\ref{Eq:ManyPoints}) on $n$.

\end{enumerate}

\begin{figure}
    \begin{center}      
        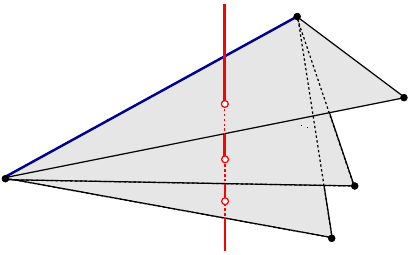
        \caption{\small A $(d-2)$-simplex $\kappa$ along with the adjacent simplices $\tau\in \Pi'_K(\ell_0)$, which are crossed by the canonical line $\ell_0$. Each $\tau\in \Pi'_K(\ell_0)$ contributes a vertex $p$ to the neighbor set $\Gamma'_K(\kappa)$.}
        \label{Fig:GoodSimplex}
    \end{center}
\end{figure}

Since $|\Pi'_K(\ell_0)|\geq \vartheta\left(\frac{\sigma_0}{2}\right)^{\beta_{d-1}}{\lceil\eps n\rceil\choose d}$, the pigeonhole principle yields the following property.

\begin{proposition} \label{Claim:GoodSimplices}
There exist at least $\frac{\vartheta}{4}\left(\frac{\sigma_0}{2}\right)^{\beta_{d-1}}{\lceil \eps n \rceil\choose d-1}$ good $(d-2)$-simplices in ${P'_K\choose d-1}$.
\end{proposition}


\medskip
\noindent{\bf Definition.}  We say that a good $(d-2)$-simplex $\kappa\in {P'_K\choose d-1}$ is {\it very good} if for any two simplices $\tau_1,\tau_2\in \Pi'_K(\ell_0)$ that are adjacent to $\kappa$, their supporting hyperplanes $H_1:=H(\tau_1)$ and $H_2:=H(\tau_2)$ determine {\it identical} zones in $\D$ (i.e., $H_1$ and $H_2$ cross the {\it same} cells of $\D$).

\medskip
To establish the lemma, it is enough to show that every good $(d-2)$-simplex $\kappa\in {P_K\choose d-1}$ gives rise to a ``synthetic" simplex $\pi_\kappa\in \Pi_X$, of the general form $\conv(\kappa\cup \{x\})$; the supporting hyperplane $H\left(\pi_\kappa\right)$ then meets $\ell_0$ at a point of $Y_{\ell_0}\cap K$. Repeating this argument for each good $(d-2)$-simplex $\kappa\in {P'_K\choose d-1}$ would yield, according to Proposition \ref{Claim:GoodSimplices}, a total of at least $\frac{\vartheta}{8}\left(\frac{\sigma_0}{2}\right)^{\beta_{d-1}}{\lceil \eps n \rceil\choose d-1}$ such points of $Y_{\ell_0}\cap K$ (which are easily seen to be distinct due to the general position of $P$, and the generic choice of the lines of $\L(P,s)$).
Therefore, the convex set $K$ under consideration must have been pierced by a point of the 1-dimensional net $N_{\ell_0}\subseteq N_{\leq \delta,2}$, that was constructed over $Y_{\ell_0}$.

\medskip
We keep the good simplex $\kappa\in {P'_K\choose d-1}$ fixed, and distinguish between two possible scenarios.

\medskip
\noindent{\bf Case (i).} The good $(d-2)$-simplex $\kappa\in {P_K\choose d-1}$ is not very good. Thus, there must exist a pair of simplices $\tau_1,\tau_2\in \Pi'_K(\ell_0)$ that are adjacent to $\kappa=\tau_1\cap \tau_2$, and a cutting cell $\varphi\in \D$ that is intersected by $H(\tau_1)$ but not $H(\tau_2)$. 

Let us continuously rotate a hyperplane $H$ around the $(d-2)$-flat $\aff(\kappa)$, from $H(\tau_1)$ to $H(\tau_2)$, so that $H$ never ceases intersecting the vertical line $\ell_0$. As is easy to check, there is exactly one direction that permits such a continuous rotation around $\kappa$.
By continuity, there must be a point in time when the intersection $H\cap \varphi$ is about to vanish from $H$. By the general position, this can only happen when $H$ coincides with a vertex $v$ of $\varphi$ and, therefore, supports the ``synthetic" simplex $\pi_\kappa:=\conv(\kappa\cup \{v\})$ (so that $H=H\left(\pi_\kappa\right)$). Since $v$ has been included in $X$, this simplex $\pi_\kappa$ belongs to $\Pi_X$. Hence, the set $Y_{\ell_0}$ must  include the point $H\cap {\ell_0}$, which clearly lies in the vertical interval between the points $\tau_1\cap \ell_0$ and $\tau_2\cap \ell_0$, and well within $K\cap \ell_0$. See Figure \ref{Fig:NotVeryGood}.

\begin{figure}
    \begin{center}      
        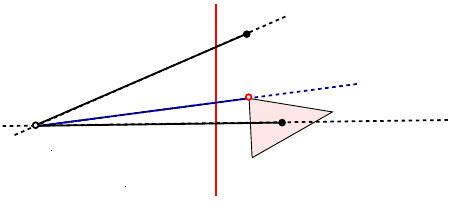
        \caption{\small Proof of Lemma \ref{Lemma:PrincipalHyperplane} -- Case (i). The good $(d-2)$-simplex $\kappa$ is not very good (view in $\reals^3$, in the direction parallel to $\kappa$). There exist $(d-1)$-simplices $\tau_1,\tau_2\in \Pi'_K(\ell_0)$ so that $H(\tau_1)$ crosses a cell $\varphi\in \D$ that is missed by $H(\tau_2)$. There is a vertex $v$ of $\varphi$, and a simplex $\pi_\kappa=\conv(\kappa\cup \{v\})\in \Pi_X$,  such that the hyperplane $H(\pi_\kappa)\in \HH_X$ meets $\ell_0$ at a point of $Y_{\ell_0}$ between  $\tau_1\cap\ell_0$ and $\tau_2\cap \ell_0$.}
        \label{Fig:NotVeryGood}
    \end{center}
\end{figure}


\medskip
\noindent {\bf Case (ii).} If the $(d-2)$-simplex $\kappa\in {P'_K\choose d-1}$ is very good,
then all of its $|\Gamma'_K(\kappa)|\geq (\vartheta/8d)(\sigma_0/2)^{\beta_{d-1}}\lceil\eps n\rceil$ neighbors $p\in \Gamma'_K(\kappa)$ must belong, in $\D$, to the zone of {\it every} hyperplane $H(\tau)$ that supports some $\kappa$-adjacent simplex $\tau\in \Pi'_K(\ell_0)$ (with $\kappa\subset \tau$). 

Note that every $\kappa$-adjacent simplex $\tau\in \Pi'_K(\ell_0)$ yields a (tentative) assignment $H_K:=H(\tau)$ and $A_K:=\Gamma'_K(\kappa)$ which meets the first two criteria of Lemma \ref{Lemma:PrincipalHyperplane}, yet fails to meet the last one.

\medskip
\noindent{\bf Definition.} 
For any simplex $\tau\in \Pi'_K(\ell)$ that is adjacent to $\kappa$, let $\Gamma'_K(\kappa,\tau)$ denote the subset of all such points in $\Gamma'_K(\kappa)$ that see $H(\tau)$.

\medskip
It can be further assumed that every adjacent simplex $\tau\in \Pi'_K(\ell_0)$ of $\kappa$, satisfies

\begin{equation}\label{Eq:NotSee}
	|\Gamma'_K(\kappa,\tau)|<\frac{\vartheta}{16d}\left(\frac{\sigma_0}{2}\right)^{\beta_{d-1}} \lceil\eps n \rceil,
\end{equation}

\medskip
\noindent or, else, all the three conditions (i) -- (iii) would have been met by the combination $H_K:=H(\tau)$ and $A_K:=\Gamma'_K(\kappa,\tau)$. 

\medskip
Specializing to the $|\Gamma'_K(\kappa)|$ $\kappa$-adjacent simplices $\tau\in \Pi'_K(\ell_0)$, let us fix such a simplex $\tau_\kappa=\tau$ whose supporting hyperplane $H(\tau)$ attains the {\it lowest} possible $\ell_0$-intercept $H(\tau)\cap \ell_0$, and denote $H_\kappa:=H(\tau_\kappa)$. 
Clearly,  the point $H_\kappa\cap \ell_0$ belongs to the interval $K\cap \ell_0$, for $\ell_0$ crosses the supporting simplex $\tau=\tau_\kappa$ of $H_\kappa$.

Combining the estimates (\ref{Eq:Good}) and (\ref{Eq:NotSee}) implies that the set $\Gamma'_K(\kappa)\setminus \Gamma'\left(\kappa,\tau_\kappa\right)$ is comprised of
at least 

$$
|\Gamma'_K(\kappa)|-|\Gamma'_K(\kappa,\tau_\kappa)|\geq \frac{\vartheta}{16d}\left(\frac{\sigma_0}{2}\right)^{\beta_{d-1}} \lceil\eps n \rceil
$$ 

\noindent such neighbouring points $p$ of $\kappa$ that {\it do not} see $H_\kappa$.
Furthermore, it can be assumed, with no loss of generality, that at least half of these points in $\Gamma'_K(\kappa)\setminus \Gamma'\left(\kappa,\tau_\kappa\right)$ lie {\it above} the hyperplane $H_\kappa$. Let $\Gamma_K^+(\kappa)$ denote this latter subset within $\Gamma'_K(\kappa)\setminus \Gamma'\left(\kappa,\tau_\kappa\right)$, whose cardinality satisfies

$$
|\Gamma_K^+(\kappa)|>\frac{\vartheta}{32d}\left(\frac{\sigma_0}{2}\right)^{\beta_{d-1}} \lceil\eps n \rceil.
$$

  \begin{figure}
    \begin{center}      
        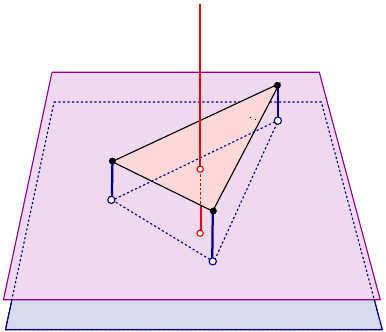
        \caption{\small Proof of Lemma \ref{Lemma:PrincipalHyperplane} in $\reals^3$. The canonical line $\ell_\kappa\in \L(P,s)$ crosses at least one simplex $\phi$, that is determined by $P_K^+(\kappa)$, at some point $z$. The vertices $p_i$ of $\phi$ are ``blocked" from seeing $H_\kappa$ by the hyperplane $G_\kappa\in \R$. The depicted $\ell_\kappa$-intercept $x_\kappa\in X_\kappa$ of $G_\kappa$ too lies above $H_\kappa$.} 
        \label{Fig:NotExcellent}
    \end{center}
\end{figure}

Notice that, for each $p\in \Gamma_K^+(\kappa)$, the vertical segment $e(p,H)$, between $p$ and $[p]_{H_\kappa}$, must be crossed by one or more hyperplanes $G\in \R$, so that, by the pigeonhole principle, one can fix such a hyperplane $G=G_\kappa\in \R$ that crosses at least 
 $$
 |\Gamma^+_K(\kappa)|/r\geq \frac{\vartheta}{32rd}\left(\frac{\sigma_0}{2}\right)^{\beta_{d-1}} \lceil\eps n \rceil
 $$ 
 \noindent segments in $\{e(p,H_\kappa)\mid p\in \Gamma^+_K(\kappa)\}$. 

Let $P^+_K(\kappa)$ denote the resulting subset 
$\{p\in \Gamma^+_K(\kappa)\mid e(p,H_\kappa)\cap G_\kappa\neq \emptyset\}$. Note that the points of $P^+_K(\kappa)$ lie above both hyperplanes $G_\kappa$ and $H_\kappa$, and that

$$
|P_K^+(\kappa)|>\frac{\vartheta}{32rd}\left(\frac{\sigma_0}{2}\right)^{\beta_{d-1}} \lceil\eps n \rceil \geq \varepsilon\cdot n/\vartheta\geq d,
$$ 

\noindent where the second and the third inequalities follow from, respectively, the choice (\ref{Eq:Spread}) of $\varepsilon>0$, and the lower bound (\ref{Eq:ManyPoints}) on $n$ in Section \ref{Sec:MainRecurrence}. Since the convex set $K$ is $\varepsilon$-spread in the vertical partition $\V(P,s)$, and the line family $\L(P,s)$ is $\vartheta$-canonical with respect to $\V(P,s)$, there must be such a line $\ell_\kappa\in \L(P,s)$ that crosses at least 
{\it one} simplex $\phi$ whose $d$ vertices belong to $P_K^+(\kappa)$. See Figure \ref{Fig:NotExcellent}.


Since $G_\kappa\in \R$, its $\ell_\kappa$-intercept $x_\kappa=\ell_\kappa\cap G_\kappa$ must belong to the previously defined subset $X_{\ell_\kappa}
\subset X$, so that the simplex $\pi_\kappa:=\conv(\kappa\cup \{x_\kappa\})$ must belong to $\Pi_X$. Thus, the intersection point $y_\kappa=H(\pi_\kappa)\cap \ell_0$, between $\ell_0$ and the supporting plane $H(\pi_\kappa)$, must have been included in the set $Y_{\ell_0}$.

\begin{proposition}\label{Claim:YL}
Let $\kappa\in {P_K\choose d-1}$ be a very good $(d-2)$-simplex.
Then the previously described point $y_\kappa=H(\pi_\kappa)\cap \ell_0$ in $Y_{\ell_0}$ must lie in the interval $K\cap \ell_0$. 
\end{proposition}

\begin{proof}[Proof of Proposition \ref{Claim:YL}.] 
In view of the choice of $\ell_\kappa$, one can fix a simplex $\phi=\conv(p_1,\ldots,p_d)$, that is crossed by $\ell_\kappa$ at some point $z$, and whose $d$ vertices $p_1,\ldots,p_d$ belong to the above set $P_K^+(\kappa)$. 
Let us emphasize that, as $P^+_K(\kappa)\subseteq \Gamma^+_K(\kappa)$, each vertex $p_i$ of $\phi$ lies above the hyperplane $H_\kappa=H(\tau_\kappa)$.
Furthermore, as each of the downward segments $e(p_i,G_\kappa)$ of the vertices of $\phi$ is crossed by the hyperplane $G_\kappa\in \R$, it follows, by the convexity of $\phi$, that it lies entirely above $G_\kappa$, and the $\ell_\kappa$-intercepts of $H_\kappa$, $G_\kappa$ and $\phi$ appear along $\ell_\kappa$ in this increasing order of their $d$-th coordinates.

\begin{figure}
    \begin{center}      
    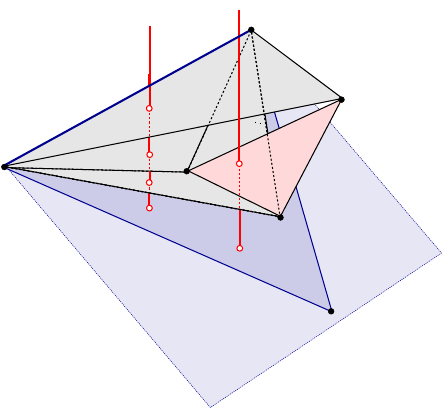\hspace{1.7cm}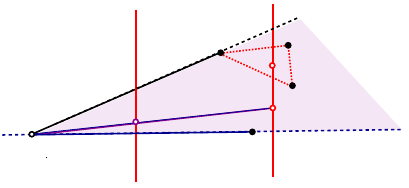
        \caption{\small Proof of Proposition \ref{Claim:YL} in dimension $d=3$. Left: Depicted are the $4$ simplices $\tau_\kappa$, $\mu_1$, $\mu_2$ and $\mu_3$, whose intercepts appear in this order along $\ell_0$. 
        Right: Depicted in the $\kappa$-parallel direction is the dihedral wedge $W$ which intersects $\ell_\kappa$ and contains $\phi=\conv(p_1,p_2,p_3)$. Since $x_\kappa=\ell_\kappa\cap G_\kappa$ belongs to $X_{\ell_\kappa}$, and lies in $W$, the hyperplane through $\pi_\kappa=\kappa\ast \{x_\kappa\}$ lies within $W$, and meets $\ell_0$ at some point $y_\kappa$ of $Y_{\ell_0}\cap K$.} 
        \label{Fig:WedgeW}
    \end{center}
\end{figure}

Recall that each vertex $p_i$ of $\phi$ corresponds to a $\kappa$-adjacent simplex $\mu_i=\conv(\kappa\cup\{p_i\})$ of $\Pi'_K(\ell_0)\subseteq \Pi'_K$ (which, in particular, meets $\ell_0$). Assume with no loss of generality that the $\ell_0$-intercept $\mu_d\cap \ell_0$ of $\mu_d$ lies above the rest of such $\ell_0$-intercepts $\mu_i\cap \ell_0$, for $1\leq i\leq d-1$. See Figure \ref{Fig:WedgeW} (left). 
Consider the half-open dihedral wedge $W$ between $H_\kappa$ and $H(\mu_d)$, which is comprised of all such points that lie above $H_\kappa=H(\tau_\kappa)$ and below (or at) $H(\mu_d)$.
Then $W$ contains the entire simplex $\phi=\conv(p_1,\ldots,p_d)$ and, moreover, its cross-section $W\cap \ell_0$ is contained in $K$, as it is delimited by the $\ell_0$-intercepts of the simplices $\tau_\kappa,\mu_d\in \Pi'_K(\ell_0)$. See Figure \ref{Fig:WedgeW} (right). 
 
 Therefore, to ``place" the intersection point $H(\pi_\kappa)\cap \ell_0$ within the interval $\ell_0\cap K$, it suffices to check that the point $x_\kappa=G_\kappa\cap \ell_\kappa$ too lies in the wedge $W$
 (so that the hyperplane $H(\pi_\kappa)$, through $\kappa$ and $x_\kappa$, would remain in the same wedge $W$ between $H(\tau_\kappa)$ and $H(\mu_d)$ and, thereby, cross $K\cap \ell_0$, at the aforementioned point $y_\kappa$).
 However, this easily follows from the fact that, by the convexity of $W$ and $\phi$, the point $z=\ell_\kappa\cap \phi$ lies in $W$; thus, $x_\kappa=e(z,H_\kappa)\cap G_{\kappa}$ has to lie below $z$ (and, therefore, below $H(\mu_d)$), and above $H_\kappa$.
\end{proof}

To recap, Proposition \ref{Claim:GoodSimplices} yields at least $(\vartheta/4)\left(\frac{\sigma_0}{2}\right)^{\beta_{d-1}}{\eps n \choose d-1}$ good $(d-2)$-simplices $\kappa\in {P'_K\choose d-1}$. Each of these simplices $\kappa$ (whether it is very good, or not) contributes a $(d-1)$-simplex $\pi_\kappa\in \Pi_X$ whose supporting hyperplane $H(\pi_\kappa)$ meets $\ell_0$ at a point of $Y_{\ell_0}\cap K$. Furthermore, the generic choice of $\L(P,s)$ easily implies that no pair of these simplices $\pi_\kappa,\pi_{\kappa'}\in \Pi_X$ meet $\ell_0$ at the same point. Hence, the interval $K\cap \ell_0$ must contain at least $(\vartheta/4)\left(\frac{\sigma_0}{2}\right)^{\beta_{d-1}}{\eps n \choose d-1}$ points of $Y_{\ell_0}$ and, therefore, at least one point of the 1-dimensional net $N_{\ell_0}\subset N_{\leq \delta,2}$ (that was constructed over $Y_{\ell_0}$), which concludes the proof of Lemma \ref{Lemma:PrincipalHyperplane}.
\end{proof}

\subsection{Step 3}\label{Subsec:3rdNetHollow}
To finish the proof of Theorem \ref{Theorem:BoundHollow}, we describe the third net $N_{\leq \delta,3}$, which pierces all the remaining sets in $\K_{\leq \delta}$, and then bound the cardinality of the entire net $N_{\leq \delta}=\bigcup_{i=1}^3N_{\leq \delta,i}$. 



\medskip
\noindent{\bf Setup.} Recall that every cell $\varphi$ in our cutting $\D=\D(\R,P)$ encompasses at most $\lfloor n/k\rfloor$ points of $P$, and is colored as either red or blue. Namely, the blue cells are the simplicial cells that have been inherited en toto from the bottom-vertex triangulation $\D(\R)$ of the sampled arrangement $\A(\R)$, and their number can be as large as $O\left(r^d\right)$. In contrast, the number of the red cells is at most $2k$. 

In view of Lemmas \ref{Lemma:1stNetHollow} and \ref{Lemma:PrincipalHyperplane}, every remaining set $K\in \K_{\leq \delta}$ is now endowed with the principal hyperplane $H_K\in \HH(\Pi'_K)$, which crosses at most $4c_{\it cut}\cdot h\frac{k}{r}\log r=O\left(h\frac{k}{r}\log r\right)$ red cells of $\D$, and whose zone in $\D$ encompasses a subset $A_K\subseteq P_K$ of at least $c'_2 \sigma_0^{\beta_{d-1}} \lceil\eps n \rceil=\Theta(|P_K|)$ points which see $H_K$ within the underlying random arrangement $\A(\R)$. (Here $c_{\it cut}>0$ denotes the constant in Observation \ref{Prop:Cutting}.)

\medskip
\noindent{\bf Definition.} For any open cell $\Delta$ in $\A(\R)$, and any open cell $\varphi$ in the cutting $\D$, let us denote $P_\Delta:=P\cap \Delta$, and $P_\varphi:=P\cap \varphi$. Note that both subsets $P_\Delta$ and $P_\varphi$ are contained in $P'=P\setminus \left(\bigcup \R\right)$, and every point of $P'$ belongs to exactly one subset of each kind.

\medskip
 We will say that a point of $P'$ is {\it blue} (resp., {\it red}) if it lies in a blue (resp., purple) cell of $\D(\R,P)$, and denote the resulting two subsets of $P'$ by, respectively, $B$ and $R$, so that $P'=R\uplus B$. 
Note that a given open cell $\Delta$ of $\A(\R)$ may give rise to both red and blue cells within $\D$, so that its subset $P(\Delta)$ may include points of both color classes $R$ and $B$.

\medskip
\noindent{\bf Overview.} Here is a high-level sketch of our divide-and-conquer treatment of the remaining sets $K\in \K_{\leq \delta}$. Since at least half of the points in every subset $A_K$, that is induced by some remaining $\delta$-hollow set $K\in \K_{\leq \delta}$, must be of the same color, our task comes down to constructing a separate $\Theta(\eps)$-net $N_R$ (resp., $N_B$) over the point set $R$ (resp., $B$).

\begin{enumerate}
	\item {\it The net $N_R$.} Since the red points in every subset $A_K$ are distributed over $O\left(h\frac{k}{r}\log r\right)$ red cells $\varphi\in \D$, which are crossed by $H_K$, and each of these cells satisfies $|P_\varphi|\leq \lfloor n/k\rfloor$, it is enough to construct for each of these cells $\varphi\in \D$ a $\tilde{\eps}$-net $N_\varphi$ with some $\tilde{\eps}>0$ that satisfies
	 $$
	 \tilde{\eps}=\Theta\left(\frac{\eps n}{h\cdot \frac{k}{r}\log r\cdot (n/k)}\right)=\Theta\left(\frac{\eps \cdot r}{h\log r}\right) \gg \eps\cdot r^{1-\eta},
	 $$ 
	 \noindent for the total cardinality $|N_R|=O\left(k \cdot f\left(\eps \cdot r^{1-\eta}\right)\right)=O\left(r^{\alpha_d}\cdot f\left(\eps\cdot r^{1-\eta}\right)\right)$. (Here the last estimate stems from our choice $k=r^{(d+\sqrt{d^2-2d})/2}$, which is bounded by $r^{\alpha_d}$ for all $d\geq 3$.)
	 \item {\it The net $N_B$.} If the majority of the points in the subset $A_K$ are blue, these points can be spread, within $\D$, over a far larger number of the blue cells. 
	 To still end up with an efficient recurrence, our decomposition of the blue set $B$ will be based on the cells of the random arrangement $\A(\R)$.
	 
	 A simple random sampling argument will yield a $(d-1)$-dimensional polytope $T_K\subset H_K$ with the property that a fixed fraction of the blue points of $A_K$ lie in such cells $\Delta$ of $\A(\R)$ that meet the relative boundary of $T_K$. According to Lemma \ref{Lemma:Zone}, the overall number of the latter cells $\Delta$ in $\A(\R)$ is only $O\left(r^{d-2}\right)$.
	 Though some of these cells $\Delta\in \A(\R)$ may contain far more than $n/r^d$ points of $B$, Theorem \ref{Theorem:ManyCells} implies that such excessively ``heavy" cells $\Delta$ encompass relatively few points of $B$, which can be ``dispatched" by the means of a separate recursive $\Omega(\eps\log r)$-net. 
	 	 
	 With the aforementioned exception, we will construct for each cell $\Delta$ in $\A(\R)$ a separate $\eps_\Delta$-net, this time over the respective ``blue" subset $B\cap \Delta$. 
	 Choosing a suitable threshold $\eps_\Delta$ for each cell $\Delta\in \A(\R)$ will again result in a favourable upper bound on $|N_{B}|$, which involves only the terms of the ``right" type $O\left(w^{\alpha_d+\eta}\cdot f\left(\eps\cdot w^{1-\eta}\right)\right)$ (where $w$ is a fixed and positive degree of $r$).
\end{enumerate}

Let us now cast the missing details in the outlined construction of the nets $N_R$ and $N_B$, over the respective color classes $R$ and $B$ within $P'$. To this end, let us fix a suitably small constant $c'_3>0$, which will be chosen in the sequel so as to suit the analysis of $N_{>\delta,3}$. (This choice will depend on $d\geq 3$, $\eta>0$, $\sigma_0>0$, and also the constants $c_0$ and $c'_2$ in, respectively, Lemma \ref{Lemma:ConvexProjSample} and Lemma \ref{Lemma:PrincipalHyperplane}.)

For the rest of this construction, and its subsequent analysis, it can be assumed that $\eps<c'_3$ for, otherwise, any of the previous methods \cite{AlonSelections,Chazelle,MatWag04} yields a constant-size net $N$ of cardinality $O\left(1/{c'_3}^{d+1}\right)$ for the entire family $\K(P,\eps)$.

In what follows, we use $\D_R$ (resp., $\D_B$) to denote the sub-family of all the red (resp., blue) cells in our cutting $\D=\D(\R,P)$.

\medskip
\noindent{\bf The ``red" net $N_R$.}
	For every red cell $\varphi\in \D$ we construct the net $N_\varphi:=N(P_\varphi,\tilde{\eps})$ for the family $\K\left(P_\varphi,\tilde{\eps}\right)$, with
	\begin{equation}\label{Eq:RedNet}
		\tilde{\eps}:=\frac{c'_2 \sigma_0^{\beta_{d-1}} \lceil \eps n\rceil}{8\cdot (n/k)\cdot c_{\it cut}\cdot \frac{hk}{r}\log r}=\Theta\left(\frac{r}{h\log r}\cdot \eps\right).
	\end{equation}

 \noindent We then set 
 $$
 N_R:=\bigcup_{\varphi\in \D_R}N_\varphi.
 $$

\medskip
\noindent{\bf The ``blue" net $N_B$.}  
Let

\begin{equation}\label{Eq:HeavyCells}
M:=c'_3\cdot \frac{k^2}{r^d(\log r)^{2\zeta_d+2}},
\end{equation}

\noindent where $\zeta_d>0$ is the constant in Theorem \ref{Theorem:ManyCells}.

For every cell $\Delta\in \A(\R)$, let us denote $B_\Delta:=B\cap \Delta$.
We then fix a threshold $b:=n/M$ and classify the cells $\Delta\in \A(\R)$ according to the number $|B_\Delta|$ of the blue points in each cell.
The first sub-family $\A'$ is comprised of all such cells $\Delta\in \A(\R)$ that satisfy either $|B_\Delta|<c'_3n/r^{d+2}$ or $|B_\Delta|\geq b$. The second sub-family $\tilde{\A}:=\A(\R)\setminus \A'$ encompasses all the remaining cells $\Delta$ of $\A(\R)$, which satisfy
	$$
c'_3n/r^{d+2}	\leq |B_\Delta|<b.
	$$


 



The net $N_B$ combines the following ingredients.

\begin{enumerate}
	\item {\it The nets $N_\Delta$.} For every cell $\Delta\in \tilde{\A}$, we include in $N_B$ the net $N_\Delta:=N\left(B_\Delta,\eps_\Delta\right)$ for the family $\K\left(B_\Delta, \eps_\Delta\right)$, with
\begin{equation}\label{Eq:Ni}
\eps_\Delta:=c'_3\cdot \frac{\sigma_0^{\beta_{d-1}} \eps}{r^{d-2}}\frac{n}{|B_\Delta|},
\end{equation}

\item {\it The net $N_{B'}$.} Lastly, we let $B':=\bigcup_{\Delta\in \A'}B_\Delta$, and add to $N_B$ the net $N_{B'}:=N\left(B',\eps \log r\right)$ for the family $\K\left(B',\eps\log r\right)$.

\end{enumerate}

\medskip
\noindent {\bf The analysis.} The properties of our third (and last) net $N_{\leq \delta,3}:=N_R\cup N_B$ for the family $\K_{\leq \delta}$ are summarized in the following lemma, whose somewhat technical proof is postponed to Section \ref{Subsec:3rdNetHollowProofs}.

\begin{lemma}\label{Lemma:3rdNetHollow}
\begin{enumerate}

    \item With a suitably small choice of the constant $c'_3>0$, every remaining convex set in $\K_{\leq \delta}$, that was missed by the combination $N_{\leq \delta,1}\cup N_{\leq \delta,2}$, is pierced by the just described net $N_{\leq \delta,3}=N_R\cup N_B$.

	\item The cardinality of $N_{\leq \delta,3}$ satisfies 
$$
|N_{\leq \delta,3}|=O\left(f\left(\eps\log r\right)+k\cdot f\left(\eps\cdot r^{1-\eta}\right)+\sum_{i=1}^{l} \frac{2^ik^2}{r^d}\cdot  f\left(\eps\cdot 2^i\cdot \frac{k^2}{r^{2d-2+\eta/(10d)}}\right)\right)
$$
\noindent for some $l=O(\log r)$.

\end{enumerate}

\end{lemma}

\subsection{Proof of Lemma \ref{Lemma:3rdNetHollow}}\label{Subsec:3rdNetHollowProofs}

\noindent{\bf Part 1.} Let us first check that every set $K\in \K_{\leq \delta}$ (that survived so far) is pierced by $N_{\leq\delta,3}=N_R\cup N_B$, provided a suitably small choice of the constant $c'_3>0$ in (\ref{Eq:HeavyCells}). To this end, we fix such a set $K$, and distinguish between several scenarios according to the distribution of the (at least) $c'_2 \sigma_0^{\beta_{d-1}}\lceil \eps n\rceil$ points of the subset $A_K\subseteq P_K$ described in Lemma \ref{Lemma:PrincipalHyperplane}. 

\medskip
\noindent {\bf Case (a).} Assume first that at least $|A_K|/2$ of all points in $A_K$ are red and, therefore, belong to the set $R=\bigcup_{\varphi\in \D_R}P_\varphi$. Since the principal hyperplane $H_K$ belongs to $\HH(\Pi')$ and, therefore, crosses a total of at most $4c_{\it cut}\cdot h\frac{k}{r}\log r$ red cells $\varphi\in \D_R$, there must be such a cell $\varphi\in \D_R$ that encompasses at least
$$
\frac{c'_2 \sigma_0^{\beta_{d-1}}\eps n}{8c_{\it cut}\cdot h\frac{k}{r}\log r}\geq \tilde{\eps}\cdot \frac{n}{k}\geq \tilde{\eps}\cdot |P_\varphi|.
$$

\noindent points of $A_K$ (where the second inequality uses the definition (\ref{Eq:RedNet}) of $\tilde{\eps}$). Hence, such a set $K$ is pierced by the local net $N_\varphi\subseteq N_{R}$, which was constructed for the instance $\K\left(P_\varphi,\tilde{\eps}\right)$.

\medskip
\noindent {\bf Case (b).} Let us assume, then, that at least $|A_K|/2\geq c'_2 \sigma_0^{\beta_{d-1}}\lceil \eps n\rceil/2$ among the points of $A_K$ belong to $B$ and, thereby, lie in the blue cells $\varphi\in \D_B$. 
In view of the (potentially) very non-uniform distribution of the blue points between the cells of $\A(\R)$, two sub-scenarios are to be considered.

\medskip
\noindent{\bf Case (b1).} If $|A_K\cap B'| \geq |A_K|/4$, then a suitably small choice of $c'_3>0$ yields
$$
|A_K\cap B'|\geq c'_2 \sigma_0^{\beta_{d-1}}\lceil \eps n\rceil /4\geq \eps \log r \cdot |B'|,
$$
\noindent where the last inequality follows from Proposition \ref{Prop:SpecialSets} below. As a result, such a set $K$ is $(\eps\log r)$-heavy with respect to $B'$ and, therefore, must have been pierced by the net $N_{B'}\subseteq N_B$.

\begin{proposition}\label{Prop:SpecialSets}  With a suitably small choice of the constant $c'_3>0$ in the definition of $N_{\leq \delta,3}$, we have that
$|B'|\leq c'_2\sigma_0^{\beta_{d-1}} n/\left(6\log r\right)$.

\end{proposition}
\begin{proof}
Since $|\A'|\leq |\A(\R)|=O\left(r^d\right)$, the cells of $\Delta\in \A'$ that satisfy $|B_\Delta|< c'_3 n/r^{d+2}$ account for a total of 
\begin{equation}\label{Eq:LightCells}
	O\left(|\A(\R)|\cdot n/r^{d+2}\right)=O\left(n/r^2\right)
\end{equation}
\noindent points of $B'$. 

The remaining points of $B'$ are covered by at most $\lfloor n/b\rfloor\leq M$ cells of $\A'$ which satisfy $|B_\Delta|\geq b$.
Using the estimate of Theorem \ref{Theorem:ManyCells} (and the definition (\ref{Eq:HeavyCells}) of $M$), we obtain the bound

\begin{equation}\label{Eq:SmallTotalComplexity}
O\left(M^{1/2}r^{d/2}(\log r)^{\zeta_d}\right)=O\left(k/\log r\right)
\end{equation}

\noindent on the overall complexity of the latter cells $\Delta$ in $\A'$. Hence, their bottom vertex triangulation yields a total of $O\left(k/\log r\right)$ simplices of $\D_B$, each of them containing at most $n/k$ points of $B'$, for a total amount of $O\left(\left(k/\log r\right) \cdot (n/k)\right)=O\left(n/\log r\right)$ points.

Provided a small enough, albeit, fixed choice of $c'_3>0$ (which depends on $\sigma_0$ and $c'_2$), the implicit constants of proportionality on the right hand sides of both estimates (\ref{Eq:LightCells}) and (\ref{Eq:SmallTotalComplexity}) are small enough to ensure that the cells of $\A'$ encompass a total of at most $c'_2 \sigma^{\beta_{d-1}}n/(6\log r)$ points of $B$.
\end{proof}

\medskip
\noindent {\bf Case (b2).} We have that 

\begin{equation}\label{Eq:ManyBlue}
|A_K\cap \tilde{B}| \geq |A_K|/4\geq c'_2 \sigma_0^{\beta_{d-1}}\lceil \eps n\rceil/4,
\end{equation}

\noindent where $\tilde{B}$ denotes the set $B\setminus B'=\bigcup_{\Delta\in \tilde{\A}} B_\Delta$.
 Then the claim is an straightforward consequence of the following observation.

\begin{lemma}\label{Proposition:SampleBoundary}
Let $K\in \K_{\leq \delta}$ be a convex set that is missed by $N_{\leq \delta,1}\cup N_{\leq \delta,2}$, and that falls into case (b2). Then, provided a sufficiently small choice of the constant $c'_3>0$, there is a $(d-1)$-dimensional convex polytope $T_K$ within $H_K$, and a subset $B_{K}\subseteq A_K\cap \tilde{B}$ of at least $|A_K|/12$ such blue points whose ambient cells in $\A(\R)$ meet the relative boundary of $T_K$ within $H_K$.
\end{lemma}

To complete the proof of Lemma \ref{Lemma:3rdNetHollow} using Lemma \ref{Proposition:SampleBoundary},
let $\tilde{\A}_K$ denote the subset of all such cells $\Delta\in \tilde{\A}$ that meet the relative boundary of the $(d-1)$-dimensional polytope $T_K\subset H_K$ described in Lemma \ref{Proposition:SampleBoundary}.
Notice that every cell $\Delta\in \tilde{\A}_K$ corresponds to a unique $(d-1)$-dimensional cell $\Delta\cap H_K$, which too meets the relative boundary of $T_K$, and is part of the arrangement that is determined by the family $\{H\cap H_K\mid H\in \R\}$ of $(d-2)$-planes within $H_K\equiv \reals^{d-1}$.
Applying Lemma \ref{Lemma:Zone} to this $(d-1)$-dimensional arrangement implies the estimate $|\tilde{\A}_K|\leq cr^{d-2}$, with a suitable constant $c>0$ which depends on the dimension $d\geq 3$. 
As a result, the pigeonhole principle yields such a cell $\Delta\in \tilde{\A}_K$ that satisfies 
$$
|B_K\cap B_\Delta|\geq \frac{|B_K|}{cr^{d-2}}\geq \frac{|A_K|}{12cr^{d-2}}\geq \frac{c'_2}{12cr^{d-2}}\cdot \sigma_0^{\beta_{d-1}}\lceil \eps n\rceil.
$$
\noindent Hence, provided that $c'_3\leq c'_2/(12c)$, such a convex set $K$ must have been pierced by the ``local" net $N_\Delta\subseteq N_B$, which was constructed for the family $\K\left(B_\Delta,\epsilon_\Delta\right)$.

\begin{proof}[Proof of Lemma \ref{Proposition:SampleBoundary}.] 
Let us fix a convex set $K\in \K_{\leq \delta}$ as in the hypothesis, and denote $\tilde{B}_K:=A_K\cap \tilde{B}$. Consider the partition 
$\left\{A_K\cap B_\Delta \mid \Delta\in \tilde{\A}\right\}$
of $\tilde{B}_K$, by the cells of $\tilde{\A}$, into $|\tilde{\A}|\leq (d+1){r\choose d}\leq r^d$ subsets (see, e.g., \cite[Proposition 6.1.1]{JirkaBook}). According to Lemma \ref{Lemma:Sample}, a random $\tilde{r}$-sample $S\subseteq \tilde{B}_K$, with $\tilde{r}=\min\left\{\left\lceil 10^3 r^d\right\rceil,|\tilde{B}_K|\right\}$, meets the following condition with probability at least $1/2$: 

\medskip
\noindent {\tt (C1)} {\it The cells of $\tilde{\A}_S=\{\Delta\in \tilde{\A}\mid S\cap \Delta\neq \emptyset\}$ encompass a total of at least $|\tilde{B}_K|/3$ points of $\tilde{B}_K$.}

\medskip
\noindent Furthermore, the fact that the principal subset $P_K$ is $\delta$-hollow implies, in particular, that there exist at most $\delta{\lceil \eps n\rceil\choose d+1}$ punctured $(d+1)$-subsets in ${\tilde{B}_K\choose d+1}$. Therefore, and because of (\ref{Eq:ManyBlue}), the point set $\tilde{B}_K$ must be vertically $\tilde{\delta}$-convex, with $\tilde{\delta}\geq 0$ that satisfies
$$
\tilde{\delta}\leq \delta {\lceil \eps n\rceil\choose d+1}/{|\tilde{B}_K|\choose d+1}  \leq \delta {\lceil \eps n\rceil\choose d+1}/{\lceil c'_2 \sigma_0^{\beta_{d-1}} \eps n\rceil/4\choose d+1}
=O(\delta)=O\left(1/r^{d(d+1)+1}\right).
$$

\begin{figure}
    \begin{center}      
        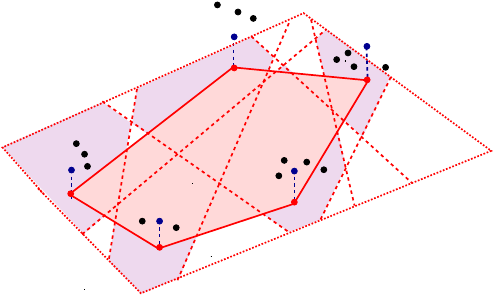
        \caption{\small Proof of Lemma \ref{Proposition:SampleBoundary}. We sample a $0$-hollow subset $S$ of $\tilde{r}=\min\left\{\lceil 10^3r^d\rceil,|\tilde{B}_K|\right\}$ blue points from $\tilde{B}_K$. The projected red points of $[S]_{H_K}$ determine the $(d-1)$-dimensional convex polytope $T_K:=\conv\left([S]_{H_K}\right)$ within the principal hyperplane $H_K$. The cross-sections $\Delta\cap H_K$ of the ambient cells $\Delta\in \A(\R)$ of the points $p\in S$ are shaded as well.}
        \label{Fig:ZoneU}
    \end{center}
\end{figure}

Since we have that $r=\left\lceil (1/\eps)^{\tilde{\eta}^3}\right\rceil$, $\tilde{\delta}=o\left(1/r^{d(d+1)}\right)$, and $\eps<c'_3$, any sufficiently small, albeit fixed choice of $c'_3>0$ (which does not depend on the on the quantities $\eps$ and $n$) yields the inequality

$$
\tilde{\delta}\leq \frac{c_0}{\left\lceil 10^3 r^d \right\rceil^{(d+1)}}\leq \frac{c_0}{\tilde{r}^{d+1}}.
$$ 

\noindent According to Lemma \ref{Lemma:ConvexProjSample}, the sample $S$ 
 is $0$-hollow with probability at least $0.99$ (as depicted in Figure \ref{Fig:ZoneU}). Hence, with probability at least $0.4$, it both (i) satisfies the previous condition (C1), and (ii) is $0$-hollow.

Let us assume, then, that both criteria hold for the sample $S$ at hand. Consider its $H_K$-projection 
$[S]_{H_K}:=\{[p]_{H_K}\mid p\in S\}$. We claim that the $(d-1)$-dimensional polytope $T_K:=\conv\left([S]_{H_K}\right)\subset H_K$ meets the criteria of the proposition with 
$B_K:=\tilde{B}_K\cap \left(\bigcup \tilde{\A}_S\right)$. 
In view of property (C1), it suffices to check that every cell of $\tilde{\A}_S$ meets the relative boundary of $T_K$ within $H_K$.
Indeed, fix such a cell $\Delta\in \tilde{\A}_S$, with a point $p\in S\cap \Delta$. Then, by the $0$-hollowness of our sample $S$, the point $p$ projects to a boundary vertex $[p]_{H_K}$ of $T_K$. Furthermore, since $p$ sees $H_K$ (as $S\subseteq A_K$), the vertex $[p]_{H_K}$ too lies in the cell $\Delta$, which then intersects the relative boundary of $T_K$. 
\end{proof}

\noindent{\bf Part 2.} Let us now bound the cardinality of $N_{\leq \delta,3}=N_R\cup N_B$. 
To estimate $|N_R|$, recall that $|\D_{R}|\leq 2k$. Together with that fact that $h\lll_\eta r\lll_\eta 1/\eps$, this yields
$$
|N_R|=\left|\left(\bigcup_{\varphi\in \D_{R}} N_\varphi\right)\right|=O\left(k\cdot f\left(\frac{\sigma_0^{\beta_{d-1}}r}{h\log r}\cdot \eps\right)\right)=O\left(k\cdot f\left(\eps\cdot r^{1-\eta}\right)\right).
$$

To bound the remaining quantity $|N_B|=|N_{B'}|+\sum_{\Delta\in \tilde{\A}}|N_\Delta|$, note first that $|N_{B'}|\leq f\left(\eps\log r\right)$.

To estimate the overall cardinality of the remaining ``blue" nets $N_\Delta$, let us first arrange their respective cells $\Delta\in \tilde{\A}$ into sub-families with roughly the same density of the blue points in each cell. To this end, we fix  $l:=\log \lceil r^{d+2}/M\rceil=O(\log r)$, and set

$$
\tilde{\A}_i:=\left\{\Delta\in \tilde{\A}\ \Bigg\vert \frac{b}{2^i} \leq |B_\Delta|<\frac{b}{2^{i-1}}\right\}
$$

\noindent for all $1\leq i\leq l$. Notice that $\tilde{\A}=\biguplus_{i=1}^l\tilde{\A}_i$ and, according to (\ref{Eq:Ni}), every cell $\Delta\in \A_i$ satisfies

$$
\eps_\Delta=c'_3\cdot \frac{\sigma_0^{\beta_{d-1}} \eps}{r^{d-2}}\frac{n}{|B_\Delta|}=\Theta\left(\frac{\eps}{r^{d-2}}\cdot 2^iM\right)=\Theta\left(\frac{2^i k^2\eps}{r^{2d-2}\cdot (\log r)^{2\zeta_d+2}}\right).
$$

\noindent Since we have that

$$
\left|\tilde{\A}_i\right|\leq 2^i n/b=\Theta(2^i\cdot M)=\Theta\left(\frac{2^i k^2}{r^d(\log r)^{2\zeta_d+2}}\right), and
$$
\noindent and $\log r\lll_\eta r$, it follows that
$$
\left|\bigcup_{\Delta\in \tilde{\A}_i}N_\Delta\right|\leq \left|\tilde{\A}_i\right|\cdot f\left(\frac{2^i k^2\eps}{r^{2d-2}\cdot (\log r)^{2\zeta_d+2}}\right)=O\left(\frac{2^ik^2}{r^{d}}f\left(\eps \cdot \frac{2^i k^2}{r^{2d-2+\eta/(10d)}}\right)\right)
$$

\noindent for all $1\leq i\leq l$, which concludes the proof of the second part of Lemma \ref{Lemma:3rdNetHollow}. $\Box$

\subsection{Wrap-up for the $\delta$-hollow sets} 
According to the first part of Lemma \ref{Lemma:3rdNetHollow}, every set in $\K_{\leq \delta}$ is pierced by the combined net
$$
N_{\leq \delta}=N_{\leq \delta,1}\cup N_{\leq \delta,2}\cup N_{\leq \delta,3}.
$$

\noindent Putting together the bounds in Lemmas \ref{Lemma:1stNetHollow}, \ref{Lemma:2ndNetHollow}, and \ref{Lemma:3rdNetHollow} yields

$$
|N_{\leq \delta}|\leq f(\eps,\lambda/h,\sigma/2)+f\left(\eps\log r\right)+
$$

$$
+O\left(k\cdot f\left(\eps\cdot r^{1-\eta}\right)+\sum_{i=1}^{l} \frac{2^ik^2}{r^d}\cdot  f\left(\eps\cdot 2^i\cdot \frac{k^2}{r^{2d-2+\eta/(10d)}}\right)\right)+O\left(\frac{1}{\eps^{d-1+\eta}}\right),
$$

\noindent where $l=O(\log r)$.
Recalling our choice $k=r^{(d+\sqrt{d^2-2d})/2}=r^{d-1/2-o_d(1)}$, and fixing $w:=k^2/r^{2d-2+\eta/(10d)}$, yields

$$
|N_{\leq \delta}|\leq f(\eps,\rho/h,\sigma/2)+f\left(\eps \cdot \log r \right)+
$$

$$
+O\left(r^{(d+\sqrt{d^2-2d})/2+\eta}\cdot  f\left(\eps\cdot r^{1-\eta}\right)+\sum_{i=1}^{l} 2^iw^{(d+\sqrt{d^2-2d})/2+\eta}\cdot f\left(\eps \cdot 2^i  w\right)\right)+
$$

$$
+O\left(\frac{1}{\eps^{d-1+\eta}}\right).
$$
\noindent Since we have that $\alpha_d\geq (d+\sqrt{d^2-2d})/2$ for all $d\geq 3$, this completes the proof of Theorem \ref{Theorem:BoundHollow}. $\Box$

\section{Concluding remarks}\label{Sec:Final}
\begin{itemize}

\item Despite the upper bound $f_d(\eps)=O\left((1/\eps)^{d-1/2-o_d(1)}\right)$ that Theorem \ref{Thm:Main} yields in all dimensions $d\geq 4$, the present analysis cannot be used to show that $f_2(\eps)=o\left((1/\eps)^{3/2}\right)$. (Nevertheless, it can be easily adapted to show a weaker bound $f_2(\eps)=O\left((1/\eps)^{a}\right)$, for some constant $3/2<a<2$.)

\item 
The author conjectures that the actual asymptotic behaviour of the functions $f_d(\eps)$  in any dimension $d\geq 1$ is close to $1/\eps$, as is indeed the case for their ``strong" counterparts which exist with respect to simply shaped objects in $\reals^d$ \cite{HW87}.

\item The author anticipates that a major improvement of the upper bound on $f_d(\eps)$ (i.e., beyond $O^*\left(1/\eps^{d-1}\right)$) would require stronger selection-type results for piercing simplices with points, lines and other flats, in the general vein of Theorem \ref{Theorem:MultipleSelection}.

\item 
The recently improved analysis of the transversal numbers $C_d(p,q)$ that arise in the Hadwiger-Debrunner problem (see Section \ref{sec:intro}), due to Keller, Smorodinsky, and Tardos \cite{Shakhar}, implies that

\begin{equation}\label{Eq:Chaya}
C_d(p,q)\leq f_d\left(\Omega\left(p^{-1-\frac{d-1}{q-d}}\right)\right).
\end{equation}

Here, as before, $f_d(\eps)$ denotes the smallest possible number $f$ with the property that any $n$-point set $P\subset \reals^d$ admits a weak $\eps$-net of cardinality $f$ with respect to convex sets.

Plugging the result of Theorem \ref{Thm:Main}, and the previous planar bound \cite{FOCS18}, into 
(\ref{Eq:Chaya}) yields improved bounds in all dimensions $d\geq 2$, namely,

$$
C_2(p,q)=O\left(p^{(3/2+\gamma)\left(1+\frac{d-1}{q-d}\right)}\right)
$$

and

$$
C_d(p,q)=O\left(p^{(\alpha+\gamma)\left(1+\frac{d-1}{q-d}\right)}\right),
$$

for all $d\geq 3$. Here $\gamma>0$ is an arbitrary positive constant, and $p$ has to be larger than a certain constant threshold which depends on $\gamma$.

\item Our proof of Theorem \ref{Thm:Main} is fully constructive and combines the following explicit ingredients across the various instances $\K\left(P,\Pi,\eps,\sigma\right)$ which are encountered in the course of our recurrence in $\eps>0$, the ground set $P$, and the subset $\Pi\subseteq {P\choose d}$ of $(d-1)$-simplices:

\begin{enumerate}
\item $1$-dimensional strong $\Omega^*(\eps^{\alpha_d})$-nets that are constructed via Theorem \ref{Theorem:Sparse}  within the lines $\ell$ of the canonical families $\L(P,s)$. Each of these nets is defined over the $\ell$-intercepts $\ell\cap \tau$ of the $(d-1)$-simplices $\tau\in \Pi$.

\item $1$-dimensional strong $\Omega^*\left(\eps^\alpha\right)$-nets that are obtained in Step 2 of the construction of $N_{>\delta}$, for the ``secondary" canonical lines $\ell\in \L(P_i,u)$. Each of these nets $N'_\ell$ is defined with respect to the $\ell$-intercepts of the $O^*\left(t^d\right)$ auxiliary $(d-1)$-simplices of ${V'\choose d}$ (where the set $V'$ is comprised of all the vertices of the clipped cells $\varphi'(p)$ which have been constructed around the points $p\in P$).

\item $1$-dimensional strong $\Omega^*\left(\eps^{d-1}\right)$-nets that are obtained in Step 2 of the construction of $N_{\leq \delta}$, for the lines $\ell\in \L(P,s)$. Each of these nets $N_\ell$ over the $\ell$-intercepts of the ``mixed" $(d-1)$-simplices of ${P\choose d-1}\ast X$. (The set $X$ is comprised of the $\ell'$-intercepts of the sampled hyperplanes of $\R$, taken over all canonical lines $\ell'\in \L(P,s)$, and also of the vertices of the cutting $\D$.)

\item Strong $\Theta\left(\frac{\delta\eps}{st^{1-1/d}}\right)=\Omega(\eps^{\alpha_d})$-nets that are constructed in Section \ref{Subsec:MainPartition}, via Lemma \ref{Thm:StrongNet}, with respect to convex $(2d)$-hedra.
\end{enumerate}

\item As the primary focus of this study is on the quantity $f_d(\eps)$, we did not seek to optimize the construction cost of our net $N$.
A straightforward implementation of the recursive construction runs in time $O^*\left(n^{d}\right)$, which is dominated by the non-recursive overhead spent on maintaining the sets $\Pi\subseteq {P\choose d}$ and tracing the zones of the hyperplanes of $\HH(\Pi)$ within the cutting $\D=\D(P,\R)$ of Section \ref{Sec:VerticallyConvex}. This also accounts for the cost of constructing the $1$-dimensional nets of Theorem \ref{Theorem:Sparse} with respect to the $\ell$-intercepts $\ell\cap\tau$, for all simplices $\tau\in \Pi$ and canonical lines $\ell\in \L(P,s)$.

Since $s\lll_\eta u\lll_\eta 1/\eps$, all the vertical partitions $\V(P,s)$ and $\V(P_i,u)$, along with the induced canonical line families $\L(P,s)$ and $\V(P_i,u)$, can be obtained in $O^*(n)$ time via Matou\v{s}ek's Theorem \ref{Theorem:Simplicial}. Since $r\lll_\eta s$, the same holds true for the arrangement $\A(\R)$ and its cutting $\D=\D(\R,P)$ in Section \ref{Sec:VerticallyConvex}, and the assignment of the points of $P$ to their ambient cells in $\A(\R)$ and $\D$. 
Furthermore, as Matou\v{s}ek's partition theorem holds with any partition parameter {$1\leq r\leq n$}, the secondary partition $\P(P,s,t)$ in Section \ref{Sec:Surrounded} can be constructed in similar time $O^*(n)$. Armed with the partition $\P(P,s,t)$, the set $\tilde{\Pi}$ in Lemma \ref{Lemma:SparsePi} can be obtained in $O\left(|\Lambda(P,s,t)|\cdot (st)^{d-1}+n^d\right)=O\left(n^d\right)$ time. 

\item It would be interesting to obtain a more elementary weak epsilon-net construction in all dimensions $d\geq 2$, whose cardinality is close to $1/\eps^{d-1/2}$ (or at least far smaller than $1/\eps^d$), and that can be implemented in $O^*(n)$ time.

\end{itemize}

\end{document}